\documentclass[12pt,reqno,titlepage]{amsart}

\usepackage[round]{natbib}

\usepackage{amsmath,amssymb,amsfonts}
\usepackage{pgffor}
\usepackage{url}       %
\usepackage[table]{xcolor}

\usepackage{color}
\definecolor{cadmiumgreen}{RGB}{33,114,13}
\definecolor{mypink2}{RGB}{219, 48, 122}
\usepackage[margin=1in]{geometry}
\usepackage{graphicx}
\usepackage{bm}
\usepackage{bbm}
\usepackage{multirow}
\usepackage{threeparttable}
\usepackage{lscape}
\usepackage{todonotes}
\usepackage{setspace}
\numberwithin{equation}{section}

\usepackage{subcaption}
\usepackage{soul} 

\usepackage{chngcntr}
\newtheorem{theorem}{Theorem}

\newtheorem{assumption}{Assumption}

\newtheorem{lemma}{Lemma}

\usepackage{booktabs}
\newcommand{\E}{\mathbb{E}}
\newcommand{\PR}{\mathbb{P}}
\DeclareMathOperator*{\plim}{plim}

\DeclareMathOperator*{\argmax}{argmax}
\makeatletter
\def\@xfootnote[#1]{%
	\protected@xdef\@thefnmark{#1}%
	\@footnotemark\@footnotetext}
\makeatother

\newcommand{\A}{\mathbf{A}}
\newcommand{\X}{\mathbf{X}}
\newcommand{\W}{\mathbf{W}}
\newcommand{\w}{\mathbf{w}}
\newcommand{\z}{\mathbf{z}}
\newcommand{\x}{\mathbf{x}}
\newcommand{\y}{\mathbf{y}}
\newcommand{\an}{\mathbf{a}_N}
\newcommand{\bn}{\mathbf{b}_N}
\newcommand{\Y}{\mathbf{Y}}
\newcommand{\Z}{\mathbf{Z}}
\newcommand{\G}{\mathbf{G}}
\newcommand{\D}{\mathbf{D}}

\newcommand{\I}{\mathbf{I}}
\newcommand{\PP}{\mathbf{P}}
\newcommand{\Q}{\mathbf{Q}}
\newcommand{\q}{\mathbf{q}}
\newcommand{\R}{\mathbf{R}}
\newcommand{\rr}{\mathbf{r}}
\newcommand{\M}{\mathbf{M}}
\newcommand{\Ss}{\mathbf{S}}

\newcommand{\B}{\mathbf{B}}
\newcommand{\HH}{\mathbf{H}}
\newcommand{\h}{\mathbf{h}}
\newcommand{\bupsilon}{\bm{\upsilon}}
\numberwithin{theorem}{section} 

\usepackage{wrapfig}

\usepackage{chngcntr}
\usepackage{apptools}
\AtAppendix{\numberwithin{theorem}{subsection}}
\AtAppendix{\numberwithin{equation}{subsection}}

\title[Estimation of Peer Effects in Endogenous Social Networks]
{\bf Estimation of Peer Effects in Endogenous Social Networks: Control Function Approach}
\thanks{We thank Bryan Graham and three referees for their helpful and valuable comments and suggestions. We are particularly grateful to one of the referees for suggesting the idea that is presented in Section \ref{sec: estimation with x,a as control} of the paper. We also appreciate the comments and discussions of the participants at the 2015 USC Dornsife INET Conference on Networks, the 2016 North American Summer Meeting of the Econometric Society, the 2016 California Econometrics Conference, the 2017 Asian Meeting of Econometric Society, the 2017 IAAE conference, the 2018 UCLA-USC Mini Conference, and the econometrics seminars at University of British Columbia and Ohio State University. The first draft of the paper was written while Johnsson was a graduate fellow of USC Dornsife INET and Moon was the associate director of USC Dornsife INET. Moon acknowledges that this work was supported by the Ministry of Education of the Republic of Korea and the National Research Foundation of Korea (NRF-2017S1A5A2A01023679). 
	\\ 
	Ida Johnsson:  ida.johnsson@clutter.com. Hyungsik Roger Moon: Corresponding author. moonr@usc.edu. }

\makeatletter
\renewcommand{\email}[2][]{%
	\ifx\emails\@empty\relax\else{\g@addto@macro\emails{,\space}}\fi%
	\@ifnotempty{#1}{\g@addto@macro\emails{\textrm{(#1)}\space}}%
	\g@addto@macro\emails{#2}%
}
\makeatother

\usepackage{amsaddr}
\author{Ida Johnsson}
\address{Clutter Inc}
\author{Hyungsik Roger Moon}
\address{Department of Economics, University of Southern California,\\
	and School of Economics, Yonsei University}
\date{\today}
\begin{document}

\begin{abstract}
We propose methods of estimating the linear-in-means model of peer effects in which the peer group, defined by a social network, is endogenous in the outcome equation for peer effects. Endogeneity is due to unobservable individual characteristics that influence both link formation in the network and the outcome of interest. We propose two estimators of the peer effect equation that control for the endogeneity of the social connections using a control function approach. 
 We leave the functional form of the control function unspecified and treat it as unknown. To estimate the model, we use a sieve semiparametric approach, and we establish asymptotics of the semiparametric estimator. 
 \\
  {\sc Keywords: peer effects, endogenous network, sieve estimation, control function} 
  \\
  {\sc JEL Classification: C14, C21}
\end{abstract}

\maketitle

\section{Introduction}
The ways in which interconnected individuals influence each other are usually referred to as peer effects. One of the first  to formally model peer effects is \cite{Manski1993a}. He proposes the linear-in-means model, in which an individual's action depends on the average action of other individuals and  possibly also on their average characteristics. \cite{Manski1993a} assumes that all individuals within a given group are connected. Later literature allows for more complex patterns of connections, in which an individual might be directly  influenced by a subset of the group. Examples are \cite{Bramoulle2009}, \cite{Lee2010}, \cite{Lee2007} among others. Models of peer effects have been applied in various areas, such as education, health and development. Examples of applications are found in recent review papers such as \cite{Blume2010}, \cite{Manski2000}, \cite{Epple2011}, \cite{brock2001interactions} and \cite{Graham2011}.

Many models considered in earlier literature assume that connections between individuals are independent of unobserved individual characteristics that influence  outcomes. However, assuming exogeneity of the network or peer group is restrictive in many applications. For example, 
consider the following widely studied empirical application of peer effects:  peer influence on scholarly achievement. The assumption that friendships are exogenous in the outcome equation for scholarly achievement means that there are no unobserved variables that influence both friendship formation and individual grades. However, even if a study controls for observable individual characteristics such as gender, age, race and parents' education, it is likely to omit factors that influence both students' choice of friends and their GPA; for example parental  expectations, psychological disorders, or non-reported substance use. For more examples of endogenous peer groups see \cite{brock2001interactions}, \cite{Weinberg2007},   \cite{Shalizi2012} and \cite{hsieh2016social}, among others. 

In this paper we propose a method for estimating a linear-in-means model of peer effects, where the peer group is defined by a network that is endogenous in the outcome equation. Our model allows for correlation between the unobserved individual heterogeneity that impacts network formation and the unobserved characteristics of the outcome. For this, we use a dyadic network formation model that allows the unobserved individual attributes of two different agents to influence link formation, and in which links are pairwise independent conditional on the observed and unobserved individual attributes.  The network formation we consider in the paper is dense \label{dense1} and nonparametric. 

The main contributions of the paper are  methodological. First, given the endogenous peer group formation, we show that we can identify the peer effects by controlling the unobserved individual heterogeneity of the network formation equation. Second, we propose an empirically tractable implementation of the control function, whose functional form is not parametrically specified. For this, we propose two approaches, one based on an estimator of the unobserved individual heterogeneity and the other one based on the average node degrees of the network.\footnote{We acknowledge that this approach is developed based on an idea provided by one of the referees. We thank  the referee.}  Our estimation method is semiparametric because we do not restrict the functional form of the control function. Finally, we derive the limiting distributions of the estimators within a large single network. The main challenge of the asymptotics is handling the strong dependence of observables caused by the dense network. 
Other peer effects papers that have considered endogenously formed peer groups and have controlled the endogeneity via  various control functions include \cite{GoldsmithP2013}, \cite{hsieh2016social}, \cite{Qu2015}, \cite{Arduini2015a} and \cite{Auerbach2016}. We provide more detail on these papers in Section \ref{sec:related.literature}.

The remainder of the paper is organized as follows. In Section \ref{section: main.idea} we present a high level description of our approach and provide intuition as to its empirical applications.
In Section \ref{section: model of peer effects with endogenous network} we formally present our model. In Section \ref{section: identification} we show {\color{black} how to identify  peer effects using control functions.} Estimation is discussed in Section \ref{section: estimation}, and in
Section \ref{subsection: estimation, limiting distributin of estimator} we discuss the limiting distribution of the estimator and propose standard errors. 
In Section \ref{section: monte carlo} we present results of Monte Carlo simulations. There we compare the finite sample performance of our two semiparametric estimators against an estimator that assumes unobserved characteristics enter in a linear way, as well as an instrumental variables (IV) estimator that does not control for network endogeneity. {\color{black} We investigate both high degree and low degree networks.}
Section \ref{section: conclusions} concludes.

A word on notation: in what follows we denote scalars by lowercase letters, vectors by lowercase bold letters, and matrices by uppercase bold  letters.
\section{Main Idea}\label{section: main.idea}

In this section we introduce a simple model in order to illustrate the main points of our approach. A more general model and detailed discussion of the model will follow later.   

\subsection{Simple Model}
A simple peer effect model for the purpose of illustration of the main idea is  
\begin{equation}\label{eq: outcome simplified}
y_i = \beta^0 \left( \frac{\sum_{j \neq i}d_{ij}x_j}{\sum_{j\neq i}d_{ij}}  \right) + v_i, \quad i = 1,...,N,
\end{equation}
where $x_i$ is a measure of observable characteristics of individual $i$ and $d_{ij}$ is an indicator of individual $i$'s peer, so $d_{ij}=1$ if $i$ and $j$ are directly linked and $0$ otherwise. In (\ref{eq: outcome simplified}), the regressor of interest is the average of the characteristics of those individuals who are linked with $i$, $\frac{\sum_{j \neq i}d_{ij}x_j}{\sum_{j\neq i}d_{ij}}$. For simplicity, we assume that $x_i$ is exogenous with respect to all the unobserved components of the model; this will be relaxed later.

For the link formation, we consider the following dyadic network formation model, 
\begin{equation}\label{eq: network simplified}
d_{ij} = \mathbb{I}(g(a_i,a_j)\geq u_{ij}) \mathbb{I}(i \neq j),
\end{equation}
where $a_i$ and $a_j$ are unobserved individual specific characteristics, $u_{ij}$ is a link specific component, and $g(\cdot,\cdot)$ is some function.  It should be noted that this model of network formation does not allow for network effects in link formation, as a link between $i$ and $j$ only depends on the characteristics of $i$ and $j$.  

The unobserved individual characteristic $a_i$ can be interpreted as social capital that increases the likelihood of forming a link. Depending on the context this could be factors like trustworthiness, socioeconomic status, or outspokenness. 

For example, \cite{fafchamps2011} measure the risk sharing links between households in Tanzania and they construct links between households based on the question whom individuals could ``personally rely on for help.'' \cite{fafchamps2007risk} examine the formation of risk-sharing networks using data from the rural Philippines. \cite{Banerjee2013} examine how participation in micro-finance diffuses through a social network which they measure using lending and trust.  In these settings, we can think of $a_i$ as a measure of individual trustworthiness and integrity in financial matters. 
\cite{Ductor2011} analyze whether knowledge of a researcher's co-authorship network is helpful in predicting his or her productivity. In this setting $a_i$ can be interpreted as some unobserved productivity trait that induces the researcher to have more coauthors, and also to be more productive at writing papers.  
\subsection{Control Function and Its Implementation}
The key feature of the peer effect model (\ref{eq: outcome simplified}) and (\ref{eq: network simplified}) is that individual $i$'s unobserved characteristic $a_i$, which impacts link formation, is correlated with $v_i$, $i$'s unobserved characteristic that affects the outcome $y_i$. For example,  $a_i$ could be an unobserved component that affects a researcher's publication rate $y_i$, and also his or her co-authorship relationships, $d_{ij}$. Alternatively, we can think of a situation where 	there are two
types of agents: popular and unpopular. The popular agents are 
more likely to be friends with other  agents, and popular agents have better
outcomes even in the absence of a peer effect.  Then the peer formation $d_{ij}$ becomes correlated with the unobserved component $v_i$ of the outcome, and, as a consequence, the regressor of the peer effect, $\frac{\sum_{j \neq i}d_{ij}x_j}{\sum_{j\neq i}d_{ij}}$, becomes endogenous. 

In this paper we use a control function method to handle the endogenous peer group problem. Let $\D_N$ be the $N \times N$ adjacency matrix that describes the network links $d_{ij}$. Suppose that the unobserved characteristics $(a_i,v_i)$ and $u_{ij}$ are randomly drawn over $i$ and $(i,j)$, respectively. Also assume that $u_{ij}$ is independent of $(a_i,v_i)$. Then, for any $i \neq j$, the link $d_{ij} = \mathbb{I}(g(a_i,a_j)\geq u_{ij})$ and $v_i$ are dependent only through $a_i$. Therefore, controlling for $a_i$, the network $\D_N$ and $v_i$ become mean independent, that is,
\begin{align*}
\E(v_i \,|\, \D_N,a_i) = \E(v_i \,|\, a_i) =: h(a_i).
\end{align*} 

Suppose that we observe $a_i$. Consider the outcome equation which controls for $a_i$ nonparametrically,  
\[
	y_i = \beta^0 \left( \frac{\sum_{j \neq i}d_{ij}x_j}{\sum_{j\neq i}d_{ij}}  \right) + h(a_i) + \varepsilon_i,
\]
where $\varepsilon_i := v_i - h(a_i)$. Once we control the endogeneity of the network with $a_i$, then the regressor of the peer effect becomes exogenous, and we can estimate the peer effect coefficient $\beta^0$ using the conventional partially linear regression estimation method (e.g. \cite{Robinson1988}).

However, in most empirical applications, $a_i$ is not observed. Then the question becomes how to implement the control function. In this paper, as the main methodological contribution, we propose the following two procedures. Both procedures are implemented with a single snapshot of an observed network.\label{remark-snapshot}
\begin{enumerate}
	\item[(i)] First, suppose that $a_i$ can be consistently estimated. An example can be found in \cite{Graham2017} with the specification $g(a_i,a_j) = a_i + a_j$. Then, we estimate $\beta^0$ by running the partially linear regression of $y_i$ on $ \frac{\sum_{j \neq i}d_{ij}x_j}{\sum_{j\neq i}d_{ij}}$ and  $h(\widehat{a}_i)$  as in \cite{Robinson1988}. 
	\item[(ii)] The second method is to use an observed control function that asymptotically carries the same information  as $a_i$. 
	For this, first notice by the WLLN, 
	\begin{align*}
	{\rm deg}_i 
	&:= \frac{1}{N} \sum_{j \neq i} d_{ij} = \frac{1}{N} \sum_{j \neq i} \mathbb{I}(g(a_i,a_j) \geq u_{ij}) \rightarrow_p \mathbb{P}(d_{ij} = 1 \,|\, a_i). 
	\end{align*}
	Suppose that the  network formation probability conditional on $a_i$, $\mathbb{P}(d_{ij} = 1 \,|\, a_i)$, is a monotonic function of $a_i$. A sufficient condition for this is that $g(\cdot, a_j)$ is monotonic in the same direction for all $a_j$, for example\label{remark-g-conditions}
	\begin{equation}
		g(a_i,a_j) = a_i + a_j - \tau |a_i-a_j| \label{ex.simple.one-to-one}
	\end{equation}
	with $ 0 \leq \tau < 1$.
	In this case, the limit of the average node degree, $ \lim_{N \rightarrow \infty} \frac{1}{N} \sum_{j \neq i} d_{ij}$, carries the same information as the control function $a_i$, which justifies ${\rm deg}_i$ as a proxy of the control function $a_i$, that is, $\E(v_i \,|\, a_i) \simeq \E(v_i \,|\, {\rm deg_i}) =: h_{*}({\rm deg}_i)$. The peer effect coefficient $\beta^0$ can be estimated by using ${\rm deg}_i$ as a control function. More specifically, we estimate $\beta^0$ by running the partially linear regression of $y_i$ on $\frac{\sum_{j \neq i}d_{ij}x_j}{\sum_{j\neq i}d_{ij}}$ and $h_{*}({\rm deg}_i)$.\\
	Intuitively, unobserved characteristics $a_i$ drive heterogeneous degree sequences. We can therefore control for degree when estimating peer effects, ignoring the specific choice of a structural model explaining heterogeneous degrees.
\end{enumerate} 
	The use of degree as a control function requires much fewer restrictions on the specification of the network. Intuitively, the unobserved node (or individual) fixed effects $a_i$ control for heterogeneous degree sequences. Therefore, from an economic point of view, what needs to be controlled is the agent's degree, which validates the control function approach that uses $\rm{deg_i}$. This approach does not require a specification of the specific structural model explaining heterogeneous degree sequences. \label{remark-control-structural-interpretation}
	Consistent estimation of $a_i$ usually  requires a specific functional form. For example, \cite{Graham2017} assumed an additive model and \cite{ChenFernandez-ValWeidner2018} require an interactive form.
However, there is a disadvantage in the degree approach. The degree approach cannot identify the coefficient of the observed exogenous regressor if the same regressor also impacts  the network formation. \label{remark-comparison-approaches-intro}

In Section \ref{section: model of peer effects with endogenous network}, we generalize the simple model (\ref{eq: outcome simplified}) by allowing for an additional peer effect, $\frac{\sum_{j \neq i}d_{ij}y_j}{\sum_{j\neq i}d_{ij}}$, known as the endogenous peer effect, which measures the effects of the outcomes of the peer group on an individual outcome. In this case we have to deal with two kinds of endogeneity in the peer effect regressors: one from the endogenous regressors $y_{j}$ and the other one from the endogenous peers $d_{ij}$. 
In Section \ref{section: model of peer effects with endogenous network}, we also generalize the dyadic network formation model by introducing a dyadic component based on observed individual characteristics. We provide application examples of the general model and discuss its features there.
The identification of the peer effects in the general model will be discussed in Section \ref{section: identification}. In Section \ref{section: estimation} we shows how to implement the two aforementioned estimation methods in the general framework.
In the appendix we provide the regularity conditions that are required for the asymptotic results of the paper. All the technical proofs and comprehensive Monte Carlo simulation results are found in the Online Supplement material which is available in \cite{JohnssonMoon2019}.

\subsection{Related Literature} \label{sec:related.literature}

Closely related papers that adopt a control function approach include \cite{GoldsmithP2013}, \cite{hsieh2016social}, \cite{Qu2015}, \cite{Arduini2015a} and \cite{Auerbach2016}. Our paper adopts a frequentist approach based on a nonparametric specification of the network formation, while \cite{GoldsmithP2013} and \cite{hsieh2016social} use the Bayesian method based on a full parametric specification of the network formation and the outcome equation. Like our paper, \cite{Qu2015} assume the network (spatial weights in their model) to be endogenous through  unobserved individual heterogeneity. However, our paper is different from \cite{Qu2015} in many aspects. They consider sparse network formation models while we  consider a dense network. They restrict the functional form of the control function to be linear, while we impose no restriction on the functional form. The two papers propose different implementations of the control function. Also, in \cite{GoldsmithP2013}, unobserved components account for homophily in link formation, whereas in our setup they mainly drive degree heterogeneity but are allowed to account for homophily as well, as in the example (\ref{ex.simple.one-to-one}).

Our paper is different from \cite{Arduini2015a} regarding the main source of the endogeneity of the network and the form of the control function. \cite{Arduini2015a} assume that the endogeneity of the network is allowed through dependence between the outcome equation error and the idiosyncratic network formation error, like the conventional sample selection model. This model can be interpreted as meeting opportunities being correlated with unobserved ability of the agent that affects the outcome.
\cite{Arduini2015a} consider control functions (both parametric and semiparametric) to deal with the selection bias problem and propose a semiparametric estimator that uses a power series to approximate selectivity bias terms. Regarding asymptotics, in both \cite{Qu2015} and \cite{Arduini2015a}, the asymptotics are derived using near-epoch dependence and are based on the assumption that the number of connections does not increase at the same rate as the square of the network size. 

\label{reference-Auerbach}Among the aforementioned related papers, probably the one most closely related to ours is  \cite{Auerbach2016}. As a result, we would like to discuss the differences between the two papers in more detail.
The  outcome model of \cite{Auerbach2016} is a partially linear regression model where the nonparametric component is an unknown function of the unobserved network heterogeneity,
\begin{align*}
y_i &= \beta^0 x_i + h(a_i) + \varepsilon_i, \\
d_{ij} &= \mathbb{I}(g(a_i,a_j)\geq u_{ij}) \mathbb{I}(i \neq j).
\end{align*}
In the simple peer effect example, the exogenous peer effect corresponds to the regressor $x_i$ above. The network formation is the same as (\ref{eq: network simplified}). 

To compare the identification ideas, let's assume that $a_i \sim U[-1/2,1/2]$ and $u_{ij} \sim U[0,1]$. In this case, $d_i := (d_{i1},...,d_{in})'$ and the distribution of $d_i$ of node $i$, whose characteristic is $a_i$, is fully characterized by the link formation probability profile $g(a_i, \bullet)$. 

The key condition of \cite{Auerbach2016} is that $h(a_i)$ and the the link formation distribution profile $g_i(\bullet) := g(a_i,\bullet)$ be one-to-one a.s., that is, $g(a,\bullet) \neq g(a^*,\bullet)$ a.s. if and only if $h(a) \neq h(a^*)$. Then, for any distance measure between the two profiles $g_i$ and $g_j$, $d(g_i,g_j)$, it follows that $d(g_i,g_j) = 0$ if and only if  $h(a_i) = h(a_j)$. 

Based on this, \cite{Auerbach2016} finds that one can control the network endogeneity by pair-wise differencing\footnote{This resembles \cite{Powell1987}, \cite{Heckmanetal1998}, and \cite{AbadieImbens2006}.}  of the observations of the two individuals, $i$ and $j$, whose network formation distributions are the same, $d(g_i,g_j) = 0$, and proposes a semiparametric estimator based on matching pairs of agents
with similar columns of the squared adjacency matrix.

Notice that the identification condition of \cite{Auerbach2016}  is satisfied if $g(a_i,\bullet)$ and $a_i$ have a one-to-one relation. However, our second identification is based on the condition that $a_i$ and the marginal network probability, $\int g(a_i,\tau) d \tau$, have a one-to-one relation. We admit that this condition is more restrictive than the identification condition of \cite{Auerbach2016}, because our restriction is a special case of his restriction. However, as mentioned in the introduction, our identification under the stronger condition allows for the omitted variable in the peer effects equation to be nonparametrically directly
estimated, which results in the peer effect estimator having the parametric convergence rate ($\sqrt{N}$). This feature is not necessarily guaranteed in the framework of \cite{Auerbach2016}.\footnote{We thank one of the referees for suggesting the comparisons.}

\section{General Model of Peer Effects with an Endogenous Network}\label{section: model of peer effects with endogenous network}

In this section, we introduce a general linear-in-means peer effect model that extends the simple illustrative outcome model with a peer effect in (\ref{eq: outcome simplified}) and the simple dyadic network formation model in (\ref{eq: network simplified}).

\subsection{General Linear-In-Means Peer Effects Model}

As in Section \ref{section: main.idea}, $d_{ij}$ are the observed binary variables that measure undirected links among individuals $i\in \{1,2,\ldots,N\}$.  
We assume that individual outcomes are given by the linear-in-means model of peer effects
\begin{equation}\label{model:outcome}
y_i = \left( \sum_{j=1 \atop j\neq i}^Ng_{ij}y_j \right) \beta_1^0 + \x'_{1i}\beta_2^0 + \left(\sum_{j=1\atop j\neq i}^Ng_{ij}\x_{1j}\right)^\prime\beta_3^0+\upsilon_i,
\end{equation}
where $\x_{1i}$ are observed individual characteristics that affect the outcome $y_i$, $v_i$ are unobserved individual characteristics, and 
\[
g_{ij} = 
\left\{ 
\begin{array}{cc}
0 & \quad {\rm  if } \quad i=j  \\ 
\frac{d_{ij}} {\sum_{j \neq i} d_{ij}} & {\rm otherwise} 
\end{array} 
\right.
\] 
is the weight of the peer effects. 
Using the terminology of \cite{Manski1993a}, $\beta_1^0$ captures the endogenous social effect, and $\beta_3^0$ measures the exogenous social effect.
We let $\beta^0 := (\beta^0_1, \beta_2^{0'}, \beta_3^{0'})'$ and denote $\beta = (\beta_1, \beta_2^{'}, \beta_3^{'})'$.

We let $\D_N$  be the $ (N \times N)$ adjacency matrix of the network whose $(i,j)^{th}$ element is $d_{ij}$. We let $d_{ii}=0$ for all $i$, following convention. Let $\G_N$ be the matrix whose $(i,j)^{th}$ element is $g_{ij}$. Recall that $\G_N$ is obtained by row-normalizing $\D_N$. 
Denote $\X_{1N}=(\x_{11}',\ldots,\x_{1N}')'$, $\y_N=(y_1,\ldots,y_N)'$ and $\bupsilon_N=(\upsilon_1,\ldots,\upsilon_N)'$.\label{remark-equilibirum}
{\color{black} Using this notation, we can express the linear-in-means peer effects model (\ref{model:outcome}) as
\begin{equation}
	\y_N = \G_N \y_N \beta_1^0  + \X_{1N} \beta_2^0 + \G_N \X_{1N} \beta_3^0 + \bupsilon_N. \label{model.outcome.matrix}
\end{equation}
Throughout the paper, we assume that $| \beta_1^0 | < 1$. 
It is known that when $\G_N$ is row normalized (i.e., $\sum_{j \neq i}g_{ij} = 1$) and $ | \beta_1^0 | < 1$, the (equilibrium) solution of the peer effect model uniquely exists (e.g., see \cite{Bramoulle2009}) as 
\begin{align}
\y_N 
&= (\I_N-\beta_1^0 \G_N)^{-1}(\X_{1N}\beta_2^0 + \G_N \X_{1N}\beta_3^0 + \bupsilon_N) \nonumber
\\
&= \sum_{k=0}^{\infty} \left( \beta_1^0 \G_N \right)^k(\X_{1N} \beta_2^0 + \G_N\X_{1N} \beta_3^0 + \bupsilon_N). \label{model.outcome.reduced.form}
\end{align}
}
In the standard linear-in-means model of peer effects, the main focus has been identification and estimation of peer effects, assuming that the peer group (or the network) is exogenous, that is, $\E[\upsilon_i|\X_{1N},\G_N]=0$.  For example, see \cite{Manski1993a} and \cite{Bramoulle2009}, \cite{Lee2007}, and \cite{blume2015linear}.  
To identify and estimate the linear-in-means model of peer effects when the peer group is exogenous, it is necessary to take into account the fact that the regressor  $\sum_{i=1}^N g_{ij}y_{j}$ is correlated with the error term $\upsilon_i$. For example, if $\upsilon_i\sim\ i.i.d. (0,\sigma^2)$,
it is true that
	\begin{equation}
	\begin{split}
	\E[(\G_N\y_N)'\bupsilon_N]&=[(\G_N(\I_N-\beta_1^0 \G_N)^{-1}(\X_{1N}\beta_2^0+\G_N\X_{1N}\beta_3^0+\bupsilon_N))'\bupsilon_N]\\
	&=\E[(\G_N(\I_N-\beta_1^0 \G_N)^{-1}\bupsilon_N)'\bupsilon_N]=\sigma_0 tr(\G_N(\I_N-\beta_1^0 \G_N)^{-1})\neq 0.
	\end{split}
	\end{equation}
	To solve this endogeneity problem different estimators have been proposed in the literature, see for example \citet{Kelejian1998}, \cite{Lee2003} and \cite{Lee2007a}.  One of the widely used estimation methods is the Instrumental Variables (IV) approach. 
\label{remark-IV}
	 {\color{black} In view of the expression of (\ref{model.outcome.reduced.form}), when $\beta_2^0 \neq 0$, we can use  $\G^2_N\X_{1N}$ as the IV of the endogenous regressor $\G_N\y_N$ because $\G^2_N\X_{1N}$ is uncorrelated with $\bupsilon_N$ while it is correlated with the endogenous regressor $\G_N \y_N$  (see for example  \cite{Kelejian1998}, \cite{Lee2003}, and \cite{Bramoulle2009})\footnote{ {\color{black} If $\beta_2^0 = 0$, $\y_N$ does not depend on $\X_{1N}$ and $\G_N^2\X_{1N}$ is not a relevant instrument for $\G_N \y_N$.  }}. }Then, the natural estimator is the Two-Stage Least Squares (2SLS) estimator,
	\begin{equation}
	\widehat{\beta}_N^{2SLS}=(\W_N'\Z_N(\Z_N'\Z_N)^
	{-1}\Z_N\W_N)^{-1}\W_N'
	\Z_N(\Z_N'\Z_N)^{-1}\Z_N'\y_N,
	\end{equation} 
	where $\W_N=[\G_N\y_N,\ \X_{1N},\ \G_N\X_{1N}]$ and $\Z_N=[\X_{1N},\ \G_N\X_{1N},\ \G^2_N\X_{1N}]$ is the matrix of instruments. For the IVs $\Z_N$ to be strong, we assume that $\beta_2^0 \neq 0$. 

	When the network matrix is endogenous, $\E[\G_N\bupsilon_N]\neq 0$, and the procedure used by  \cite{Kelejian1998}, \cite{Lee2003}, \cite{Bramoulle2009} and others is no longer valid since the IV matrix $\Z_N=[\X_{1N},\ \G_N\X_{1N},\ \G^2_N\X_{1N}]$ is correlated with the error term $\bupsilon_N$.  
Specifically, the validity of the 2SLS estimator depends on the orthogonality condition $\E[\bupsilon_N|\Z_N] = 0$, which is implied if $\E[\bupsilon_N|\X_{1N},\G_N] = 0$. However, it does not hold if the (row normalized) network $\G_N$  is correlated with $\bupsilon_N$, which is true if unobserved individual characteristics of $\G_N$ directly influence both link formation and individual outcomes.

In this paper, we consider the case where it may be that $\E[\bupsilon_N|\X_{1N},\G_N]\neq 0$, so that  unobserved characteristics that influence link formation  can also have a direct effect on individual outcomes. This is an important consideration in many common applications, like the impact of school friendships on scholarly achievement or substance use.	Imagine kids from homes where parents help with homework who only form friendships with kids from similar homes. If this unobserved characteristic of parental behavior is not taken into account, and if this is what really determines grades, this effect might falsely be classified as a peer effect.\label{remark-scholarly} {\color{black} A more elaborate discussion of our framework and its empirical applications can be found in Section \ref{section: main.idea}.}
\subsection{Model of Network Formation} \label{sec: model of network formation}
Let $\x_{2i}$ be a vector of observable characteristics of individual $i$, and let $\x_i=\x_{1i}\cup \x_{2i}$. 
Define $\X_{2N}$ analogously to $\X_{1N}$ and let $\X_N=\X_{1N}\cup\X_{2N}$.
We introduce $a_i$, a scalar unobserved characteristic of individual $i$, which is treated as an individual fixed effect, and hence,  might be correlated with $\x_i$.  We denote the vector of individual unobserved characteristics by $\an=(a_1,a_2,\ldots,a_N)'$. Individuals are connected by an undirected network $\D_N$, with the $(i,j)^{th}$ element $d_{ij} = 1$ if $i$ and $j$ are directly connected and $0$ otherwise. We assume the network to be undirected\footnote{{\color{black}Our analysis can be extended to the directed network case, but we do not pursue it in this paper.}},  $d_{ij} = d_{ji}$, and assume $d_{ii}=0$ for all $i$, following the convention. In this case, there are $n=\binom{N}{2}$ dyads. Let $\mathbf{t}_{ij}$ denote an $l_T\times 1$ vector of dyad-specific characteristics of dyad $ij$, and we assume that $\mathbf{t}_{ij}=t(\x_{2i},\x_{2j})$. 
Agents form links according to 
\begin{equation}
d_{ij}=\mathbb{I}( g( t(\x_{2i},\x_{2j}), a_i, a_j ) - u_{ij} \geq 0), \label{model.network.formation}
\end{equation}
where  $\mathbb{I}( \bullet)$ is an indicator function. In this setup, link surplus is transferable across directly linked agents and consists of three components: $\mathbf{t}_{ij} := t(\x_{2i},\x_{2j}) $ is a systematic component that varies with observed dyad attributes and accounts for homophily, $a_i$ and $a_j$ account for unobserved dyad attributes (degree heterogeneity), and $u_{ij}$ is an idiosyncratic shock that is i.i.d. across dyads and independent of $\mathbf{t}_{ij}$ and $a_i$ for all $i,j$. 
Since links are undirected, the surplus of link $d_{ij}$ must be the same for individual $i$ and $j$. Hence, we assume that the function $t_{ij}$ is symmetric in $i$ and $j$, and the function $g$ is symmetric in $a_i$ and $a_j$. 

In the literature, various parametric versions of the network formation in (\ref{model.network.formation}) are used, ({\color{black} see for example \cite{jackson2005survey}, \cite{Graham2017})}). An important example of a parametric specification is the one in \cite{Graham2017},
\begin{equation}
d_{ij}=\mathbb{I}(t(\x_{2i},\x_{2j})'\lambda+a_i+a_j - u_{ij} > 0). \label{model.network.formation.parametric}
\end{equation} 
For the purpose of the paper, particularly in constructing the estimators that we  introduce in Section \ref{section: estimation}, we do not need a parametric specification.

Regarding the network formation (\ref{model.network.formation}), we impose restrictions (Assumption \ref{assumption:limit.dist} (iii) - (vi) in the Appendix) that imply the following two features. The first feature is that the link formation probability of individual $i$ with characteristics $(\x_{2i},a_i)$ is one-to-one with respect to the unobserved characteristic $a_i$, that is, for all $x_{2i}$,
\begin{equation}
a_i \neq a_i^* \text{ if and only if } \PR \left( d_{ij} = 1 \,|\, \x_{2i},a_i \right) \neq \PR \left( d_{ij} = 1 \,|\, \x_{2i},a_i^* \right). \label{eq.monotone.link.formation}
\end{equation}
Obviously, this condition is satisfied in the parametric model (\ref{model.network.formation.parametric}). This monotonic condition justifies the use of the average node degree in implementing the control function as introduced in Section \ref{section: main.idea} and will be discussed in Section \ref{sec: estimation with x,a as control}. The second feature is that the network formed by (\ref{model.network.formation}) is dense \label{remark-dense} in the sense that the expected number of connections is proportional to the square of the network size. This is satisfied if the error $u_{ij}$ is drawn randomly from a distribution with  full support, while $g( \mathbf{t}_{ij}, a_i, a_j )$  is bounded (see Assumption \ref{assumption:limit.dist} (iii),(iv), and (v) in the Appendix). In this case, the probability of any two individuals forming a link is bounded away from zero and strictly less than one. The dense network model is appropriate for scenarios where any two individuals can plausibly form a link.
Notice that the dense network assumption and the sharing restriction on the net surplus function $g$ are necessary for implementing the control function in Section \ref{section: estimation} and establishing the asymptotic theory of the control function based estimators in Section \ref{subsection: estimation, limiting distributin of estimator}. 
If $a_i$ is observed, we can identify and estimate peer effects without these assumptions (see Section \ref{section: identification}).

Regarding the network formation model (\ref{model.network.formation}), it is important to note that the network formation model (\ref{model.network.formation}) rules out interdependent link preferences, and it assumes that links are formed independently conditional on observed individual characteristics and unobserved fixed effects. \label{no network externalities}\label{remark-conditional-independence} As discussed in \cite{Graham2017}, this assumption is appropriate for settings where link formation is driven predominantly by bilateral concerns, such as certain types of friendship networks, trade networks and some models of conflict between nation-states. The model in (\ref{model.network.formation}) is not a good choice when important strategic aspects influence link formation, like when the identity of the nodes to which $j$ is linked influences $i$'s return from forming a link with $j$. A discussion of networks with interdependent links can be found in \cite{Graham2017} and \cite{dePaula2016}. Also, when network externalities are present, the additional complication of multiple equilibria has to be considered, see for example \cite{Sheng2012} for more details.

\section{Identification of peer effects using a control function approach}\label{section: identification}

In this section we provide an identification argument for the peer effect equation based on a control function when the network is endogenous.

\subsection{Control Function of Network Endogeneity}
In this subsection we discuss how to control the endogeneity of the peer group defined by the network formed in equation (\ref{model.network.formation}). 
First we introduce a basic assumption that we will maintain throughout the paper.
\begin{assumption}[]\label{as:basic}
		(i) $(\x_i,a_i,\upsilon_i)$ are i.i.d. for all $i$, $i=1,\ldots,N$, 
		(ii) $\{u_{ij}\}_{i,j=1,\ldots,N}$ are independent of $(\X_{N},\mathbf{a}_N, \bupsilon_N )$ and i.i.d. across $(i,j)$ with cdf $\Phi(\cdot)$, and 
		(iii) $\mathbb{E}(v_i|\x_i,a_i) = \mathbb{E}(v_i|a_i).$  	
\end{assumption}
Assumption \ref{as:basic}(i) implies that the observables $\x_i$ and the unobservable characteristics $(a_i,\upsilon_i)$ are randomly drawn. This is a standard assumption in the peer effects literature. Assumption \ref{as:basic}(ii) assumes that the link formation error $u_{ij}$ is orthogonal to all other observables and unobservables in the model. This means that the dyad-specific unobservable shock $u_{ij}$ from the link formation process does not influence  outcomes $(y_1,\ldots,y_N)'$.  However, we allow for endogeneity of the social interaction group through dependence between the two unobserved components $a_i$ and $\upsilon_i$. This means that the unobserved error $\upsilon_i$ in the outcome equation can be correlated with  unobserved individual characteristics $a_i$ that are determinants of link formation. We also allow the observed characteristics $\x_i$ of the outcome equation and the network formation to  be correlated with the unobserved components $(\upsilon_i,a_i)$, so that the regressor $\x_{1i}$ can be endogenous in the outcome equation, and the network formation observables $\x_{2i}$ can be arbitrarily correlated with the unobserved individual characteristic $a_i$. In Assumption \ref{as:basic}(iii), we assume that the dependence between $\x_i$ and $\upsilon_i$ exists only through $a_i$. That is, $a_i$ is the fixed effect of individual $i$ and controls the endogeneity of $\x_i$ with respect to $\upsilon_i$.

Notice that the network $\D_N$ defined in (\ref{model.network.formation}) and the (row normalized) network $\G_N$ are measurable functions of  $ (\x_{2i},\x_{2,-i},a_i,\mathbf{a}_{-i},\{u_{ij}\}_{i,j=1,\ldots,N}),$ where $\x_{2,-i}=(\x_{2,1},\ldots,\x_{2,i-1},\x_{2,i+1},\ldots,\x_{2,N})$ and $\mathbf{a}_{-i}$ is defined analogously. 
Under Assumption \ref{as:basic} we have
\begin{eqnarray}
\E[ \upsilon_i|\X_{N},\G_N,a_i ] &=& \E[\upsilon_i|\x_{-i},\G_N(\x_{2,-i},\mathbf{a}_{-i},\{u_{ij}\}_{i,j=1,\ldots,N},\x_{2i},a_i),\x_{i},a_i] \nonumber \\
&=&\E[\upsilon_i|\x_{i},a_i] \nonumber 
= \E[\upsilon_i| a_i],
\label{eq.control.function}
\end{eqnarray}
where the second equality holds because $(\x_{-i},\mathbf{a}_{-i},\{u_{ij}\}_{i,j=1,\ldots,N})$ and $ (\x_{i},a_i,\upsilon_i)$ are independent under Assumptions \ref{as:basic} (i) and (ii). This shows $v_i$ and $(\x_{-i}, \G_N(\x_{2,-i},\mathbf{a}_{-i},\{u_{ij}\}_{i,j=1,\ldots,N},\x_{2i},a_i))$ are mean-independent conditioning on $(\x_{i},a_i)$. The last line follows by the fixed effect assumption, Assumption \ref{as:basic} (iii).

Result (\ref{eq.control.function}) shows that conditional on the unobserved heterogeneity $a_i$ in the network formation (and any subcomponents of $\x_i$), the unobserved characteristic $\upsilon_i$ that affects the outcome $y_i$ becomes uncorrelated with the (row normalized) network $\G_N$ (and the observables $\X_N$). This implies that the network endogeneity can be controlled by $a_i$ (or together with any subcomponents of $\x_i$). 
We summarize the discussion above in the following lemma:  

\begin{lemma}[Control Function of Peer Group Endogeneity]
\label{lemma:control.function} 
Suppose that Assumption \ref{as:basic} holds. Then, $\E[\upsilon_i|\X_{N},\G_N,a_i]=\E[\upsilon_i| \x_{i}, a_i].$
\end{lemma}

\subsection{Identification of Peer Effects with $a_i$ as Control Function}

In this section we show how to identify the peer effects in the outcome question when the endogenous network is formed by (\ref{model.network.formation}). We provide two identification methods depending on whether we control the network (peer group) endogeneity with $a_i$ or $a_i$ together with $\x_{2i}$, in the case when $\x_{2i}$ and $\x_{1i}$ do not overlap.

First notice that regardless of the possible endogeneity of the (row normalized) network $\G_N$, we need to control for the endogeneity of the term $\sum_{j \neq i} g_{ij}y_j$ that represents the so-called endogenous peer effects. When the peer group $\G_N$ is exogenous and uncorrelated with $\upsilon_N$, $\G^2_N \X_{1N}$ is often used as an IV for the endogenous peer effects term $\G_N \y_N$ (See, for example, \cite{Kelejian1998}, \cite{Lee2003}, \cite{Bramoulle2009}.).

Let $\Z_N=[\X_{1N}, \G_N \X_{1N}, \G^2_N\X_{1N} ]$ be the usual IV matrix used in 2SLS estimation of the peer effects equation. 
Note that $\Z_N$ is not a valid IV matrix anymore in our framework because the peer group defined by the network $\G_N$ is correlated with $\upsilon_N$ due to potential correlation between the unobserved $\upsilon_i$ and $a_i$.
Let $\W_N=[\G_N \y_N, \X_{1N}, \G_N \X_{1N}]$. Further, denote the transpose of the $i$th row of $\Z_N$ and $\W_N$ by $\z_i$ and $\w_i$, respectively. 
 
Suppose that Assumption \ref{as:basic} holds and so $a_i$ controls the network endogeneity. Then, 
\begin{eqnarray}
	\E\left[\: \left( \z_i -\E[\z_i|a_i]\right)(\upsilon_i - \E(\upsilon_i|a_i))\: |\:a_i \right]
	&=& \E[\z_i\upsilon_i \: | \: a_i] - \E[\z_i \: | \: a_i]\E[\upsilon_i \: | \: a_i] \nonumber \\
	&=&\E\left[\E[\z_i\upsilon_i \: | \: a_i,\X_{1N},\G_N] \: | \: a_i\right] - \E[\z_i|a_i]\E[\upsilon_i \: | \: a_i] \nonumber  \\
	&=&\E\left[\z_i\E[\upsilon_i \: | \: a_i,\X_{1N},\G_N] \: | \: a_i\right] - \E[\z_i \:| \: a_i]\E[\upsilon_i \:| \: a_i] \nonumber  \\
	&\stackrel{(1)}{=}& \E\left[\z_i\E[\upsilon_i \:| \: a_i] \:|\: a_i\right] - \E[\z_i \:|\: a_i]\E[\upsilon_i \:|\: a_i] \nonumber  \\
	&=& 0,\label{eq.orthogonality}
\end{eqnarray}
where  equality $(1)$ holds by Lemma \ref{lemma:control.function}(a). 
This shows that the instrumental variables $\z_{i}$ or  $\z_{i} -\E[\z_{i}|a_i]$ become orthogonal to $ \upsilon_i-\E[\upsilon_i|a_i],$ the residual of $\upsilon_i$ after projecting out $a_i$.  

Furthermore, if $\E\left[  \left( \z_i -\E[\z_i|a_i] \right)
\left( \w_i - \E[\w_i|a_i]\right)^{\prime}\right]$ 	has full rank, then we can identify the peer effect coefficients $\beta^0$ as
\begin{eqnarray*}
0 &=& \E\left[ \left( \z_i -\E[\z_i|a_i]\right)
\left(y_i-\w'_i\beta-\E[y_i - \w'_i \beta |a_i] \right) \right] \\
&=& \E[\left( \z_i -\E[\z_i|a_i]\right)(\w_i-\E[\w_i|a_i])'](\beta-\beta^0)+\E[\left( \z_i -\E[\z_i|a_i]\right)(\upsilon_i-\E[\upsilon_i|a_i])]\\
&\stackrel{(1)}{=}&\E[\left( \z_i -\E[\z_i|a_i]\right)(\w_i-\E[\w_i|a_i])'](\beta-\beta^0) \\
& \stackrel{(2)} {\Leftrightarrow} & \beta=\beta^0,
\end{eqnarray*}
where  equality $(1)$ follows by the orthogonality result in (\ref{eq.orthogonality}) and  equality $(2)$ follows from the full rank condition. 
\begin{assumption}[Rank condition]\label{assumption: rank}
	$\E\left[  \left( \z_i -\E[\z_i|a_i] \right)
	\left( \w_i - \E[\w_i|a_i]\right)^{\prime}\right]$ 	has full rank.
\end{assumption}
For the full rank condition in Assumption \ref{assumption: rank}, it is necessary that the IVs $\z_i$ and the regressors $\w_i$ have additional variation after projecting out the control function $a_i$. As shown in the Supplementary Appendix \ref{appendix: distribution of est}, when $N$ is large, both $\z_i$ and $\w_i$ become close to functions that depend only on $(\x_i,a_i)$. In this case, for the full rank condition to be satisfied, it is necessary that there be additional random components in $\x_{i}$ that are different from $a_i$, so that the limits of $\z_i$ and $\w_i$ are not linearly dependent. \label{rank - discussion}
As a summary, we have the following first identification theorem.
\begin{theorem}[Identification]Under Assumptions  \ref{as:basic} and \ref{assumption: rank}, the parameter $\beta^0$ is identified by the moment condition $\E[\left( \z_i -\E (\z_i|a_i) \right)(y_i - \E (y_i|a_i) - (\w_i-\E(\w_i|a_i))'\beta^0)]=0$:
\[
\E[\left( \z_i -\E (\z_i|a_i) \right)(y_i - \E (y_i|a_i) - (\w_i-\E(\w_i|a_i))'\beta)]=0\ \iff\ \beta=\beta^0.
\] \label{theorem:identification}
\end{theorem}
Theorem \ref{theorem:identification} shows that we can identify the parameter $\beta^0$ by controlling the unobserved network heterogeneity $a_i$ in the outcome equation and taking the residuals $y_i - \E (y_i|a_i) - (\w_i-\E(\w_i|a_i))'\beta$ and using the instrumental variables $\z_i - \mathbb{E}[\z_i|a_i]$. 
\subsection{Identification of Peer Effects using $(\x_{2i},a_i)$ as Control Function} \label{section: alternative identification}
In view of the derivation of the control function in (\ref{eq.control.function}) under Assumption \ref{as:basic}, it is possible to use any regressors in $\x_i$ in addition to the unobserved heterogeneity $a_i$. 
In this section, we discuss identification of the peer effects using $(\x_{2i},a_i)$ as control function. The reason to consider this particular control function is that we can {\color{black} implement it in the absence of a consistent estimator of $a_i$}, which will be discussed in detail in Section \ref{section: estimation}.

First, suppose that there is no overlap between the regressors in the outcome equation $\x_{1i}$ and the regressors in the network formation equation $\x_{2i}$ and assume the conditions in Assumption \ref{as:basic}.\footnote{Later in this section, we will discuss a more general case where $\x_{1i}$ and $\x_{2i}$  intersect.}
\begin{assumption}[]\label{as:basic.alternative}
	Assume that the conditions (i),(ii), and (iii) of Assumption \ref{as:basic} hold. Also, assume that (iv) the explanatory variables in $\x_{1i}$ and $\x_{2i}$ do not overlap (i.e., $\x_{1i} \, \cap \, \x_{2i} = \emptyset$).	
\end{assumption}

Then, under Assumption \ref{as:basic} and by (\ref{eq.control.function}), it follows that 
\begin{equation}
\E[ \upsilon_i|\X_{N},\G_N,a_i ] 
= \E[\upsilon_i|a_i] 
= \E[\upsilon_i|\x_{2i},a_i], \label{eq.control.function.alternative}
\end{equation}  
where the last line holds by Assumption \ref{as:basic}(iii). Then, similar to (\ref{eq.orthogonality}), we can show that 
\begin{eqnarray}
&& \E\left[\: \left( \z_i -\E[\z_i|\x_{2i},a_i]\right)(\upsilon_i - \E(\upsilon_i|\x_{2i},a_i))\: |\: \x_{2i}, a_i \right] 
= 0. \label{eq.orthogonality.alternative}
\end{eqnarray}
Furthermore, suppose that the following full rank assumption is satisfied:
\begin{assumption}[Rank condition]\label{assumption:rank.alternative}
	$\E\left[  \left( \z_i -\E[\z_i|\x_{2i},a_i] \right)
	\left( \w_i - \E[\w_i|\x_{2i},a_i]\right)^{\prime}\right]$ 	has full rank.
\end{assumption}
Notice that if $\x_{1i}$ and $\x_{2i}$ are overlapped, then the full rank condition in Assumption \ref{assumption:rank.alternative} does not hold. 

Using similar arguments that lead to Theorem \ref{theorem:identification}, we can identify the peer effect coefficients $\beta^0$ as
\begin{eqnarray}
0 &=& \E\left[ \left( \z_i -\E[\z_i|\x_{2i},a_i]\right)
\left(y_i-\w'_i\beta-\E[y_i - \w'_i \beta |\x_{2i},a_i] \right) \right] 
\Leftrightarrow  \beta=\beta^0, \label{eq.identification.alternative}
\end{eqnarray}
This is summarized in the following theorem.

\begin{theorem}[Alternative Identification]Under Assumptions \ref{as:basic}, \ref{as:basic.alternative}, and \ref{assumption:rank.alternative}, the parameter $\beta^0$ is identified by the moment condition
	\[
	\E[\left( \z_i -\E( \z_i|\x_{2i},a_i )\right)(( y_i - \E(y_i|\x_{2i},a_i) -  ( \w'_i -\E (\w_i |\x_{2i},a_i))'\beta]=0\ \iff\ \beta=\beta^0.
	\] \label{theorem:identification.alternative}
\end{theorem}
So far, we have considered the case where the regressors $\x_{i1}$ and $\x_{2i}$ do not intersect. 
A more general case is when the regressors $\x_{1i}$ consist of two components, where one component is different from the observed control function $\x_{2i}$ and the other is part of $\x_{2i}$. That is, $\x_{1i} = (\x_{11i},\x_{12i})$, where $\x_{11i}$ does not share any elements with $\x_{2i}$ and $\x_{11i}$ is nonempty, and $\x_{12i} \subset \x_{2i}$.  Let $\beta^0_2 = (\beta^0_{21},\beta^0_{22}), \beta^0_3 = (\beta^0_{31},\beta^0_{32})$ conformable to the dimensions of $(\x_{11i},\x_{12i})$. Similarly let $\beta_2 = (\beta_{21},\beta_{22}), \beta_3 = (\beta_{31},\beta_{32}).$

In this case, with a properly modified rank condition of $\z_{(2),i}$ and $\w_{(2),i}$ which excludes the variables associated with $\x_{12,i}$ and $\sum_{j=1,\neq i}^N g_{ij} \x_{12,j}$, we can identify the coefficients $\beta^0_{(2)} := (\beta^0_1,\beta^0_{21},\beta^0_{31})$ using the same argument that leads to the identification in (\ref{eq.identification.alternative}).
However, we cannot identify the coefficients that correspond to the variable $\x_{12,i}$ and $\sum_{j=1,\neq i}^N g_{ij} \x_{12,j}$. The reason is that  controlling the network endogeneity with the control variable $(\x_{2i},a_i)$ wipes out the information in $(\x_{12,i},\sum_{j=1,\neq i}^N g_{ij} \x_{12,j})$:
\begin{align*}
\x_{12,i} - \E[ \x_{12,i} | \x_{2i},a_i] &= 0 \\
\sum_{j=1,\neq i}^N g_{ij} \x_{12,j} -  \E\left[\sum_{j=1,\neq i}^N g_{ij} \x_{12,j} \left| \x_{2i},a_i \right. \right] &\rightarrow_p 0,
\end{align*}
where the second convergence holds because $\sum_{j=1,\neq i}^N g_{ij} \x_{12,j}$ converges to a function that depends only on $(\x_{2i}, a_i)$ (see Section \ref{appendix: distribution of est} in the Supplementary Appendix.).  

Throughout the rest of the paper, when we consider $(\x_{2i},a_i)$ as control function, we will without loss of generality apply the restriction in Assumption \ref{as:basic.alternative} that $\x_{1i}$ and $\x_{2i}$ do not overlap.

\section{Estimation}\label{section: estimation}

In this section we present two estimation methods. In subsections \ref{subsection: estimation, network formation} and \ref{sec: estimation with x,a as control} we discuss estimation using $a_i$ and $(\x_{2i},a_i)$  as control functions, respectively.

\subsection{With $a_i$ as Control Function}\label{subsection: estimation, network formation}

The identification scheme of Theorem \ref{theorem:identification} identifies the parameter of interest $\beta^0$ with the two step procedure: (i) control $a_i$ in the outcome equation and yield $y_i - \E(y_i|a_i) = (\w_i - \E(\w_i|a_i))'\beta^0 + \upsilon_i - \E(\upsilon_i)$, and then (ii) use $\z_i - \mathbb{E}(\z_i|a_i)$ as IVs for $\w_i - \E(\w_i|a_i)$. If we observe $a_i$ and know the conditional mean functions  $\h(a_i)=(h^y(a_i),\h^{w}(a_i),\h^z(a_i))
:=(\mathbb{E}[y_i|a_i], \mathbb{E}[\w_i|a_i],\mathbb{E}[\z_i|a_i])$, then $\beta^0$ can be estimated using 2SLS as 
\begin{align}\label{def.2sls.inf.ai}
&\widehat{\beta}_{2SLS}^{\text{inf}} \nonumber \\
&= \left[  \sum_{i=1}^N (\w_i - \h^{w}(a_i))(\z_i - \h^{z}(a_i))' \left(\sum_{i=1}^N (\z_i - \h^{z}(a_i)) (\z_i - \h^{z}(a_i))'\right)^{-1}
 \sum_{i=1}^N (\z_i - \h^{z}(a_i)) (\w_i - \h^{w}(a_i))' \right]^{-1}  \nonumber \\
&\times 
\left[  \sum_{i=1}^N (\w_i - \h^{w}(a_i))(\z_i - \h^{z}(a_i))' \left(\sum_{i=1}^N (\z_i - \h^{z}(a_i)) (\z_i - \h^{z}(a_i))'\right)^{-1} 
 \sum_{i=1}^N (\z_i - \h^{z}(a_i)) (y_i - h^{y}(a_i))' \right].
\end{align}  
However, since the individual heterogeneity $a_i$ is not observed and the conditional mean functions $\h(a_i) = (\mathbb{E}(y_i|a_i), \mathbb{E}(\w_i|a_i),\mathbb{E}(\z_i|a_i))$ are not known either, the estimator $\widehat{\beta}_{2SLS}^{\text{inf}}$ is not feasible. 

A natural implementation of the infeasible estimator $\widehat{\beta}_{2SLS}^{\text{inf}}$ is to replace the conditional mean function $\h(a_i)$ with its estimate. Suppose that $\widehat{a}_i$ is an estimator of $a_i$ and $\widehat{\h}(\widehat{a}_i)$ is a nonparametric estimator of $\h(a_i)$. Then we can implement the infeasible estimator $\widehat{\beta}_{2SLS}^{\text{inf}}$ with
\begin{align}
&\widehat{\beta}_{2SLS} \\
&:= 
\left[  \sum_{i=1}^N (\w_i - \widehat{\h}^w(\widehat{a}_i))(\z_i - \widehat{\h}^z(\widehat{a}_i))' 
\left(\sum_{i=1}^N (\z_i - \widehat{\h}^z(\widehat{a}_i)) (\z_i - \widehat{\h}^z(\widehat{a}_i))'\right)^{-1} 
 \sum_{i=1}^N (\z_i - \widehat{\h}^z(\widehat{a}_i)) (\w_i - \widehat{\h}^w(\widehat{a}_i))' \right]^{-1}  \nonumber \\
&\times 
\left[  \sum_{i=1}^N (\w_i - \widehat{\h}^w(\widehat{a}_i))(\z_i - \widehat{\h}^z(\widehat{a}_i))' 
\left(\sum_{i=1}^N (\z_i - \widehat{\h}^z(\widehat{a}_i)) (\z_i - \widehat{\h}^z(\widehat{a}_i))'\right)^{-1} 
\sum_{i=1}^N (\z_i - \widehat{\h}^z(\widehat{a}_i)) (y_i - \widehat{\h}^y(\widehat{a}_i))' \right].
\label{def.betahat}
\end{align}
See Section \ref{appendix:beta_hat} in the Appendix for more details on the estimator $\widehat{\beta}_{2SLS}$.

{\noindent \bf Estimation of $\h(\cdot)$:} We can estimate $\h(\cdot)$ using various standard nonparametric methods. In this paper we consider a (linear) sieve estimation method.\footnote{In principle we can use other nonparametric estimation methods such as kernel smoothing or local polynomial methods.\label{kernel}}
Suppose that $h^l(a)$ is the $l^{th}$ element in $\h(a)$ for $l=1,...,L$, where $L$ is the dimension of $(y_i,\w_i',\z_i')'$. 
The sieve estimation method assumes that each function $h^l(a)$, $l=1,...,L$ is well approximated by a linear combination of base functions $(q_1(a),...,q_{K_N}(a))$: 
\begin{equation}
h^l(a) \cong \sum_{k=1}^{K_N}q_k(a)\alpha_k^l, \label{eq.approximation}	
\end{equation}
as the truncation parameter $K_N \rightarrow \infty$. 
A linear sieve (or series) estimator of a function, for example $\widehat{h}^y(\widehat{a}_i)$, is the OLS projection of $y_i$ on the sieve basis $\q^K(\cdot) = (q_1(\cdot),...,q_K(\cdot))'$ with $\widehat{a}_i$ plugged in,
\[
	\widehat{h}^y(\widehat{a}_i) := \q^K(\widehat{a}_i)' \left( \sum_{i=1}^N \q^K(\widehat{a}_i)\q^K(\widehat{a}_i)' \right)^{-1} 
	 \sum_{i=1}^N \q^K(\widehat{a}_i)y_i.
\]

For the regularity conditions of the sieve basis $\q^K(a_i)$, we impose standard conditions such as those proposed by \cite{Newey1997} and \citet{Li2008a}. These assumptions ensure that $\sum_{i=1}^N \q^K(a_i) \q^K(a_i)'$ is asymptotically non-singular and control the rate of approximation of the sieve estimator. These assumptions are formally stated in Assumptions \ref{assumption:sieve basis} and \ref{assumption:sieve basis.alternative} of the Appendix. 

Additionally, we require  that the sieve basis satisfy a Lipschitz condition, which allows us to control for the error introduced by the estimation of $a_i$ with $\widehat{a}_i$ in the estimation of $\widehat{\beta}_{2SLS}$\footnote{This issue is similar to the two step series estimation problem in \cite{Newey2009}. \label{newey2009reference}Other papers that investigated the problem of nonparametric or semiparametric analysis with generated regressors include \cite{Ahn1993}, \cite{Mammenetal2012}, \cite{HahnRidder2013}, and \cite{Escancianoetal2014}, for example.} (see Assumptions \ref{assumption: Lipschitz condition} and \ref{assumtion:sieve with x2, Lipschitz}).  
As an example, define the polynomial sieve as follows.
Let $Pol(K_N)$ denote the space of polynomials on $[-1,1]$ of degree $K_N$, 
\[
Pol(K_N)=\left\{
\nu_0 + \sum_{k=1}^{K_N}\nu_k a^k, \ a \in [-1,1], \nu_k \in \mathbb{R}
\right\}.
\]
For any $k$ we have 
\[
\big|   a_1^k  - a_2^k  \big| 
= k | \tilde{a}^k | | a_1- a_2 | \leq M k | a_1- a_2 |
,
\]
where $\tilde{a} \in [-1,1]$ and $M$ is a finite constant.

In sieve estimations an important issue is choosing the truncation parameter $K_N$. Well-known procedures for selecting $K_N$ are Mallows' $C_P$, generalized cross-validation and leave-one-out cross-validation. For more on these methods see Chapter 15.2 in \cite{Li2008a}, \cite{li1987asymptotic}, \cite{wahba1985comparison}, \cite{andrews1991asymptotic} and \cite{hansen2014nonparametric}. However, these methods are mainly applicable when the observations are cross-sectionally independent, which is not true in our case, especially when the network is dense\label{dense3}, as we assume. Developing a data-driven choice of $K_N$ is beyond the scope of this paper and we leave it for future work.
\bigskip

{\noindent \bf Estimation of $a_i$:} A desired estimator of $a_i$ should satisfy the following high level condition.  
\begin{assumption}[Estimation of $a_i$]\label{assumption: estimation of a_i}
	We assume that we can estimate $a_i$ with $\widehat{a}_i$ such that $\max_i |\widehat{a}_i-a_i|=O_p\left( \zeta_{a}(N)^{-1} \right)$,
	where $\zeta_{a}(N) \rightarrow \infty$ as $N \rightarrow \infty$, satisfying Assumption \ref{assumption: Lipschitz condition} in the Appendix.  
\end{assumption}
{\color{black} Here $\zeta_a(N)$ is the order of  magnitude that measures the Lipschitz smoothness of the sieve basis. The assumption puts restrictions on the uniform bound of the convergence rate of $\widehat{a}_i$, and we need a more accurate estimator of $a_i$ when the average curvature of the sieve basis is larger.} 

For the purpose of our paper, any estimation method that yields an estimator $\widehat{a}_i$ satisfying the restriction in Assumption \ref{assumption: estimation of a_i} can be adopted. 
For example,  assuming the parametric specification as in (\ref{model.network.formation.parametric}), 
\begin{equation}
d_{ij}=\mathbb{I}(t(\x_{2i},\x_{2j})'\lambda+a_i+a_j \geq u_{ij}) \label{model.network.formation.parametric.2}
\end{equation}
with regularity conditions of Assumption \ref{assumption:Graham} in the Appendix, {\color{black} including the error $u_{ij}$ following a logistic distribution,} 
\cite{Graham2017} showed that the joint maximum likelihood estimator that solves
\begin{align*}
& (\widehat{a}_1,...,\widehat{a}_N) 
\\
&:= \argmax_{\lambda,(a_1,...,a_N) } 
\left(\sum_{i=1}^N\sum_{j<i}d_{ij} \exp\left(t(\x_{2i},\x_{2j})'\lambda+a_i+a_j
\right)-\ln\left[1+\exp(t(\x_{2i},\x_{2j})'\lambda+a_i+a_j)\right]\right) 
\end{align*}
satisfies
\begin{equation}
\sup_{1\leq i\leq N} |\widehat{a}_i-a_i|
\leq O\left(\sqrt{\frac{\ln N}{N}}\right) \label{eq.a_ihat.unform.convergence}
\end{equation}
with probability $1-O(N^{-2})$. 
In this case we have $\zeta_a(N) = \sqrt{\frac{N}{\ln N}}$.
Notice that the requirement that the network formation in (\ref{model.network.formation.parametric.2}) be dense is necessary for $\widehat{a}_i$ to satisfy the desired uniform convergence rate in (\ref{eq.a_ihat.unform.convergence}). \label{remark dense-limitation}
Examples of other estimation methods include \cite{Fernandez-val}, \cite{jochmans2016modified}, \cite{dzemski2017empirical}, and \cite{jochmans2018semiparametric}.\label{dzemski, jochman}
\subsection{With $(\x_{2i},a_i)$ as Control Function}\label{sec: estimation with x,a as control} 
As we assume in Section \ref{section: alternative identification}, we consider the case where $\x_{1i}$ and $\x_{2i}$ do not overlap. When $a_i$ is observed and the conditional expectations $\h_{*}(\x_{2i},a_i) = (h_{*}^y(\x_{2i},a_i),\h_{*}^w(\x_{2i},a_i),\h_{*}^z(\x_{2i},a_i)):= (\E(y_i|\x_{2i},a_i),\E(\w_i|\x_{2i},a_i),\E(\z_i|\x_{2i},a_i))$ are known, we can estimate $\beta^0$ by the 2SLS similar to $\widehat{\beta}^{\inf}_{2SLS}$ in (\ref{def.2sls.inf.ai}),
\begin{align}\label{def.2sls.inf.x2iai}
&\bar{\beta}_{2SLS}^{\text{inf}} \nonumber \\
&= \left[  \sum_{i=1}^N (\w_i - \h_{*}^{w}(\x_{2i},a_i))(\z_i - \h_*^{z}(\x_{2i},a_i))' \left(\sum_{i=1}^N (\z_i - \h_*^{z}(\x_{2i},a_i)) (\z_i - \h_*^{z}(\x_{2i},a_i))'\right)^{-1} \right. \nonumber \\
& \left. \qquad \times \sum_{i=1}^N (\z_i - \h_*^{z}(\x_{2i},a_i)) (\w_i - \h_*^{w}(\x_{2i},a_i))' \right]^{-1}  \nonumber \\
&\times 
\left[  \sum_{i=1}^N (\w_i - \h_*^{w}(\x_{2i},a_i))(\z_i - \h_*^{z}(\x_{2i},a_i))' \left(\sum_{i=1}^N (\z_i - \h_*^{z}(\x_{2i},a_i)) (\z_i - \h_*^{z}(\x_{2i},a_i))'\right)^{-1} \right. \nonumber \\
& \left. \qquad  \times \sum_{i=1}^N (\z_i - \h_*^{z}(\x_{2i},a_i)) (y_i - h_*^{y}(\x_{2i},a_i))' \right]^{-1}.
\end{align}


When $a_i$ is unknown and $\x_{2i}$ is also used in the control function, under the monotonicity condition of the link formation as in (\ref{eq.monotone.link.formation}), we can implement the infeasible estimator using the average node degree without estimating $a_i$. 
To be more specific, first we denote 
\begin{align*}
&\PR( d_{ij} = 1 | \x_{2i},a_i) =: \text{deg}(\x_{2i},a_i) =: \text{deg}_i. 
\end{align*}
Under the monotonicity condition in (\ref{eq.monotone.link.formation}), $(\x_{2i},a_i)$ and $(\x_{2i}, \text{deg}_i)$ are one-to-one. This implies that for any $b_i \in \{ y_i,\w_{i},\z_{i} \}$, 
\[
h_*^b(\x_{2i},a_i) = \E(b_i| \x_{2i},a_i) = \E(b_i| \x_{2i},{\rm deg}_i) =: h_{**}^b(\x_{2i},{\rm deg}_i).
\]

Notice that the natural estimator of ${\rm deg}_i$ is the node degree of $i$, the number of connections with node (individual) $i$ in the network scaled by the network size: 
\[
\widehat{\text{deg}}_i := \frac{1}{N-1} \sum_{j=1, \neq i}^N d_{ij}.
\]
Recall that the link $d_{ij}$ is formed by 
\[
d_{ij}=\mathbb{I}( g( t(\x_{2i},\x_{2j}), a_i, a_j ) - u_{ij} \geq 0).
\]
Also recall that the unobserved link-specific error terms $u_{ij}$ are assumed to be independent of all the other variables and randomly drawn. Let $\Phi(\cdot)$ be the cdf of $u_{ij}$. Also let $\pi(\x_2,a)$ be the joint density function of $(\x_{2i},a_i)$.  
Then, for each $(\x_{2i},a_i)$, by the WLLN conditioning on $(\x_{2i},a_i)$, we have  
\begin{align}
\widehat{\text{deg}}_i &:= \frac{1}{N-1} \sum_{j=1, \neq i}^N \mathbb{I}(g( t(\x_{2i},\x_{2j}), a_i, a_j ) - u_{ij}\geq 0) \nonumber \\
&\rightarrow_p \int \Phi\left( g( t(\x_{2i},\x_{2}), a_i, a ) \right) \pi(\x_2,a) d\x_2 da \nonumber \\
&= \PR( d_{ij} = 1 | \x_{2i},a_i) \nonumber \\
&=: \text{deg}_i > 0 \label{eq.deg.limit}
\end{align}
as the network size $N$ grows to infinity. Here the limit of the average network $\rm{deg}_i > 0$ follows since we assume the network is dense.

This shows that $\widehat{\text{deg}}_i$ can be used as an estimator of $\text{deg}_i$.
In fact, we can show that under the regularity conditions in Assumption \ref{assumption:limit.dist} in the Appendix, $\sup_{i} \E [( \sqrt{N} (\widehat{\text{deg}}_i - \text{deg}_i ))^{2B} ] < \infty$ for any finite integer $B \geq 2$, from which we can deduce that 
\begin{equation}
\max_{1 \leq i \leq N} |\widehat{\text{deg}}_i -  \text{deg}_i | = O_p\left(\zeta_{deg}(N)^{-1}\right), \label{eq.estimation.deg}
\end{equation}
where 
\begin{equation*} 
\zeta_{deg}(N):= o(1) N^{\frac{B-1}{2B}}. 
\end{equation*} 
This corresponds to the regularity condition in Assumption \ref{assumption: estimation of a_i}.

Suppose that $\rr^{K}(\x_{2i},\text{deg}_i) =(r_1(\x_{2i},\text{deg}_i),\ldots,r_{K}(\x_{2i},\text{deg}_i))'$ is a sieve basis of the unknown function $\h_*(\x_{2i},a_i)$. 
For each $b_i \in \{ y_i,\w_{i},\z_{i} \}$, a sieve estimator of $h_{**}^b(\x_{2i},\text{deg}_i) = \E(b_i| \x_{2i},a_i)$ is the OLS projection of $b_i$ on $\rr^{K}(\x_{2i},\widehat{\text{deg}}_i)$. For example,
\begin{align*}
 \widehat{h}_*^y(\x_{2i},a_i) 
 &= \widehat{h}_{**}^y(\x_{2i},\text{deg}_i) \\
 &= \rr^K(\x_{2i},\widehat{\text{deg}}_i)' 
 \left( \sum_{i=1}^N \rr^K(\x_{2i},\widehat{\text{deg}}_i)\rr^K(\x_{2i},\widehat{\text{deg}}_i)' \right)^{-1}  
 \sum_{i=1}^N \rr^K(\x_{2i},\widehat{\text{deg}}_i)y_i. 
\end{align*}
Then, we have
\begin{align}\label{def.betabar}
&\bar{\beta}_{2SLS} \nonumber \\
&= \left[  \sum_{i=1}^N (\w_i - \widehat{\h}_{*}^{w}(\x_{2i},a_i))(\z_i - \widehat{\h}_*^{z}(\x_{2i},a_i))' \left(\sum_{i=1}^N (\z_i - \widehat{\h}_*^{z}(\x_{2i},a_i)) (\z_i - \widehat{\h}_*^{z}(\x_{2i},a_i))'\right)^{-1} \right. \nonumber \\
& \left. \qquad \times \sum_{i=1}^N (\z_i - \widehat{\h}_*^{z}(\x_{2i},a_i)) (\w_i - \widehat{\h}_*^{w}(\x_{2i},a_i))' \right]^{-1}  \nonumber \\
&\times 
\left[  \sum_{i=1}^N (\w_i - \widehat{\h}_*^{w}(\x_{2i},a_i))(\z_i - \widehat{\h}_*^{z}(\x_{2i},a_i))' \left(\sum_{i=1}^N (\z_i - \widehat{\h}_*^{z}(\x_{2i},a_i)) (\z_i - \widehat{\h}_*^{z}(\x_{2i},a_i))'\right)^{-1} \right. \nonumber \\
& \left. \qquad  \times \sum_{i=1}^N (\z_i - \widehat{\h}_*^{z}(\x_{2i},a_i)) (y_i - \widehat{h}_*^{y}(\x_{2i},a_i))' \right]^{-1}.
\end{align}
For more details see Section \ref{appendix:beta_bar} in the Appendix.

The two different estimators $\widehat{\beta}_{2SLS}$ and $\bar{\beta}_{2SLS}$ are implemented using different control functions, and these two approaches have their own pros and cons. For $\widehat{\beta}_{2SLS}$, a good estimator of $a_i$ is required, which imposes restrictions on the network formation model (\ref{model.network.formation}) in the form of (\ref{model.network.formation.parametric}). Compared to this, the estimator $\bar{\beta}_{2SLS}$ that uses $(\x_{2i}, {\rm deg}_i)$ as control functions does not require a restriction like (\ref{model.network.formation.parametric}). It requires only the monotonicity of the net surplus function as in (\ref{eq.monotone.link.formation}) of Section \ref{sec: model of network formation}. However, $\bar{\beta}_{2SLS}$ has disadvantages: because it uses $x_{2i}$ as a part of the control function, as discussed in Section \ref{section: alternative identification}, this approach cannot identify and estimate the coefficients of the regressor $\x_{2i}$ if $\x_{2i}$ is a relevant regressor of the outcome.\label{remark-comparison-of-approaches}
 Later in Section \ref{section: monte carlo}, where we present the Monte Carlo simulations, we compare the finite sample properties of $\widehat{\beta}_{2SLS}$ and $\bar{\beta}_{2SLS}$ in both dense and sparse network setups. 



\section{Limit Distribution and Standard Error}\label{subsection: estimation, limiting distributin of estimator} 

In this section we present the asymptotic distributions of the two 2SLS estimators $\widehat{\beta}_{2SLS}$ and $\bar{\beta}_{2SLS}$, and show how to estimate standard errors. We also discuss key technical issues in deriving the limits. All  details of the technical derivations and proofs can be found in the Appendix. 

\subsection{Limiting Distribution and Standard Error of $\widehat{\beta}_{2SLS}$}

Recall the definitions $h^{y}(a_i):= \E[y_i|a_i], \quad h^{\upsilon}(a_i):= \E[\upsilon_i|a_i], \quad \h^{\w}(a_i) := \mathbb{E} (\w_i|a_i), \quad \h^{\z}(a_i) := \mathbb{E} (\z_i|a_i).$ 
Define
$\eta^{y}_i: = y_i - h^{y}(a_i), \quad \eta^{\upsilon}_i: = \upsilon_i - h^{\upsilon}(a_i), \quad \eta^{\w}_i = \w_i - \h^{\w}(a_i), \quad \eta_i^{\z} = \z_i - \h^{\z}(a_i).$
Let $\bm{\eta}_N^{\upsilon} = (\eta^{\upsilon}_1,...,\eta^{\upsilon}_N)'$ and $\HH^{\upsilon}_N(\an) = (h^{\upsilon}(a_1),...,h^{\upsilon}(a_N))'$.   
Let $\widehat{h}^{\upsilon}(a_i)$, $\widehat{\h}^{\w}(a_i)$, and $\widehat{\h}^{\z}(a_i)$ denote the sieve estimators of $h^{\upsilon}(a_i)$, $h^{\w}(a_i)$ and $h^{\z}(a_i)$, respectively. 

In the Appendix, we derive the asymptotic distribution of $\widehat{\beta}_{2SLS}$ in three steps.  
First, we show that the sampling error caused by the use of $\hat{a}_i$ instead of $a_i$ is asymptotically negligible (see  Lemma \ref{lemma: error from A_hat} of the Supplementary Appendix \ref{appenxid: error from A_hat-A}.).  
Next, we control the error introduced by the non-parametric estimation of $h^{l}(a_i)$, where $l \in \{\upsilon,\w,\z\}$. In Lemma \ref{lemma: series approximation error} of the Supplementary Appendix \ref{appendix: series approximation error} we show  that under the regularity conditions,  the estimation error in  $\widehat{h}^l(a_i)$ vanishes at a suitable rate. Combining these two, we deduce
\[
\sqrt{N} (\widehat{\beta}_{2SLS} - \widehat{\beta}^{\inf}_{2SLS}) = o_p(1).
\]
The last step is to derive the limiting distribution of the infeasible estimator $\sqrt{N} ( \widehat{\beta}^{\inf}_{2SLS} - \beta^0)$.
In the Supplementary Appendix \ref{appendix: distribution of est} we show the following:  
\begin{align}
\frac{1}{N}\sum_{i=1}^N(\w_i-\h^{\w}(a_i))
(\z_i-\h^\z(a_i))' &\xrightarrow{p} \mathbf{S}^{\w\z} \label{eq.WLLN.wz}  \\
\frac{1}{N}\sum_{i=1}^N(\z_i-\h^{\z}(a_i))
(\z_i-\h^{\z}(a_i))' &\xrightarrow{p}  \mathbf{S}^{\z\z} \label{eq.WLLN.zz}  \\
\frac{1}{\sqrt{N}}\sum_{i=1}^N(\z_i-\h^{\z}(a_i))
\eta^{\upsilon}_i 
& \Rightarrow \mathcal{N}(0,\mathbf{S}^{\z\z\sigma}), \label{eq.asy.normal.a_i}
\end{align}
where the closed forms of the limits $\mathbf{S}^{\w\z}$ and $\mathbf{S}^{\z\z}$ are found in Lemma \ref{lemma: limit of S^ZZ and S^WZ} and $\mathbf{S}^{\z\z\sigma}$ in Lemma \ref{lemma:numerator.limit variance} of Supplementary Appendix.

Notice that the derivation of the limiting distribution in (\ref{eq.asy.normal.a_i}) allows $\eta^{\upsilon}_i = \upsilon_i - \E(\upsilon_i|a_i)$ to be conditionally heteroskedastic, and so \label{heteroskedasticity} $\sigma^2(\x_i,a_i) := \mathbb{E}[(\upsilon_i - \E[ \upsilon_i|a_i])^2| \x_i,a_i]$ is allowed to depend on $(\x_i,a_i)$.

Combining all the limit results leads to the following theorem. 
\begin{theorem}[Limiting Distribution]\label{theorem: central limit theorem}
Suppose that Assumptions \ref{as:basic}, \ref{assumption: rank}, \ref{assumption: estimation of a_i}, \ref{assumption:sieve basis}, \ref{assumption: Lipschitz condition}, and \ref{assumption:limit.dist}(i)-(v) in the Appendix hold. Then, we have
\begin{align*}
\sqrt{N}(\widehat{\beta}_{2SLS}-\beta^0)
& \Rightarrow 
\mathcal{N}
\left(0, \Omega \right),
\end{align*}
where
\begin{align}
\Omega &= \left(\mathbf{S}^{\w\z}\left(\mathbf{S}^{\z\z}\right)^{-1}(\mathbf{S}^{\w\z})^{\prime}\right)^{-1} 
\left( 
\mathbf{S}^{\w\z}\left(\mathbf{S}^{\z\z}\right)^{-1}\mathbf{S}^{\z\z\sigma}
\left(\mathbf{S}^{\z\z}\right)^{-1} (\mathbf{S}^{\w\z})^{\prime}
\right) \left(\mathbf{S}^{\w\z}\left(\mathbf{S}^{\z\z}\right)^{-1}(\mathbf{S}^{\w\z})^{\prime}\right)^{-1}.
\end{align}
\end{theorem}

The theorem requires several regularity conditions which are presented in Appendix \ref{appendix: assumptions}. 
In addition to conditions of random sampling of
$(y_i,\x_i,a_i)$ in Assumption \ref{as:basic} and the full rank condition in Assumption \ref{assumption: rank}, we assume conditions that ensure $a_i$ can be consistently estimated, and that the error between $\mathbf{h}(a_i)$ and $\widehat{\mathbf{h}}(\widehat{a}_i)$ converges to zero at a suitable rate (Assumptions \ref{assumption: estimation of a_i}, \ref{assumption:sieve basis} and \ref{assumption: Lipschitz condition}).  
We also impose restrictions on the outcome model (\ref{model:outcome}) and the network formation model (\ref{model.network.formation}) (Assumption \ref{assumption:limit.dist}). We assume $|\beta_1^0|$ is bounded below $1$ so that the spillover effect has a unique solution, and $\| \beta_2^0 \|$ is bounded above $0$ so that the IVs are strong. We also assume the observables $(y_i,\x_i)$ and $\mathbf{t}_{ij}$ are bounded, and $a_i$ has a compact support in $[-1,1]$. This boundedness condition is required as a technical regularity condition that simplifies the proofs of the limits in (\ref{eq.WLLN.wz}), (\ref{eq.WLLN.zz}), and (\ref{eq.asy.normal.a_i}), which involves some uniformity in the limit.  \label{boundedness1}

The asymptotic variance can be consistently estimated by 
\begin{align}
\widehat{\Omega} 
= \left(\widehat{\mathbf{S}}^{\w\z} \left(\widehat{\mathbf{S}}^{\z\z}\right)^{-1} (\widehat{\mathbf{S}}^{\w\z})^{\prime}\right)^{-1} 
\left( 
\widehat{\mathbf{S}}^{\w\z} \left(\widehat{\mathbf{S}}^{\z\z}\right)^{-1} \widehat{\mathbf{S}}^{\z\z\sigma}
\left(\widehat{\mathbf{S}}^{\z\z}\right)^{-1} (\widehat{\mathbf{S}}^{\w\z})^{\prime}
\right) 
\left(\widehat{\mathbf{S}}^{\w\z}\left(\widehat{\mathbf{S}}^{\z\z}\right)^{-1}(\widehat{\mathbf{S}}^{\w\z})^{\prime} \right)^{-1},
\end{align}
where 
\begin{align*}
\widehat{\mathbf{S}}^{\w\z} 
&=\frac{1}{N}\sum_{i=1}^N\left(\w_i - \widehat{\h}^\w(\widehat{a}_i) \right)
\left(\z_i - \widehat{\h}^\z(\widehat{a}_i) \right)' \\
\widehat{\mathbf{S}}^{\z\z} 
&=\frac{1}{N}\sum_{i=1}^N\left(\z_i - \widehat{\h}^\z(\widehat{a}_i) \right)
\left(\z_i - \widehat{\h}^\z(\widehat{a}_i) \right)' \\
\widehat{\mathbf{S}}^{ZZ\sigma^2} 
&=\frac{1}{N}\sum_{i=1}^N\left(\z_i - \widehat{\h}^\z(\widehat{a}_i) \right)
\left(\z_i - \widehat{\h}^\z(\widehat{a}_i) \right)' (\widehat{\eta}^{\upsilon}_i)^2,
\end{align*}
and $\widehat{\eta}^{\upsilon}_i=y_i-\widehat{h}^y(\widehat{a}_i)-(\w_i-\widehat{\h}^{\w}(\widehat{a}_i))'\widehat{\beta}_{2SLS}.$

\subsection{Limiting Distribution and Standard Error of $\bar{\beta}_{2SLS}$}
The process is analogous to the one presented in the previous section.
Again, let $b_i^l$ be the $l^{th}$ element in $(y_i,\w_i',\z_i')'$.
Recall the definition that 
\begin{align*}
h^l_{*}(\x_{2i},a_i) =\mathbb{E}[b^l_i|\x_{2i},a_i]  = \mathbb{E}[b^l_i|\x_{2i},\text{deg}_i] =: h^l_{**}(\x_{2i},\text{deg}_i). 
\end{align*} Further, let $\eta^l_{*i}=b^l_i-h^l_{*}(\x_{2i},a_i) = b^l - h^l_{**}(\x_{2i},\text{deg}_i)$, and 
let $\widehat{h}^l_{**}(\x_{2i},\text{deg}_i)$ denote a sieve estimator of $h^l_{**}(\x_{2i},\text{deg}_i)$. 

As in the previous section, we derive the asymptotic distribution of $\bar{\beta}_{2SLS}$ in three steps.  
	First, we show that the error that stems from the use of the estimate $\widehat{\text{deg}_i}$ for $\text{deg}_i$, $\widehat{h}^l_{**}(\x_{2i},\widehat{\text{deg}}_i) - \widehat{h}^l_{**}(\x_{2i},\text{deg}_i)$, is asymptotically negligible. 
	In the second step, we control the error introduced by the non-parametric estimation of $h^l_{**}(\x_{2i},\text{deg}_i)$, $\widehat{h}^l_{**}(\x_{2i},\text{deg}_i)-h^l_{**}(\x_{2i},\text{deg}_i)$. This implies
	\[
	\sqrt{N} (\bar{\beta}_{2SLS} - \bar{\beta}^{\inf}_{2SLS}) = o_p(1).
	\]
	The last step is to derive the limiting distribution of the infeasible estimator $\sqrt{N} ( \bar{\beta}^{\inf}_{2SLS} - \beta^0)$ by showing
	\begin{align*}
	\frac{1}{N}\sum_{i=1}^N(\w_i-\h_*^{\w}(\x_{2i},a_i))
	(\z_i-\h_*^\z(\x_{2i},a_i))' &\xrightarrow{p} \mathbf{\bar{S}}^{\w\z} \\
	\frac{1}{N}\sum_{i=1}^N(\z_i-\h_*^{\z}(\x_{2i},a_i))
	(\z_i-\h_*^{\z}(\x_{2i},a_i))' &\xrightarrow{p}  \mathbf{\bar{S}}^{\z\z} \\
	\frac{1}{\sqrt{N}}\sum_{i=1}^N(\z_i-\h_*^{\z}(\x_{2i},a_i))
	\eta^{\upsilon}_{*i}
	& \Rightarrow \mathcal{N}(0,\mathbf{\bar{S}}^{\z\z\sigma}),
	\end{align*}

	Combining all the limit results we have the following theorem. 
	\begin{theorem}[Limiting Distribution]\label{theorem: central limit theorem beta_bar}
		Suppose that Assumptions   \ref{as:basic}, \ref{as:basic.alternative}, \ref{assumption:rank.alternative}, 
		\ref{assumption:sieve basis.alternative}, \ref{assumtion:sieve with x2, Lipschitz}, and \ref{assumption:limit.dist} hold. Then, we have
		\begin{align*}
		\sqrt{N}(\bar{\beta}_{2SLS}-\beta^0)
		& \Rightarrow 
		\mathcal{N}
		\left(0, \bar{\Omega} \right),
		\end{align*}
		where
		\begin{align}
		\bar{\Omega} &= \left(\mathbf{\bar{S}}^{\w\z}\left(\mathbf{\bar{S}}^{\z\z}\right)^{-1}(\mathbf{\bar{S}}^{\w\z})^{\prime}\right)^{-1} 
		\left( 
		\mathbf{\bar{S}}^{\w\z}\left(\mathbf{\bar{S}}^{\z\z}\right)^{-1}\mathbf{\bar{S}}^{\z\z\sigma}
		\left(\mathbf{\bar{S}}^{\z\z}\right)^{-1} (\mathbf{\bar{S}}^{\w\z})^{\prime}
		\right) 
		\left(\mathbf{\bar{S}}^{\w\z}\left(\mathbf{\bar{S}}^{\z\z}\right)^{-1} (\mathbf{\bar{S}}^{\w\z})^{\prime}\right)^{-1}. \nonumber
		\end{align}
	\end{theorem}

The asymptotic result in Theorem \ref{theorem: central limit theorem beta_bar} requires the following regularity conditions which are formally presented in the Appendix. First, Assumption \ref{as:basic.alternative} assumes that the regressors in the outcome equation, $\x_{1i}$ and the observables in the network formation $\x_{2i}$ do not overlap. Assumption \ref{assumption:rank.alternative} is a full rank condition for $\bar{\beta}_{2SLS}$.  
Assumptions \ref{assumption:sieve basis.alternative} and \ref{assumtion:sieve with x2, Lipschitz} regard the sieve used in constructing the estimator $\bar{\beta}_{2SLS}$.  
Comparing with the assumptions assumed in Theorem \ref{theorem: central limit theorem}, Theorem \ref{theorem: central limit theorem beta_bar} does not require the high level condition of Assumption \ref{assumption: estimation of a_i} because we do not use an estimator of $a_i$. Instead it requires an additional restriction that the net surplus function in the link formation be strictly monotonic in $a_i$ conditional on $(\x_{2i},\x_{2j},a_j)$, which implies the required monotonicity condition in (\ref{eq.monotone.link.formation}). 

Like in the case of $\widehat{\beta}_{2SLS}$, we allow $\eta^{\upsilon}_{*i} = \upsilon_i - \E(\upsilon_i|\x_{2i},a_i)$ to be conditionally heteroskedastic, and $\sigma^2_{*}(\x_i,a_i) := \mathbb{E}[(\upsilon_i - \E[ \upsilon_i|\x_{2i},a_i])^2| \x_i,a_i]$ is allowed to depend on $(\x_i,a_i)$.

 The asymptotic variance can be consistently estimated by 
\begin{align}
\widehat{\bar{\Omega}} &= \left(\widehat{\mathbf{\bar{S}}}^{\w\z}\left(\widehat{\mathbf{\bar{S}}}^{\z\z}\right)^{-1} (\widehat{\mathbf{\bar{S}}}^{\w\z})^{\prime})\right)^{-1} 
\left( 
\widehat{\mathbf{\bar{S}}}^{\w\z}\left(\widehat{\mathbf{\bar{S}}}^{\z\z}\right)^{-1} \widehat{\mathbf{\bar{S}}}^{\z\z\sigma}
\left(\widehat{\mathbf{\bar{S}}}^{\z\z}\right)^{-1} (\widehat{\mathbf{\bar{S}}}^{\w\z})^{\prime})
\right) 
\left(\widehat{\mathbf{\bar{S}}}^{\w\z}\left(\widehat{\mathbf{\bar{S}}}^{\z\z}\right)^{-1}(\widehat{\mathbf{\bar{S}}}^{\w\z})^{\prime})\right)^{-1},
\end{align}
where 
\begin{align*}
\widehat{\mathbf{\bar{S}}}^{\w\z} 
&=\frac{1}{N}\sum_{i=1}^N\left(\w_i - \widehat{\h}_{**}^\w(\x_{2i},\widehat{\text{deg}}_i) \right)
\left(\z_i - \widehat{\h}_{**}^\z(\x_{2i},\widehat{\text{deg}}_i) \right)' \\
\widehat{\mathbf{\bar{S}}}^{\z\z} 
&=\frac{1}{N}\sum_{i=1}^N\left(\z_i - \widehat{\h}_{**}^\z(\x_{2i},\widehat{\text{deg}}_i) \right)
\left(\z_i - \widehat{\h}_{**}^\z(\x_{2i},\widehat{\text{deg}}_i) \right)' \\
\widehat{\mathbf{\bar{S}}}^{\z\z\sigma^2} 
&=\frac{1}{N}\sum_{i=1}^N\left(\z_i - \widehat{\h}_{**}^\z(\x_{2i},\widehat{\text{deg}}_i) \right)
\left(\z_i - \widehat{\h}_{**}^\z(\x_{2i},\widehat{\text{deg}}_i) \right)' (\widehat{\eta}^{\upsilon}_{**i})^2,
\end{align*}
and $\widehat{\eta}_{**i}^{\upsilon}=y_i-\widehat{h}_{**}^y(\x_{2i},\widehat{\text{deg}}_i)-(\w_i-\widehat{\h}_{**}^{\w}(\x_{2i},\widehat{\text{deg}}_i))'\bar{\beta}_{2SLS}.$	

\section{Monte Carlo}\label{section: monte carlo}
We consider both dense and sparse network Monte Carlo designs. In the dense network case links are formed according to\footnote{This follows the approach of \cite{Graham2017}.}
\[
d_{ij} = \mathbb{I}\left\{ x_{2i}x_{2j}\lambda_d + a_i + a_j -u_{ij} \geq 0 \right\},
\]
where $x_{2i}\in\{-1,1\}$, $\lambda_d=1$ and $u_{ij}$ follows a logistic distribution. This link rule implies that agents have a strong taste for homophilic matching since $x_{2i}x_{2j}\lambda_d=1$ when $x_{2i}=x_{2j}$ and $x_{2i}x_{2j}\lambda_d=-1$ when $x_{2i}\neq x_{2j}$. 

In the sparse network case links are formed according to
\[
d_{ij} = \mathbb{I}\left\{(|x_{2i}-x_{2j}|+3)\lambda_s + a_i + a_j -u_{ij} \geq 0 \right\},
\]
with $\lambda_s=-1$. This rule also implies homophily on observable characteristics.
Individual-level degree heterogeneity is generated according to
\[
a_i=\varphi(\alpha_L\mathbb{I}\left\{ x_{2i}=-1 \right\} +\alpha_H\mathbb{I} \left\{ x_{2i}=1 \right\} + \xi_i),
\]
with $\alpha_L \leq \alpha_H$ and $\xi_i$ a centered Beta random variable $
\xi_i|x_{2i}\sim \left\{Beta(\mu_0,\mu_1)-\frac{\mu_0}{\mu_0+\mu_1}\right\}$
so that $a_i\in \left[\alpha_L-\frac{\mu_0}{\mu_0+\mu_1},\alpha_H+\frac{\mu_1}{\mu_0+\mu_1}\right]$. We choose values of the network formation parameters so that $a_i \in[-1,1]$.
In the main text we present results based on the following parameter values. In the dense network case we set 
$\mu_0=1/4$, $\mu_1=3/4$, $\alpha_L=\alpha_H=-3/4$, which yields an average node degree $=23$ when $N=100$. The sparse network formation design is generated by setting $\mu_0=1$, $\mu_1=1$, $\alpha_L=\alpha_H=-1/4$, which gives an average degree $=1.78$ when $N=100$.\footnote{Results for 14 other network formation designs can be found in Section \ref{appendix: supplementary monte carlo} of the online appendix. Most results are similar to the ones presented in the main text.}

Individual outcomes are generated according to
\[
y_i=\beta_1\sum_{j=1 \atop j\neq i}^N g_{ij}y_j+\beta_2 x_{1i}+\beta_3 
\sum_{j=1 \atop j\neq i}^N g_{ij}x_{1j}+h(a_i)+\varepsilon_i.
\]
In the simulations, we set $\beta_1=0.8$, $\beta_2=\beta_3=5$,  $x_{1i}=3q_1+\cos(q_2)/0.8+\epsilon_i$, where $q_1,q_2\sim\mathcal{N}(x_{2i},1)$, and $\varepsilon_i,\epsilon_i\sim\mathcal{N}(0,1)$.
For $h(a_i)$ we use the following functional forms: $h(a_i)=\exp (3 a_i)$, $h(a_i)=\cos(3 a_i)$,  $h(a_i)=\sin(3 a_i)$. A plot of $h(a_i)$ for these functional forms is presented in Figure \ref{figure: h}. We can see that the exponential function yields a strongly increasing impact on the individual outcome, and with the cosine functions the returns are increasing up to a certain point and then decreasing; however the sine function gives a more irregular pattern.

We estimate the outcome equation coefficients $(\beta_1,\beta_2, \beta_3)$ using the standard 2SLS estimator for peer effects and the Hermite polynomial sieve as well as a polynomial sieve.
 For the dense network case, we estimate $a_i$ using $\widehat{a}_i$ and implement the following control functions: 
using a control function linear in $\widehat{a}_i$,
  $\widehat{h}(\widehat{a}_i)$, $\widehat{h}(a_i)$,
 $\widehat{h}(\widehat{\rm deg}_i,x_{2i})$\footnote{
 	{Note that since $x_{2i}$ is discrete with a finite support, $\{ x_1,...,x_M \}$,} we have
 	$
 	r(x_{2i},{\rm deg_i}) = \sum_{m=1}^M r(x_m,{\rm deg_i}) \mathbb{I}\{ x_{2i} = x_m \}.
 	$
 	We can then approximate 
 	$
 	r(x_{2i}, {\rm deg_i}) \simeq \sum_{k=1}^{K_N} \left\{ \sum_{m=1}^M \alpha_{m,k} q_k^d(\rm deg_i) \mathbb{I}\{ x_{2i} = x_m \} \right\}.
 	$}, and
 $h(a_i)$. For the sparse network case the estimator of $a_i$ is not reliable\footnote{To estimate $a_i$, we use the JMLE proposed in \cite{Graham2017}. As \cite{Graham2017} states, in sparse designs the JMLE rarely even exists, rendering it unusable in practice
 	when the network is too sparse. See \cite{Graham2017} for more details.} and we implement the following control functions: linear in $a_i$, $\widehat{h}(a_i)$,  $\widehat{h}(\widehat{\rm deg}_i,x_{2i})$ and $h(a_i)$. In both the dense and sparse setup we also implement a benchmark model with no control for the endogeneity of the network.

In the paper, due to space limitations, we present Monte Carlo results obtained using the Hermite polynomial sieve with $K_N=4$.
Specifically, Tables \ref{table: MC dense main text} and \ref{MC sparse main text} include results for the dense and sparse network specifications, respectively. Results for the other orders of $K_N$ are not notably different; in the Online Supplement we provide results for fourteen other network formation designs, for $K_N=4,8$ and for the Hermite polynomial and polynomial sieve functions.

\begin{figure}[!h]\caption{\bf $h(a_i)$ for selected functional forms of $h(a_i)$}\label{figure: h}
	\includegraphics[scale=.4]{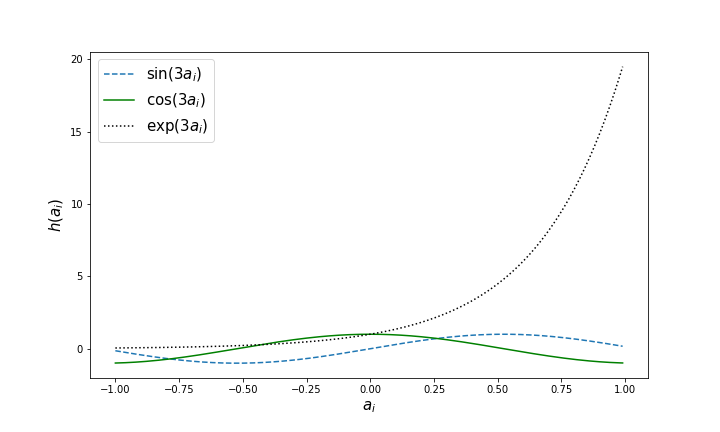}
\end{figure}

 We also perform  conventional leave-one-out cross validation to find data-dependent $K_N$ (chosen as the $K_N$ that minimizes the Root Mean Square Error (RMSE) of the prediction based on the leave-one-out estimator, see for example  \cite{andrews1991asymptotic}, \cite{hansen2014nonparametric}). We report the statistics on the cross-validation in Table \ref{table: CV}.
The differences in RMSE are very small between the different values of $K_N$.\\
\label{remark: MC discussion}
Analyzing the Monte Carlo results for the dense network specification in Table \ref{table: MC dense main text}, we can see that, as expected from our asymptotic theories, the control functions $\widehat{h}(\widehat{a}_i)$ and $\widehat{h}(\widehat{deg}_i,x_{2i})$ perform better than the estimator with a linear control function, as well as the estimator that does not control for the endogeneity of the network in terms of mean bias. This difference is more pronounced in the case when $h(a_i)$ is the sine or cosine function. 
Both the control for degree approach and the control function that uses $\widehat{h}(\widehat{a}_i)$  yield a low bias and have the correct size on all coefficients in all cases. In the simulations we also implemented the control function $\widehat{h}(a_i)$, that is, using the true $a_i$ instead of $\widehat{a}_i$.  These results are very similar to the ones obtained using $\widehat{h}(\widehat{a}_i)$, which is in line with the estimator $\widehat{a}_i$ having a very low bias, as detailed in the table footnotes. This suggest that the approach of using $\widehat{h}(\widehat{a}_i)$ as a control function works very well when a highly precise estimator of $a_i$ is available (for example when the network size $N$ is large.). 

Looking at Table \ref{MC sparse main text} and the results for the sparse design, we can see that the control for degree approach performs very well across all functional forms of $h(a_i)$. In the sparse setup, the bias of all estimates, including those that do not control for the endogeneity of the network, is small. However, the size of the no control and linear control estimates is not correct. If a precise estimator of $a_i$ is available, the control function $\widehat{h}(a_i)$ also performs well with low bias and correct size in all cases.

Table \ref{table: CV} shows that the performance of the estimators does not differ notably for different values of $K_N$. As for the choice of $K_N$ we present in the tables, we have run simulations for a range of values of $K_N$ and the results did not differ significantly. As deriving a theory for a data driven choice of $K_N$ is beyond the scope of this paper, for applied researchers we suggest estimating the model over a range of $K_N$ and seeing whether the results vary significantly. As shown in our Monte Carlo simulations, the control function approach yields results robust to the choice of $K_N$ for different non-linear functions. 

\begin{table}[!h]\caption{\footnotesize {\bf Design 4 dense network: Parameter values across 1000 Monte Carlo replications with $K_N=4$ and Hermite polynomial sieve}} \label{table: MC dense main text}
	\begin{threeparttable} 
		\centering \footnotesize
		\scalebox{.8}{\begin{tabular}{cccccccccccccc}\toprule 
				\multicolumn{14}{c}{$h(a_i) = \exp(a_i)$}\\ 
				\cellcolor{yellow}$N$&\multicolumn{6}{|c|}{\cellcolor{yellow}$100$}&\multicolumn{6}{|c|}{\cellcolor{yellow}$250$}&\\\hline 
				CF&$(0)$&$(1)$&$(2)$&$(3)$&$(4)$&$(5)$& $(0)$ &$(1)$&$(2)$&$(3)$&$(4)$&$(5)$&\\\hline 
				\multirow{4}{*}{$\beta_1=0.8$}& 0.002 & 0.004 &-0.000 &-0.000 &0.000 &-0.000& 0.004 &0.007 &-0.001& -0.001 & -0.001& -0.000 & \textit{mean bias} \\ 
				&(0.010 )&(0.013 )&(0.015 )&(0.015 )&(0.024 )&(0.010 )&(0.009 )&(0.013 )&(0.015 )&(0.015 )&(0.025 )&(0.009 )&\textit{std}\\ 
				& 0.133 & 0.115 &0.056 &0.061 &0.058 &0.058& 0.306 &0.225 &0.057& 0.057 &0.064& 0.050 &\textit{size} \\ \midrule
				\multirow{4}{*}{$\beta_2=5$}& -0.003 & -0.004 &-0.000 &-0.000 &0.000 &-0.000& -0.002 &-0.004 &0.000& 0.000 &-0.000& 0.000 &\textit{mean bias} \\ 
				&(0.031 )&(0.032 )&(0.034 )&(0.033 )&(0.035 )&(0.031 )&(0.020 )&(0.021 )&(0.020 )&(0.020 )&(0.021 )&(0.020 )&\textit{std}\\ 
				& 0.058 & 0.069 &0.074 &0.068 &0.074 &0.057& 0.069 &0.079& 0.055 &0.059 & 0.058 &0.061 &\textit{size} \\\midrule 
				\multirow{4}{*}{$\beta_3=5$}& -0.032 & -0.048 &0.006& 0.008 &0.006 &0.006 &-0.066 &-0.107 &0.009& 0.013 &0.012& 0.009 &\textit{mean bias} \\ 
				&(0.178 )&(0.217 )&(0.251 )&(0.250 )&(0.269 )&(0.174 )&(0.163 )&(0.219 )&(0.249 )&(0.248 )&(0.270 )&(0.152 )&\textit{std}\\ 
				& 0.078 & 0.078 &0.055 &0.060 &0.061 &0.061& 0.156& 0.172 &0.051 &0.054 & 0.062 &0.050 &\textit{size} \\\midrule 
				\multicolumn{14}{c}{$h(a_i) = \sin(a_i)$}\\ 
				\cellcolor{yellow}$N$&\multicolumn{6}{|c|}{\cellcolor{yellow}$100$}&\multicolumn{6}{|c|}{\cellcolor{yellow}$250$}&\\\hline 
				CF&$(0)$&$(1)$&$(2)$&$(3)$&$(4)$&$(5)$& $(0)$ &$(1)$&$(2)$&$(3)$&$(4)$&$(5)$&\\\hline 
				\multirow{4}{*}{$\beta_1=0.8$}& -0.008 & -0.005 &-0.000 &-0.000 &-0.000 &-0.000& -0.015 &-0.010 &-0.001& -0.001 & -0.001& -0.001 & \textit{mean bias} \\ 
				&(0.014 )&(0.014 )&(0.016 )&(0.015 )&(0.025 )&(0.011 )&(0.017 )&(0.015 )&(0.015 )&(0.015 )&(0.026 )&(0.010 )&\textit{std}\\ 
				& 0.464 & 0.160 &0.058 &0.061 &0.059 &0.045& 0.753 &0.293 &0.054& 0.057 &0.071& 0.053 &\textit{size} \\ \midrule
				\multirow{4}{*}{$\beta_2=5$}& 0.007 & 0.005 &-0.001 &-0.000 &0.000 &-0.000& 0.007 &0.005 &-0.000& 0.000 &-0.000& 0.000 &\textit{mean bias} \\ 
				&(0.033 )&(0.034 )&(0.035 )&(0.033 )&(0.036 )&(0.031 )&(0.022 )&(0.022 )&(0.021 )&(0.020 )&(0.021 )&(0.020 )&\textit{std}\\ 
				& 0.075 & 0.072 &0.067 &0.068 &0.072 &0.060& 0.076 &0.071& 0.056 &0.059 & 0.060 &0.055 &\textit{size} \\\midrule 
				\multirow{4}{*}{$\beta_3=5$}& 0.113 & 0.078 &0.009& 0.008 &0.009 &0.005 &0.236 &0.165 &0.010& 0.013 &0.012& 0.012 &\textit{mean bias} \\ 
				&(0.222 )&(0.231 )&(0.258 )&(0.250 )&(0.277 )&(0.191 )&(0.268 )&(0.249 )&(0.255 )&(0.248 )&(0.276 )&(0.177 )&\textit{std}\\ 
				& 0.237 & 0.100 &0.057 &0.060 &0.053 &0.053& 0.646& 0.248 &0.056 &0.054 & 0.055 &0.048 &\textit{size} \\\midrule 
				\multicolumn{14}{c}{$h(a_i) = \cos(a_i)$}\\ 
				\cellcolor{yellow}$N$&\multicolumn{6}{|c|}{\cellcolor{yellow}$100$}&\multicolumn{6}{|c|}{\cellcolor{yellow}$250$}&\\\hline 
				CF&$(0)$&$(1)$&$(2)$&$(3)$&$(4)$&$(5)$& $(0)$ &$(1)$&$(2)$&$(3)$&$(4)$&$(5)$&\\\hline 
				\multirow{4}{*}{$\beta_1=0.8$}& -0.009 & 0.004 &-0.000 &-0.000 &0.000 &-0.000& -0.017 &0.010 &-0.000& -0.001 & -0.001& -0.001 & \textit{mean bias} \\ 
				&(0.016 )&(0.014 )&(0.017 )&(0.015 )&(0.025 )&(0.010 )&(0.018 )&(0.016 )&(0.015 )&(0.015 )&(0.026 )&(0.009 )&\textit{std}\\ 
				& 0.459 & 0.104 &0.055 &0.061 &0.057 &0.053& 0.745 &0.318 &0.059& 0.057 &0.059& 0.046 &\textit{size} \\ \midrule
				\multirow{4}{*}{$\beta_2=5$}& 0.009 & -0.004 &-0.000 &-0.000 &0.001 &-0.000& 0.008 &-0.005 &0.000& 0.000 &0.000& 0.000 &\textit{mean bias} \\ 
				&(0.040 )&(0.034 )&(0.036 )&(0.033 )&(0.037 )&(0.031 )&(0.026 )&(0.022 )&(0.021 )&(0.020 )&(0.021 )&(0.020 )&\textit{std}\\ 
				& 0.075 & 0.061 &0.062 &0.068 &0.070 &0.060& 0.084 &0.077& 0.053 &0.059 & 0.055 &0.062 &\textit{size} \\\midrule 
				\multirow{4}{*}{$\beta_3=5$}& 0.123 & -0.051 &0.004& 0.008 &0.004 &0.004 &0.264 &-0.161 &0.008& 0.013 &0.010& 0.011 &\textit{mean bias} \\ 
				&(0.257 )&(0.232 )&(0.266 )&(0.250 )&(0.286 )&(0.176 )&(0.292 )&(0.258 )&(0.256 )&(0.248 )&(0.276 )&(0.157 )&\textit{std}\\ 
				& 0.224 & 0.074 &0.053 &0.059 &0.055 &0.056& 0.640& 0.256 &0.055 &0.054 & 0.057 &0.047 &\textit{size} \\\midrule 
		\end{tabular}} 
		\begin{tablenotes}\tiny 
			\item CF - control function. $(0)$ - none, $(1)$ - $\lambda_a\hat{a}_i$,  $(2)$ - $\hat{h}(\hat{a}_i)$, $(3)$ - $\hat{h}(a_i)$, $(4)$ - $\hat{h}(\widehat{deg}_i,x_{2i})$, $(5)$ - $h(a_i)$. 
			\item The network design parameters are $\mu_0=0.25$, $\mu_1=0.75$, $\alpha_L=-0.75$, $\alpha_H=-0.75$ 
			\item Average number of links for $N=100$ is $23.0$, for $N=250$ it is $57.8$. 
			\item Average skewness for $N=100$ is $0.66$, for $N=250$ it is $0.89$. 
			\item Size is the empirical size of t-test against the truth. 
			\item N$=100$, $corr(a_i,\bm{x}_{2i})=0.004$,N$=250$, $corr(a_i,\bm{x}_{2i})=0.001$ 
			\item The bias of $\hat{a}_i$ is calculated as $a_i-\hat{a}_i$. 
			\item For $N=100$, $\hat{a}_i$ mean bias$=0.018$, median bias$=0.008$, std$=0.271$. 
			\item For $N=250$, $\hat{a}_i$ mean bias$=0.007$, median bias$=0.004$, std$=0.167$. 
		\end{tablenotes} 
	\end{threeparttable} 
\end{table} 

\begin{table}[!h]\caption{\footnotesize {\bf Design 4 sparse network: Parameter values across 1000 Monte Carlo replications with $K_N=4$ and Hermite polynomial sieve}} \label{MC sparse main text}
	\begin{threeparttable} 
		\centering \footnotesize
		\scalebox{.8}{\begin{tabular}{cccccccccccc}\toprule 
				\multicolumn{12}{c}{$h(a_i) = \exp(a_i)$}\\ 
				\cellcolor{yellow}$N$&\multicolumn{5}{|c|}{\cellcolor{yellow}$100$}&\multicolumn{5}{|c|}{\cellcolor{yellow}$250$}&\\\hline 
				CF&$(0)$&$(1)$&$(2)$&$(3)$&$(4)$& $(0)$ &$(1)$&$(2)$&$(3)$&$(4)$&\\\hline 
				\multirow{4}{*}{$\beta_1=0.8$}& 0.001 & 0.001 &0.000 &0.000 &0.000 &0.002& 0.003 &0.000 &-0.000& 0.000 &\textit{mean bias} \\ 
				&(0.003 )&(0.003 )&(0.002 )&(0.003 )&(0.002 )&(0.004 )&(0.004 )&(0.002 )&(0.003 )&(0.002 )&\textit{std}\\ 
				& 0.089 & 0.090 &0.052 &0.056 &0.049 &0.269& 0.257 &0.072 &0.055& 0.064 &\textit{size} \\ \midrule
				\multirow{4}{*}{$\beta_2=5$}& -0.001 & -0.002 &-0.003 &-0.002 &-0.003 &-0.007& -0.008 &0.000 &0.001& 0.001 &\textit{mean bias} \\ 
				&(0.039 )&(0.039 )&(0.033 )&(0.041 )&(0.032 )&(0.027 )&(0.027 )&(0.021 )&(0.025 )&(0.021 )&\textit{std}\\ 
				& 0.043 & 0.046 &0.065 &0.061 &0.060 &0.078& 0.084 &0.055& 0.066 &0.049 &\textit{size} \\\midrule 
				\multirow{4}{*}{$\beta_3=5$}& -0.004 & -0.004 &-0.002& 0.002 &-0.002 &-0.027 &-0.028 &-0.001 &-0.000& -0.001 &\textit{mean bias} \\ 
				&(0.076 )&(0.077 )&(0.066 )&(0.075 )&(0.065 )&(0.063 )&(0.064 )&(0.052 )&(0.058 )&(0.051 )&\textit{std}\\ 
				& 0.034 & 0.038 &0.063 &0.063 &0.047 &0.085& 0.090& 0.056 &0.068 &0.060 &\textit{size} \\\midrule 
				\multicolumn{12}{c}{$h(a_i) = \sin(a_i)$}\\ 
				\cellcolor{yellow}$N$&\multicolumn{5}{|c|}{\cellcolor{yellow}$100$}&\multicolumn{5}{|c|}{\cellcolor{yellow}$250$}&\\\hline 
				CF&$(0)$&$(1)$&$(2)$&$(3)$&$(4)$& $(0)$ &$(1)$&$(2)$&$(3)$&$(4)$&\\\hline 
				\multirow{4}{*}{$\beta_1=0.8$}& -0.000 & 0.000 &0.000 &0.000 &0.000 &-0.002& -0.000 &0.000 &-0.000& 0.000 &\textit{mean bias} \\ 
				&(0.003 )&(0.002 )&(0.002 )&(0.003 )&(0.002 )&(0.003 )&(0.002 )&(0.002 )&(0.003 )&(0.002 )&\textit{std}\\ 
				& 0.059 & 0.048 &0.052 &0.057 &0.051 &0.170& 0.068 &0.072 &0.059& 0.071 &\textit{size} \\ \midrule
				\multirow{4}{*}{$\beta_2=5$}& -0.007 & -0.002 &-0.003 &-0.002 &-0.003 &0.005& 0.001 &0.000 &0.001& 0.000 &\textit{mean bias} \\ 
				&(0.039 )&(0.032 )&(0.033 )&(0.041 )&(0.032 )&(0.026 )&(0.022 )&(0.021 )&(0.025 )&(0.021 )&\textit{std}\\ 
				& 0.052 & 0.061 &0.066 &0.062 &0.059 &0.083& 0.061 &0.055& 0.073 &0.048 &\textit{size} \\\midrule 
				\multirow{4}{*}{$\beta_3=5$}& -0.001 & -0.001 &-0.002& 0.002 &-0.002 &0.016 &-0.001 &-0.001 &-0.000& -0.001 &\textit{mean bias} \\ 
				&(0.078 )&(0.067 )&(0.066 )&(0.076 )&(0.065 )&(0.064 )&(0.052 )&(0.052 )&(0.058 )&(0.051 )&\textit{std}\\ 
				& 0.059 & 0.053 &0.063 &0.067 &0.049 &0.079& 0.057& 0.056 &0.065 &0.057 &\textit{size} \\\midrule 
				\multicolumn{12}{c}{$h(a_i) = \cos(a_i)$}\\ 
				\cellcolor{yellow}$N$&\multicolumn{5}{|c|}{\cellcolor{yellow}$100$}&\multicolumn{5}{|c|}{\cellcolor{yellow}$250$}&\\\hline 
				CF&$(0)$&$(1)$&$(2)$&$(3)$&$(4)$& $(0)$ &$(1)$&$(2)$&$(3)$&$(4)$&\\\hline 
				\multirow{4}{*}{$\beta_1=0.8$}& 0.001 & 0.001 &0.000 &0.000 &0.000 &0.002& 0.002 &0.000 &-0.000& 0.000 &\textit{mean bias} \\ 
				&(0.003 )&(0.003 )&(0.002 )&(0.003 )&(0.002 )&(0.003 )&(0.003 )&(0.002 )&(0.003 )&(0.002 )&\textit{std}\\ 
				& 0.073 & 0.081 &0.052 &0.053 &0.049 &0.197& 0.216 &0.072 &0.067& 0.068 &\textit{size} \\ \midrule
				\multirow{4}{*}{$\beta_2=5$}& -0.002 & -0.002 &-0.003 &-0.002 &-0.003 &-0.005& -0.006 &0.000 &0.001& 0.000 &\textit{mean bias} \\ 
				&(0.038 )&(0.038 )&(0.033 )&(0.041 )&(0.032 )&(0.025 )&(0.025 )&(0.021 )&(0.025 )&(0.021 )&\textit{std}\\ 
				& 0.047 & 0.051 &0.066 &0.061 &0.062 &0.062& 0.074 &0.055& 0.065 &0.047 &\textit{size} \\\midrule 
				\multirow{4}{*}{$\beta_3=5$}& -0.003 & -0.003 &-0.002& 0.002 &-0.002 &-0.020 &-0.022 &-0.001 &0.000& -0.001 &\textit{mean bias} \\ 
				&(0.073 )&(0.073 )&(0.066 )&(0.074 )&(0.065 )&(0.061 )&(0.062 )&(0.052 )&(0.059 )&(0.051 )&\textit{std}\\ 
				& 0.038 & 0.036 &0.063 &0.065 &0.049 &0.069& 0.079& 0.056 &0.070 &0.062 &\textit{size} \\\midrule 
		\end{tabular}} 
		\begin{tablenotes}\tiny 
			\item CF - control function. $(0)$ - none, $(1)$ - $\lambda_a a_i$, $(2)$ - $\hat{h}(a_i)$, $(3)$ - $\hat{h}(\widehat{deg}_i,x_{2i})$, $(4)$ - $h(a_i)$. 
			\item The network design parameters are $\mu_0=1.00$, $\mu_1=1.00$, $\alpha_L=-0.25$, $\alpha_H=-0.25$ 
			\item Average number of links for $N=100$ is $1.8$, for $N=250$ it is $4.5$. 
			\item Average skewness for $N=100$ is $0.81$, for $N=250$ it is $0.62$. 
			\item Size is the empirical size of t-test against the truth. 
			\item N$=100$, $corr(a_i,{\bf{x}}_{2i})=-0.001$,N$=250$, $corr(a_i,{\bf{x}}_{2i})=-0.002$ 
		\end{tablenotes} 
	\end{threeparttable} 
\end{table} 
	\begin{table}\caption{\footnotesize {\bf Cross-Validation results}: Parameter values across 1000 Monte Carlo replications for dense network design 4 and Hermite polynomial sieve} \label{table: CV}
\begin{threeparttable} 
\centering \footnotesize
\scalebox{.9}{\begin{tabular}{cccccccccccccc}\toprule 
&\cellcolor{yellow}$N$&\multicolumn{6}{|c|}{\cellcolor{yellow}$100$}&\multicolumn{6}{|c|}{\cellcolor{yellow}$250$}\\\hline 
&&\multicolumn{6}{c}{$K_N$}&\multicolumn{6}{c}{$K_N$}\\\hline 
&$\beta_0-\hat{\beta_0}$&$3$&$4$&$5$&$6$&$7$& $8$ &$3$&$4$&$5$&$6$&$7$& $8$\\\midrule 
\multicolumn{14}{c}{\bf Control function: $\widehat{h}(a_i)$}\\	\midrule 
\multirow{4}{*}{\rotatebox[origin=c]{90}{$\exp(a_i)$}}&mean&1.287 &1.247 &1.279 &1.280 &1.288 &1.264 &1.172 &1.181 &1.188 &1.200 &1.185 &1.181 \\ 
&median&0.576 &0.551 &0.561 &0.562 &0.569 &0.568 &0.530 &0.534 &0.534 &0.538 &0.531 &0.531 \\ 
&std&1.864 &1.813 &1.887 &1.905 &1.889 &1.816 &1.673 &1.691 &1.702 &1.733 &1.698 &1.681 \\ 
&iqr&1.553 &1.499 &1.532 &1.543 &1.546 &1.537 &1.427 &1.436 &1.442 &1.450 &1.442 &1.441 \\
\bottomrule 
\end{tabular}} 
\begin{tablenotes}\tiny 
\item The statistics are based on conventional leave one out cross-validation. 
  \end{tablenotes} 
\end{threeparttable} 
\end{table} 

\section{Conclusions}\label{section: conclusions}
In this paper we show that, whenever the network is likely endogenous, it is important to control for this endogeneity when estimating peer effects.
Failing to control for the endogeneity of the connections matrix in general leads to biased estimates of peer effects. We show that under specific assumptions, we can use the control function approach to deal with the endogeneity problem. We assume that unobserved individual characteristics directly affect link formation and individual outcomes. We leave the functional form through which unobserved individual characteristics enter the outcome equation unspecified and estimate it using a non-parametric approach. The estimators we propose are easy to use in applied work, and Monte Carlo results show that they perform well compared to a linear control function estimator. Erroneously assuming that unobserved characteristics enter the outcome equation in a linear fashion can lead to a serious bias in the estimated parameters. 

{
	\bibliographystyle{chicago}

}
\appendix
\section*{Appendix}
\renewcommand{\thesection}{A.\arabic{section}}
\renewcommand{\thesubsection}{\thesection.\arabic{subsection}}

In this section we introduce the assumptions that are required for the two asymptotic results, Theorem \ref{theorem: central limit theorem} for $\widehat{\beta}_{2SLS}$ and Theorem \ref{theorem: central limit theorem beta_bar} for $\bar{\beta}_{2SLS}$. 
The proof of Theorem \ref{theorem: central limit theorem} is available in the Supplementary Appendix which is available in \cite{JohnssonMoon2019}. 
Since the proof of Theorem \ref{theorem: central limit theorem beta_bar} is similar to that of Theorem \ref{theorem: central limit theorem}, we provide only a sketch of the proof of Theorem \ref{theorem: central limit theorem beta_bar} in the Supplementary Appendix.

\section{Assumptions}\label{appendix: assumptions}

In this section we introduce the assumptions used in the proof of Theorem \ref{theorem: central limit theorem}. First, we introduce a set of sufficient conditions under which we can estimate $a_i$ satisfying the conditions in Assumption \ref{assumption: estimation of a_i}. This assumption corresponds to {\color{black}Assumptions 1, 2, 3 and 5} of \cite{Graham2017}.
\begin{assumption}[Sufficient Conditions for Assumption \ref{assumption: estimation of a_i}] \label{assumption:Graham}(i) $\mathbf{t}_{ij}=\mathbf{t}_{ji}$. (ii) $u_{ij}\sim\ i.i.d.$ for all $ij$ a logistic distribution. (iii)  The supports of $\lambda$, $\mathbf{t}_{ij}$, $a_i$ are compact. 
\end{assumption} 

The next four assumptions are about the sieves used in the semiparametric estimators. The first two are for $\widehat{\beta}_{2SLS}$ and the next two are for $\bar{\beta}_{2SLS}$.

\begin{assumption}[Sieve]\label{assumption:sieve basis}
	For every $K_N$ there is a non-singular matrix of constants $\B$ such that for $\tilde{\q}^{K_N}(a)=\B\q^{K_N}(a)$, we assume the following.
	(i) The smallest eigenvalue of $\mathbb{E}[\tilde{\q}^{K_N}(a_i)\tilde{\q}^{K_N}(a_i)']$ is bounded away from zero uniformly in $K_N$. 
	(ii) There exists a sequence of constants $\zeta_0(K_N)$ that satisfy the condition $\sup_{a \in \mathcal{A}} \|\tilde{\q}^{K_N}(a) \| \leq \zeta_0(K_N)$, 
	where $K_N$ satisfies $\zeta_0(K_N)^2K_N/N\ \to\ 0$ as $N\to\infty$.
	(iii) For  $f(a)$ being an element of $\h(a) = ( E[y_i|a_i =a],\E[\z_i|a_i =a],\E[\w_i|a_i =a ] )$, there exists a sequence of $\boldsymbol{\alpha}^{f}_{K_N}$ and a number $\kappa>0$ such that \[
	\sup_{a \in \mathcal{A}}
	\| f(a)-\q^{K_N}(a)^{\prime}
	\boldsymbol{\alpha}^{f}_{K_N} 
	\|
	=O(K_N^{- \kappa})\] 
	as $K_N\to\infty$. 
	(iv) As $N\ \to\ \infty,$ $K_N\ \to\ \infty$ with $\sqrt{N}K_N^{-\kappa}\to 0$ and $K_N/N\ \to\ 0$.  
\end{assumption}

\begin{assumption}[Lipschitz condition]\label{assumption: Lipschitz condition}
	The sieve basis satisfies the following condition: there exists a positive number $\zeta_1(k)$ such that \[
	\| \q_k(a)-\q_k(a')\|\leq \zeta_1(k)\| a-a'\|\ \forall\ k=1,\ldots, K_N
	\] 
	with $\frac{1}{\zeta_a(N)^2}\sum_{k=1}^{K_N}\zeta_1^2(k)=o(1) $ 
	and $ \zeta_0(K_N)^6 \left( \frac{1}{\zeta_a(N)^2}\sum_{k=1}^{K_N}\zeta_1^2(k) \right) =o(1).$
\end{assumption}

In our paper, we use the following sieves for the Monte Carlo simulations. 
\begin{itemize}
	\item[(i)] Polynomial: For $|a|\leq 1$, define 
	\[
	Pol(K_N)=\left\{\nu_0+\sum_{k=1}^{K_N}\nu_k a^k,\ a\in[-1,1]\, \nu_k \in \mathbb{R}
	\right\}
	\] 
	\item[(ii)] The Hermite Polynomial sieve: For $|a|\leq 1$, define 
	\[
	HPol(K_N)=\left\{\sum_{k=1}^{K_N+1} 
	\nu_k H_k(a)\exp\left(\frac{-a^2}{2}\right), \ a\in[-1,1], \nu_k \in \mathbb{R}
	\right\},
	\]
	where $H_k(a)=(-1)^ke^{a^2}\frac{d^k}{da^k}e^{-a^2}$.\\

\end{itemize}

For the polynomial sieve, it is known that $\zeta_0=O(K_N)$ (e.g., \cite{Newey1997}). Then, since $\zeta_1(k)=O(k)$, $\sum_{k=1}^{K_N}\zeta_1^2(k)= O(K_N^3)$. Hence, the conditions that must be satisfied for the polynomial sieve are $K_N^3/N\ \to\ 0$
and $\sqrt{N}K_N^{-\kappa}\to\ 0$. Further, when $\zeta_a(N)^2 =\frac{N}{\ln N}$, we need 
$\zeta_a(N)^{-2}O(K_N^9)=o(1).$

The next two assumptions are for the sieves used in $\bar{\beta}_{2SLS}$. These assumptions modify Assumption \ref{assumption:sieve basis} and Assumption \ref{assumption: Lipschitz condition}.

\begin{assumption}[Sieve]\label{assumption:sieve basis.alternative}
		For every $K_N$ there is a non-singular matrix of constants $\B$ such that for $\tilde{\mathbf{r}}^{K_N}(\x_{2i},deg_i)=\B\mathbf{r}^{K_N}(\x_{2i},deg_i)$. We assume the following. 
		(i) The smallest eigenvalue of \\ $\mathbb{E}[\tilde{\rr}^{K_N}(\x_{2i},deg_i)\tilde{\rr}^{K_N}(\x_{2i},deg_i)']$ is bounded away from zero uniformly in $K_N$. (ii) There exists a sequence of constants $\zeta_{0**}(K_N)$ that satisfy the condition 
		$\sup_{(\x_{2i},deg_i) \in \mathcal{S}} \|\tilde{\rr}^{K_N}(\x_{2i},deg_i) \| \leq \zeta_{0**}(K_N)$,
		where $K_N$ satisfies $\zeta_{0**}(K_N)^2K_N/N\ \to\ 0$ as $N\to\infty$, and $\mathcal{S}$ is the domain of $(\x_{2i},\text{deg}_i)$. 
		(iii) For  $f(\x_{2i},\text{deg}_i)$ being an element of \\
		$\h_{**}(\x_{2i},\text{deg}_i) = ( \E[y_i|\x_{2i},\text{deg}_i], \E[\z_i|\x_{2i},\text{deg}_i],\E[\w_i|\x_{2i},\text{deg}_i])$, there exists a sequence of $\boldsymbol{\gamma}^{f}_{K_N}$ and a number $\kappa>0$ such that \[
		\sup_{(\x_{2i},\text{deg}_i) \in \mathcal{S}}
		\| f-\rr^{K_N\prime}
		\boldsymbol{\gamma}^{f}_{K_N} 
		\|
		=O(K_N^{- \kappa})\] 
		as $K_N\to\infty$. 
		(iv) As $N\ \to\ \infty,$ $K_N\ \to\ \infty$ with $\sqrt{N}K_N^{-\kappa}\to 0$ and $K_N/N\ \to\ 0$.  
\end{assumption}

Recall from (\ref{lemma: error from deg_hat}) that $\sup_i|\widehat{deg}_i-deg_i|= O(\zeta_{deg}(N)^{-1})$ with $\zeta_{deg}(N) = o(1) N^{\frac{B-1}{2B}}$ for some integer $B \geq 2$.

\begin{assumption}[Lipschitz]\label{assumtion:sieve with x2, Lipschitz}
	
	For $\zeta_{0**}(K_N)$ being the constant from Assumption \ref{assumtion:sieve with x2, Lipschitz}, there exists a positive number $\zeta_{1**}(k)$ such that 
		\[
		\| \mathbf{r}_k(\x_{2i},deg_i)-\mathbf{r}_k(\x_{2i},deg'_i)\|
		\leq \zeta_{1**}(k)\| deg_i-deg'_i\|\ \forall\ k=1,\ldots, K_N
		\] 
		with $\zeta_{deg}(N)^{-2}\sum_{k=1}^{K_N}\zeta_{1**}^2(k)=o(1) $ and $\zeta_{0**}(K_N)^6 \left( \zeta_{deg}(N)^{-2}\sum_{k=1}^{K_N}\zeta_{1**}^2(k) \right) =o(1)$.
\end{assumption}

The next assumptions restrict the models of the outcome in (\ref{model:outcome}) and the network formation of (\ref{model.network.formation}).
We need Assumption \ref{assumption:limit.dist} to derive the limiting distribution of $\widehat{\beta}_{2SLS}$ in Theorem \ref{theorem: central limit theorem}.
\begin{assumption}\label{assumption:limit.dist} We assume the following:
	(i) The true coefficients satisfy $|\beta_1^0| \leq 1 - \epsilon$ and $ \| \beta_2^0 \| \geq  \epsilon$ for some small $\epsilon$. 
	(ii) The parameter set $\mathbb{B}$ for $\beta$ is bounded.
	(iii) The observables $(y_i, \x_{i})$ are bounded. The unobserved characteristic $a_i$ has a compact support in $[-1,1]$. 
	(iv) The network formation error $u_{ij}$ has an unbounded full support $\mathbb{R}$.  
	(v) The net surplus of the network $g(\mathbf{t}_{ij},a_i,a_j)$ is bounded by a finite constant, where $\mathbf{t}_{ij} := t(\x_{2i},\x_{2i})$.
	(vi) The net surplus of the network $g( \mathbf{t}_{ij}, a_i, a_j )$ is a strictly monotonic function of $a_i$ for fixed $(\x_{2i},\x_{2j})$ and $a_j$.  
\end{assumption}

Condition (i) is standard in the linear-in-means peer effect literature. As discussed in the main text, the condition $|\beta_1^0| \leq 1 - \epsilon$ is required for a unique solution of the spillover effect. We need the restriction $ \| \beta_2^0 \| >  \epsilon$ for the IVs to be strong. The boundedness conditions in (ii) and (iii) are important technical assumptions for asymptotics which require some uniform convergence. Also, these conditions imply key regularity conditions for the CLT.
Conditions (vi) and (v) assume that the network is dense \label{dense4} and $ 0 < \underline{\kappa} \leq \E[d_{ij} = 1] \leq \bar{\kappa} < 1$.

Finally, notice that Assumption \ref{assumption:limit.dist}  allows $\upsilon_i - \E(\upsilon_i|a_i)$ to be conditionally heteroskedastic, and so $\sigma^2(\x_i,a_i) := \mathbb{E}[(\upsilon_i - \E[ \upsilon_i|a_i])^2| \x_i,a_i]$ depends on $(\x_i,a_i)$. This is also true for $\upsilon_i - \E(\upsilon_i|a_i)$

\clearpage

\newpage
\thispagestyle{empty}

\section*{{\Large  Supplementary Appendix: Not for Publication}\\
\text{to the paper}	\\
{\large Estimation of Peer Effects in Endogenous Social Networks: \\Control Function Approach (2019)}\\
Ida Johnsson\protect\footnote[*]{Clutter Inc, ida.johnsson@clutter.com} and Roger Moon\protect\footnote[$^\dagger$]{University of Southern California and Yonsei University, moonr@usc.edu (corresponding author)}}

\renewcommand{\thesection}{S.\arabic{section}}
\renewcommand{\thesubsection}{\thesection.\arabic{subsection}}
\renewcommand{\thepage}{S\arabic{page}}
\setcounter{page}{1}
\setcounter{section}{0}
We use the following notation. $M$ denotes a finite generic constant and $a\perp b$ means that $a$ and $b$ are orthogonal to each other. For an $N\times N$ matrix $\A$, we define matrix norms as follows:
$\| \A \|=\left(\sum_{i,j}|a_{ij}|^2\right)^{1/2}$ denotes the Frobenius norm,
$\lVert \A \lVert_{o}$ denotes the operator norm of matrix $\A$, that is, $\lVert \A \lVert_{o} = \lambda_{\max}(\A'\A)^{1/2}$,
$\lambda_{min}(\A)$ denotes the minimum eigenvalue of $\A$. Notice that
\begin{equation}
\left\| \A \right\|_o \leq \left\| \A \right\| \leq \left\| \A \right\|_o {\rm{rank}}(\A). \label{eq.norm.inequality}
\end{equation} 

Further, for matrix $\A$, $[\mathbf{a}]_i$ denotes the $i$'th row of $\A$.
Denote $[\G\X_1]_i$ by $\X_{1,G,i}$, $[\G^2\X_1]_i$ by $\X_{1,G^2,i}$, $[\G\y]^i$ by $\Y_{G,i}$. The $i$th row of the instrument matrix $\Z_N$ is given by 
$\z_i'=[\X'_{2,i}, \X_{1,G,i}, \X_{1,G^2,i}]$, $\z_i$ is $(3l_x)\times 1$. Similarly, 
$\w_i'=[\Y_{G,i}, \X_{1,i}',\X_{1,G,i}]$.  We denote matrices by uppercase bold letters  and vectors by lowercase bold letters,
$\Z_N=(\Z_1',\ldots,\Z_N')'$, 
$\W_N=(\W_1',\ldots,\W_N')'$ and 
$\an=(a_1,\ldots,a_N)'$. 

\section{Estimators}

\subsection{$\hat{\beta}_{2SLS} $}\label{appendix:beta_hat}
Let $\h(a_i)=(h^y(a_i),\h^{w}(a_i),\h^z(a_i))
:=(\mathbb{E}[y_i|a_i], \mathbb{E}[\w_i|a_i],\mathbb{E}[\z_i|a_i]).$

To present the estimator $\hat{\beta}_{2SLS} $ in matrix notation, we let $\widetilde{\W}_N = (\w_1 - \h^{w}(a_1),...,\w_N - \h^{w}(a_N))'$. Similarly we define $\widetilde{\Z}_N, \tilde{\y}_N$.
Suppose that we observe $\h(a_i)$. In view of the identification scheme of Theorem \ref{theorem:identification}, we can estimate $\beta^0$ by
\[
\widehat{\beta}_{2SLS}^{\text{inf}} = \left( \widetilde{\W}_N'\widetilde{\Z}_N 
\left( \widetilde{\Z}_N'\widetilde{\Z}_N \right)^{-1} \widetilde{\Z}_N'\widetilde{\W}_N \right)^{-1}
\widetilde{\W}_N'\widetilde{\Z}_N 
\left( \widetilde{\Z}_N'\widetilde{\Z}_N \right)^{-1} \widetilde{\Z}_N'\tilde{\y}_N.
\]

Let $\q^{K}(a)=(q_1(a),\ldots,q_{K}(a))'$, 
$\Q_N:=\Q_N(\mathbf{a}_N) = (q^{K}(a_1),\ldots,q^{K}(a_N))'$, $\h^l(\an)=(h^l(a_1),\ldots,h^l(a_N))'$, 
and $\boldsymbol{\alpha}^{l}_N=(\alpha_1^l,\ldots,\alpha^l_{K_N})'$. 
Let $b_i^l$ be the $l^{th}$ element in $(y_i,\w_i',\z_i')'$ and denote $\bn^l = (b_1^l,...,b_N^l)$.

If $\an = (a_1,...,a_N)'$ is observed, in view of (\ref{eq.approximation}), we can estimate the unknown function $\h^l(\an)$ by the OLS of $b_i^l$ on $\q^K(a_i)$: for $l=1,...,L$,
\begin{equation}
\widehat{\h}^l(\an) = \PP_{\Q_N} \bn^l, \label{eq.hhat.an}
\end{equation}
where $\PP_{\Q_N} = \Q_N(\Q_N'\Q_N)^{-}\Q_N'$. Here $^-$ denotes  any symmetric generalized inverse.

Given this, we suggest to estimate  $\h^l(\an)$ as follows: (i) first, we estimate the unobserved individual heterogeneity and then (ii) plug the estimate in $\widehat{\h}^l(\an)$ of (\ref{eq.hhat.an}). To be more specific, suppose $\widehat{\mathbf{a}}_N = (\widehat{a}_1,...,\widehat{a}_N)'$ is an estimator of $\an = (a_1,...,a_N)'$. Denote 
$\widehat{\Q}_N := \Q_N(\widehat{\mathbf{a}}_N) = (\q^{K_N}(\widehat{a}_1),\ldots,\q^{K_N}(\widehat{a}_N))'$. Then the first estimator of $\h^l(\an)$ is defined by
\begin{equation}
\widehat{\h}^{l} := \widehat{\h}^l(\widehat{\mathbf{a}}_N) = \PP_{\widehat{\Q}_N} \bn^l \label{eq.hhat1}
\end{equation}
for $l=1,...,L$, and this leads the following estimator of $\beta^0$:
\begin{eqnarray}
\widehat{\beta}_{2SLS} 
&=& \left( \W_N^{\prime}\M_{\widehat{\Q}_N}\Z_N 
\left( \Z_N^{\prime} \M_{\widehat{\Q}_N} \Z_N \right)^{-1} 
\Z_N^{\prime} \M_{\widehat{\Q}_N} \W_N \right)^{-1} \nonumber \\
&&
\times \W_N^{\prime}\M_{\widehat{\Q}_N}\Z_N 
\left( \Z^{\prime} \M_{\widehat{\Q}_N} \Z_N \right)^{-1} \Z_N^{\prime} \M_{\widehat{\Q}_N} \y_N, \label{def.betahat.appendix}
\end{eqnarray}
where $\M_{\widehat{\Q}_N} = I_N - \PP_{\widehat{\Q}_N}$.

\subsection{$\bar{\beta}_{2SLS} $}\label{appendix:beta_bar}

Suppose that the function $h^l_{*}(\x_{2i}, \text{deg}_i), l=1,...,L$ is well approximated by a linear combination of base functions $(r_1(\x_2,\text{deg}_i),...,r_K(\x_2,\text{deg}_i)):$
\[
h^l_{**}(\x_{2i}, \text{deg}_i) \cong \sum_{k=1}^{K_N} r_k(\x_2,\text{deg}_i) \gamma_k^l
\]
as the truncation parameter $K_N \rightarrow \infty$.	

Let $\textbf{Deg}_N = (\text{deg}_1,...,\text{deg}_N)'$. Let $\rr^{K}(\x_{2i},\text{deg}_i) =(r_1(\x_{2i},\text{deg}_i),\ldots,r_{K}(\x_{2i},\text{deg}_i))'$, 
$\R_N:= \R_N(\X_{2N},\textbf{Deg}_N) = (\rr^{K}(\x_{21},\text{deg}_1) ,\ldots, \rr^{K}(\x_{2N},\text{deg}_N))'$,  
and $\boldsymbol{\gamma}^l=(\gamma^l_1,\ldots,\gamma^l_{K_N})'$. 
Let  $\bn^l = (b_1^l,...,b_N^l)$. 
In the case where $(\x_{2i},\text{deg}_i)$ are observed, we can estimate \\ $\h^l_{**}(\X_{2N},\textbf{Deg}_N)=(h^l_{**}(\x_{2,1},\text{deg}_1),...,h^l_{**}(\x_{2,N},\text{deg}_N)$ for $l=1,..,L$ with
\begin{equation}
\widehat{\h}_{**}^l(\X_{2N},\textbf{Deg}_N) := \PP_{\R_N} \bn^l, \label{eq.hbar}
\end{equation}
where $\PP_{\R_N} = \R_N(\R_N'\R_N)^{-}\R_N'$. Here $^-$ denotes  any symmetric generalized inverse.

In view of (\ref{eq.deg.limit}), the natural estimator of $\text{deg}_i$ is $\widehat{\text{deg}}_i$. This suggests that we estimate $\widehat{h}^l_{**}(\x_{2i},\text{deg}_i)$  by using $\widehat{\text{deg}}_i$ in place of $\text{deg}_i$. To be more specific, suppose that $\widehat{\textbf{Deg}}_N=(\widehat{\text{deg}}_1,...,\widehat{\text{deg}}_N )$. Denote $\widehat{\R}_N:= \R_N(\X_{2N},\widehat{\textbf{Deg}}_N) = (\rr^K(\x_{21},\widehat{\text{deg}}_1),..., \rr^K(\x_{2N},\widehat{\text{deg}}_N))'$. The estimator of $h^l_{*}(\x_{2i},a_i) = h^l_{**}(\x_{2i},\text{deg}_i)$ is defined by the $i^{th}$ element of 
\[
\widehat{\h}^{l}_{*}(\X_{2N},\an) := \widehat{\h}^{l}_{**}(\X_{2N},\widehat{\textbf{Deg}}_N) = \PP_{\widehat{\R}_N} \bn^l.
\]
Then, it leads to the following second estimator of $\beta^0$:
\begin{eqnarray}
\bar{\beta}_{2SLS} 
&:=& \left( \W_N^{\prime}\M_{\widehat{\R}_N}\Z_N 
\left( \Z_N^{\prime} \M_{\widehat{\R}_N} \Z_N \right)^{-1} 
\Z_N^{\prime} \M_{\widehat{\R}_N} \W_N \right)^{-1} \nonumber \\
&&
\times \W_N^{\prime}\M_{\widehat{\R}_N}\Z_N 
\left( \Z^{\prime} \M_{\widehat{\R}_N} \Z_N \right)^{-1} \Z_N^{\prime} \M_{\widehat{\R}_N} \y_N, \label{def.betabar}
\end{eqnarray}
where $\M_{\widehat{\R}_N} = I_N - \PP_{\widehat{\R}_N}$.

\section{For $\widehat{\beta}_{2SLS}$}\label{appendix: for beta_hat}


\noindent \textbf{Outline of the proof of Theorem \ref{theorem: central limit theorem}:} 
By definition, we have 
\begin{eqnarray*}
	\widehat{\beta}_{2SLS} - \beta^0 
	&=& \left( \W_N^{\prime}\M_{\widehat{\Q}_N}\Z_N 
	\left( \Z_N^{\prime} \M_{\widehat{\Q}_N} \Z_N \right)^{-1} 
	\Z_N^{\prime} \M_{\widehat{\Q}_N} \W_N \right)^{-1} \\
	&&
	\times \W_N^{\prime}M_{\widehat{\Q}_N}\Z_N 
	\left( \Z^{\prime} \M_{\widehat{\Q}_N} \Z_N \right)^{-1} \Z_N^{\prime} \M_{\widehat{\Q}_N}
	\left(\bm{\eta}^{\upsilon}_N - \h^{\upsilon}(\an) - \widehat{\Q}_N \boldsymbol{\alpha}^{\upsilon}_{K_N} \right).
\end{eqnarray*}
The derivation of the asymptotic distribution of $\widehat{\beta}_{2SLS}$ consists of three steps. 
\begin{itemize}
	\item[Step 1.] First, we control the sampling error coming from the fact that we do not observe $\an$ and approximate it with $\widehat{\mathbf{a}}_N$. Under suitable assumptions (see Appendix \ref{appenxid: error from A_hat-A}), we show that the error that stems from the estimation of $\an$ by $\widehat{\mathbf{a}}_N$ is asymptotically negligible:
	\begin{eqnarray*}
		&&\sqrt{N}\left(\widehat{\beta}_{2SLS} - \beta^0\right) \\
		&=& \left( \frac{1}{N} \W_N^{\prime}\M_{\Q_N}\Z_N 
		\left( \frac{1}{N} \Z_N^{\prime} \M_{\Q_N} \Z_N \right)^{-1} 
		\frac{1}{N} \Z_N^{\prime} \M_{\Q_N} \W_N \right)^{-1} \\
		&&
		\times\frac{1}{N} \W_N^{\prime}\M_{\Q_N}\Z_N 
		\left( \frac{1}{N} \Z_N^{\prime} \M_{\Q_N} \Z_N \right)^{-1} 
		\frac{1}{\sqrt{N}}\Z_N^{\prime} \M_{\Q_N}\bm{\eta}^{\upsilon}_N + o_p(1).
	\end{eqnarray*}
	(See Lemma \ref{lemma: error from A_hat} in Appendix \ref{appenxid: error from A_hat-A} )
	
	\item[Step 2.] Next, we consider the error introduced by the non-parametric estimation of $h(a_i)$.
	Let $\h^{\w}(a_i) = \mathbb{E} (\w_i|a_i), \eta^{\w}_i = \w_i - \h^{\w}(a_i)$, $\h^{\z}(a_i) = \mathbb{E} (\z_i|a_i)$ and $ \eta_i^{\z} = \z_i - \h^{\z}(a_i)$. Let $\widehat{\h}^{\w}(a_i)$ and $\widehat{\h}^{\z}(a_i)$ denote the series approximation of $\h^{\w}(a_i)$ and $\h^{\z}(a_i)$, respectively.
	In Lemma \ref{lemma: series approximation error} in Appendix \ref{appendix: series approximation error} we show  that under the regularity conditions (see Appendix \ref{appendix: series approximation error}),  the error from estimating $h(a_i)$ with $\widehat{h}(a_i)$ converges to zero at a suitable rate and we have
	\begin{eqnarray*}
		\frac{1}{N} \W_N^{\prime}\M_{\Q}\Z_N 
		&=& \frac{1}{N} \sum_{i=1}^{N} 
		\left(\w_i - \widehat{\h}^{\w}(a_i) \right)
		\left(\z_i - \widehat{\h}^{\z}(a_i) \right)'\\
		&=& \frac{1}{N} \sum_{i=1}^{N} 
		\left(\w_i - \h^{\w}(a_i) \right)
		\left(\z_i - \h^{\z}(a_i) \right)' + o_p(1)\\
		\frac{1}{N} \Z_N^{\prime}\M_{\Q}\Z_N 
		&=& \frac{1}{N} \sum_{i=1}^{N} 
		\left(\z_i - \widehat{\h}^{\z}(a_i) \right)
		\left(\z_i - \widehat{\h}^{\z}(a_i) \right)'\\
		&=& \frac{1}{N} \sum_{i=1}^{N} 
		\left(\z_i - \h^{\z}(a_i) \right)
		\left(\z_i - \h^{\z}(a_i) \right)' + o_p(1),\\
		\frac{1}{\sqrt{N}}\Z_N^{\prime} \M_{\Q}\bm{\eta^{v}}_N 
		&=& \frac{1}{\sqrt{N}} \sum_{i=1}^{N}
		\left(\z_i - \widehat{\h}^{\z}(a_i) \right)\eta^{v}_i 
		= \frac{1}{\sqrt{N}} \sum_{i=1}^{N}
		\left(\z_i - \h^{\z}(a_i) \right)\eta^{\upsilon}_i + o_p(1). 
	\end{eqnarray*}
	
	\item[Step 3.]  The consequence of these two approximation is that 
	$ \sqrt{N} (\widehat{\beta}_{2SLS} - \widehat{\beta}_{2SLS}^{\text{inf}}) = o_p(1)$.  Finally in Step 3, we derive the limiting distribution of the infeasible estimator $\sqrt{N}(\widehat{\beta}_{2SLS}^{\text{inf}} - \beta^0)$.
\end{itemize}

\subsection{Controlling the Sampling Error $\widehat{a}_i-a_i$ in Sieve Estimation}\label{appenxid: error from A_hat-A}
In this section, we show that the error coming from the estimation of $a_i$ by $\widehat{a}_i$ is of order $o_p(1)$. All supporting Lemmas can be found in Appendix \ref{section: supporting lemmas}.
\begin{lemma}\label{lemma: error from A_hat} 
	Assume Assumptions \ref{as:basic} \ref{assumption: rank}, \ref{assumption:sieve basis}, \ref{assumption: Lipschitz condition}, and \ref{assumption:limit.dist}. Then the following hold.
\begin{enumerate}
	\item[(a)] $\frac{1}{N}(\Z_N^{\prime}\PP_{\widehat{\Q}_N}\W_N - \Z_N^{\prime}\PP_{\Q_N}\W_N) = o_p(1)$.
	\item[(b)] $\frac{1}{N}(\Z_N^{\prime}\PP_{\widehat{\Q}_N}\Z_N - \Z_N^{\prime}\PP_{\Q_N}\Z_N) = o_p(1)$.
	\item[(c)] $\frac{1}{\sqrt{N}}(\Z_N^{\prime}\PP_{\widehat{\Q}_N}\bm{\eta}^{\upsilon}_N- \Z_N^{\prime}\PP_{\Q_N}\bm{\eta}^{\upsilon}_N) = o_p(1)$.
	\item[(d)] $\frac{1}{\sqrt{N}}(\Z^{\prime}\M_{\widehat{\Q}_N} (\h^{\upsilon}(\an)- \widehat{\Q}_N \boldsymbol{\alpha}^{\upsilon}_{K_N})) = o_p(1)$.
\end{enumerate}
\end{lemma}

\begin{proof}
Part (a).
\begin{eqnarray*}
 &&  \frac{1}{N}(\Z_N^{\prime}\PP_{\widehat{\Q}_N}\W_N - \Z_N^{\prime}\PP_{Q_N}\W_N)\\ 
  &=& \frac{\Z_N'\left({\widehat{\Q}_N} - \Q_N \right)}{N} \left(\frac{{\widehat{\Q}_N}^{\prime}{\widehat{\Q}_N}}{N} \right) ^{-1}
  	\frac{{\widehat{\Q}_N}^{\prime}\W_N}{N} 
   - 	\frac{\Z_N'\Q_N}{N} 
  \left\{\left(\frac{\Q_N^{\prime} \Q_N}{N} \right) ^{-1}
  -\left(\frac{{\widehat{\Q}_N}^{\prime}{\widehat{\Q}_N}}{N} \right) ^{-1}
  \right\}
  \frac{{\Q}_N^{\prime}\W_N}{N} \\
  && + \frac{\Z_N'{\Q}_N }{N} \left(\frac{{\widehat{\Q}_N}^{\prime}{\widehat{\Q}_N}}{N} \right) ^{-1}
  \frac{\left( {\widehat{\Q}_N} - \Q_N \right)^{\prime}\W_N}{N} \\
  &=&
  \frac{\Z_N'\left({\widehat{\Q}_N} - \Q_N \right)}{N} \left(\frac{{\widehat{\Q}_N}^{\prime}{\widehat{\Q}_N}}{N} \right) ^{-1}
  	\frac{({\widehat{\Q}_N}-\Q_N)^{\prime}\W_N}{N}  
  	+ \frac{\Z_N'\left({\widehat{\Q}_N} - \Q_N \right)}{N} \left(\frac{{\widehat{\Q}_N}^{\prime}{\widehat{\Q}_N}}{N} \right) ^{-1}
  	\frac{\Q_N'\W_N}{N} 
  	\\
  && - 	\frac{\Z_N'\Q_N}{N} 
  \left\{\left(\frac{\Q_N^{\prime} \Q_N}{N} \right) ^{-1}
  -\left(\frac{{\widehat{\Q}_N}^{\prime}{\widehat{\Q}_N}}{N} \right) ^{-1}
  \right\}
  \frac{{\Q}_N^{\prime}\W_N}{N} 
   + \frac{\Z_N'{\Q}_N }{N} \left(\frac{{\widehat{\Q}_N}^{\prime}{\widehat{\Q}_N}}{N} \right) ^{-1}
  \frac{\left( {\widehat{\Q}_N} - \Q_N \right)^{\prime}\W_N}{N} \\
  &=& I_1 + I_2 - I_3 + I_4, say.
\end{eqnarray*}
For the desired result, by (\ref{eq.norm.inequality}) we show that 
\begin{equation*}
	\left\| 
	\frac{1}{N}(\Z_N^{\prime}\PP_{\widehat{\Q}_N}\W_N - \Z_N^{\prime}\PP_{\Q_N}\W_N)
	\right\|_o = o_p(1),
\end{equation*}
which follows by triangular inequality if we show
\begin{eqnarray*}
\left\| I_1 \right\|_o, \left\| I_2 \right\|_o, \left\| I_3 \right\|_o, \left\| I_4 \right\|_o = o_p(1).
\end{eqnarray*}

For term $I_1$,
\begin{eqnarray*}
 \left\| I_1 \right\|_o 
 &\leq& \left\| \frac{\Z_N}{\sqrt{N}} \right\| 
 \left\| \frac{\widehat{\Q}_N - \Q_N}{\sqrt{N}} \right\|^2
 \left\| \left( \frac{\widehat{\Q}_N^{\prime}\widehat{\Q}_N}{N} \right)^{-1}\right\|_o  \left\| \frac{\W_N}{\sqrt{N}} \right\| \\
 &=&O_p(1)\left(\frac{1}{\zeta_a(N)^2}  \sum_{k=1}^{K_N} \zeta_1(k)^2 \right)O_P(1)O(1)=o_p(1),  
\end{eqnarray*}
where the last line holds by (\ref{lemma: finite second moments}), Lemmas   \ref{lemma:Q_hat-Q} and \ref{lemma: eigenvalues of Q and Q_hat}, and  by Assumption \ref{assumption: Lipschitz condition}. 

For term $I_2$, 
\begin{eqnarray*}
	\left\| I_2 \right\|_o 
	&\leq& \left\| \frac{\Z_N}{\sqrt{N}} \right\| 
	\left\| \frac{\widehat{\Q}_N - \Q_N}{\sqrt{N}} \right\|
	\left\| \left( \frac{\widehat{\Q}_N^{\prime}\widehat{\Q}_N}{N} \right)^{-1}\right\|_o 
	\left\| \frac{\Q_N}{\sqrt{N}} \right\|  \left\| \frac{\W_N}{\sqrt{N}} \right\| \\
	&=&O_p(1)\left(\frac{1}{\zeta_a(N)^2}  \sum_{k=1}^{K_N} \zeta_1(k)^2 \right)^{1/2}O_P(1) \zeta_0(K_N)O(1)=o_p(1),  
\end{eqnarray*}
where the last line holds by (\ref{lemma: finite second moments}), Lemmas   \ref{lemma:Q_hat-Q} and \ref{lemma: eigenvalues of Q and Q_hat}, and  by Assumption \ref{assumption: Lipschitz condition}.

For term $I_3$, write 
\begin{eqnarray*}
I_3 &=& \frac{\Z_N'\Q_N}{N} 
\left(\frac{{\widehat{\Q}_N}^{\prime}{\widehat{\Q}_N}}{N} \right) ^{-1}
\left\{
\left(\frac{{\widehat{\Q}_N}^{\prime}{\widehat{\Q}_N}}{N} \right)
- \left(\frac{\Q_N^{\prime} \Q_N}{N} \right)
\right\}
\left(\frac{\Q_N^{\prime} \Q_N}{N} \right) ^{-1}
\frac{{\Q}_N^{\prime}\W_N}{N}  \\
&=& 
\frac{\Z_N'\Q_N}{N} 
\left(\frac{{\widehat{\Q}_N}^{\prime}{\widehat{\Q}_N}}{N} \right) ^{-1}
\left(\frac{\widehat{\Q}_N^{\prime}(\widehat{\Q}_N-\Q_N)}{N} \right)
\left(\frac{\Q_N^{\prime} \Q_N}{N} \right) ^{-1}
\frac{{\Q}_N^{\prime}\W_N}{N} \\
&&+ \frac{\Z_N'\Q_N}{N} 
\left(\frac{{\widehat{\Q}_N}^{\prime}{\widehat{\Q}_N}}{N} \right) ^{-1}
\left(\frac{(\widehat{\Q}_N-\Q_N)^{\prime}\Q_N}{N} \right)
\left(\frac{\Q_N^{\prime} \Q_N}{N} \right) ^{-1}
\frac{{\Q}_N^{\prime}\W_N}{N}. 
\end{eqnarray*}
Then, 
\[
\left\| I_3 \right\|_o \leq O_p(1) \zeta_0(K_N)O_p(1) \zeta_0(K_N) \left( \frac{1}{\zeta_a(N)^2}  \sum_{k=1}^{K_N} \zeta_1(k)^2 \right)^{1/2} O_p(1) \zeta_0(K_N) O_p(1)
= o_p(1),
\]
where the last equality follows by Assumption \ref{assumption: Lipschitz condition}.

The desired result of term $I_4$ follows by similar argument used for term $I_2$. \\

Part (b) can be shown in a similar way as Part (a).\\

Part (c).
\begin{eqnarray*}
	&& \frac{1}{\sqrt{N}}(\Z_N^{\prime}\PP_{\widehat{\Q}_N}\bm{\eta}^{\upsilon}_N - \Z_N^{\prime}\PP_{Q_N}\bm{\eta}^{\upsilon}_N) \\
	  &=&
  \frac{\Z_N'\left({\widehat{\Q}_N} - \Q_N \right)}{N} \left(\frac{{\widehat{\Q}_N}^{\prime}{\widehat{\Q}_N}}{N} \right) ^{-1}
  	\frac{({\widehat{\Q}_N}-\Q_N)^{\prime}\bm{\eta}^{\upsilon}_N}{\sqrt{N}} 
  	+ \frac{\Z_N'\left({\widehat{\Q}_N} - \Q_N \right)}{N} \left(\frac{{\widehat{\Q}_N}^{\prime}{\widehat{\Q}_N}}{N} \right) ^{-1}\frac{\Q_N'\bm{\eta}^{\upsilon}_N}{\sqrt{N}} \\
  && - 	\frac{\Z_N'\Q_N}{N} 
  \left\{\left(\frac{\Q_N^{\prime} \Q_N}{N} \right) ^{-1} 
   -\left(\frac{{\widehat{\Q}_N}^{\prime}{\widehat{\Q}_N}}{N} \right) ^{-1}
  \right\}
  \frac{{\Q}_N^{\prime}\bm{\eta}^{\upsilon}_N}{\sqrt{N}} 
  + \frac{\Z_N'{\Q}_N }{N} \left(\frac{{\widehat{\Q}_N}^{\prime}{\widehat{\Q}_N}}{N} \right) ^{-1}
  \frac{\left( {\widehat{\Q}_N} - \Q_N \right)^{\prime}\bm{\eta}^{\upsilon}_N}{\sqrt{N}} \\
  &=& III_1 + III_2 - III_3 + III_4, say,
\end{eqnarray*}
and the desired result of Part (c) follows if we show that for $j=1,...,4$,
\[
\left\| III_j \right\| = o_p(1). 
\]
First, for term $III_1$, we have
\begin{eqnarray*}
\| III_1 \| &\leq&  
\left\| \frac{\Z_N}{\sqrt{N}}\right\|
\left\| \frac{\widehat{\Q}_N-\Q_N}{\sqrt{N}}\right\|
\left\| \left(\frac{{\widehat{\Q}_N}^{\prime}{\widehat{\Q}_N}}{N} \right) ^{-1} \right\|
\left\| \frac{({\widehat{\Q}_N}-\Q_N)^{\prime}\bm{\eta}^{\upsilon}_N}{\sqrt{N}}  \right\| \\
&=& O_p(1) \left( \frac{1}{\zeta_a(N)^2}  \sum_{k=1}^{K_N} \zeta_1(k)^2 \right)^{1/2} O_p(1) \left\| \frac{({\widehat{\Q}_N}-\Q_N)^{\prime}\bm{\eta}^{\upsilon}_N}{\sqrt{N}}  \right\|,
\end{eqnarray*}
where the last line holds by (\ref{lemma: finite second moments}), Lemmas \ref{lemma:Q_hat-Q} and \ref{lemma: eigenvalues of Q and Q_hat}.
Under Assumption 
we can show that 
\[
\mathbb{E} \left[  \left\| \frac{({\widehat{\Q}_N}-\Q_N)^{\prime}\bm{\eta}^{\upsilon}_N}{\sqrt{N}}  \right\|^2 \: \bigg| \: \X_{1N},\G_N,\an \right] = \frac{1}{N} \left\|\widehat{\Q}_N-\Q_N \right\|^2.
\]
Then, by Lemma \ref{lemma:Q_hat-Q} and Assumption \ref{assumption: Lipschitz condition}, we have the required result for term $III_1$.

The rest of the required results follow by similar fashion and we omit the proof.

Part (d).

Notice that 
\begin{eqnarray*}
&&	\frac{1}{\sqrt{N}}(\Z_N^{\prime}\M_{\widehat{Q}_N} (\h^{\upsilon}(\an)- \widehat{\Q}_N \boldsymbol{\alpha}^{\upsilon}_{K_N})) \\
&=& \frac{1}{\sqrt{N}}\Z_N^{\prime}\M_{\widehat{\Q}_N} \h^{\upsilon}(\an) \\
&=& \frac{1}{\sqrt{N}}\Z_N^{\prime} \left( \M_{\widehat{\Q}_N} - \M_{\Q_N} \right) \h^{\upsilon}(\an)
+ \frac{1}{\sqrt{N}}\Z_N^{\prime} \M_{\Q_N} \left( \h^{\upsilon}(\an) - \Q_N \boldsymbol{\alpha}^{\upsilon}_{K_N} \right) \\
&=& IV_1 + IV_2, say. 
\end{eqnarray*}

We can show $IV_1 =o_p(1)$ by applying similar arguments used in the proof of Part (a).

For term $IV_2$, notice that
\begin{eqnarray*}
	\| IV_2 \| &=& \| IV_2 \|_o \\
	&\leq& 
	\left\lVert \frac{1}{\sqrt{N}}\Z_N \right\lVert_o 
	\left\lVert \M_{\Q_N} \right\lVert_o 
	\left\lVert \h^{\upsilon}(\an) - \Q_N \boldsymbol{\alpha}^{\upsilon}_{K_N} \right \lVert_o
	 \\
	&=&  
	\left\lVert \frac{1}{\sqrt{N}}\Z_N \right\lVert 
	\left\lVert \h^{\upsilon}(\an) - \Q_N \boldsymbol{\alpha}^{\upsilon}_{K_N} \right \lVert \\
	&=& O_p(1)\sqrt{N}O(K_N^{-\kappa})=o_p(1)
\end{eqnarray*} 
by Assumption \ref{assumption:sieve basis} (iii) and (iv).

\end{proof}

\subsubsection{Supporting Lemmas}\label{section: supporting lemmas}

First notice that by the boundedness condition (ii) and (iii) in Assumption \ref{assumption:limit.dist}, we have
\begin{equation}
\frac{1}{N} \| \Z_N \|^2 = O_p(1), \; 
\frac{1}{N} \| \W_N \|^2 = O_p(1). \label{lemma: finite second moments}
\end{equation}
\begin{lemma}\label{lemma: Q^2/N} Under Assumption \ref{assumption:sieve basis}, we have
\[
			\frac{1}{N} \| \Q_N \|^2 \leq M \zeta_0^2(K_N). 
			\label{eq.norm.1/nQ} 
\]
\end{lemma}

\begin{proof}
\[
 \frac{1}{N}\|  \Q_N \|^2
=\frac{1}{N}\sum_{i=1}^N \| \q^K(a_i)\|^2\leq \sup_i \| \q^K(a_i)\| ^2=\zeta_0^2(K_N)
\]
	by Assumption \ref{assumption:sieve basis} (ii).
\end{proof}

\begin{lemma}\label{lemma:Q_hat-Q} Under Assumptions \ref{as:basic}, \ref{assumption: estimation of a_i}, \ref{assumption:sieve basis}, and \ref{assumption: Lipschitz condition}, we have
		\begin{equation*}
		\frac{1}{N} \| \widehat{\Q}_N - \Q_N \|^2 = M \frac{1}{\zeta_a(N)^2}  \sum_{k=1}^{K_N} \zeta_1(k)^2. 
		\end{equation*}		
\end{lemma}
\begin{proof}

		\begin{eqnarray*}
			\frac{1}{N} \| \widehat{\Q}_N - \Q_N \|^2 
			&=& \frac{1}{N} \sum_{i=1}^N\sum_{k=1}^{K_N} 
			\| q_k(\widehat{a}_i) - q_k(a_i) \|^2 
			\leq  \frac{1}{N} \sum_{i=1}^N\sum_{k=1}^{K_N} 
			 \zeta_1(k)^2 \| \widehat{a}_i - a_i\|^2 \\
			&\leq&  \frac{1}{N} \sum_{i=1}^N\sum_{k=1}^{K_N} 
			 \zeta_1(k)^2 \frac{1}{\zeta_a(N)^2} 
			= \frac{1}{\zeta_a(N)^2}  \sum_{k=1}^{K_N} \zeta_1(k)^2,
 		\end{eqnarray*}
 		where the first inequality follows from Assumption \ref{assumption: Lipschitz condition} and the second inequality follows from Assumption \ref{assumption: estimation of a_i}. 
\end{proof}

\begin{lemma}\label{lemma: lambda_A-lambda_B} For symmetric matrices $\A$ and $\mathbf{B}$ it is true that 
	\[
	|\lambda_{min}(\A)-\lambda_{min}(\mathbf{B})|\leq \| \A-\mathbf{B}\|
	\] 
\end{lemma}
\begin{proof}
	Let $\underline{\x}_A$ be the eigenvector associated with the minimum eigenvalue of $\A$. Define $\underline{\x}_B$ analogously. 
	First we show $ |\lambda_{min}(\A)-\lambda_{min}(\mathbf{B})| \leq \| \A-\mathbf{B} \|.$
	\[
	\begin{split}
	\lambda_{min}(\A)-\lambda_{min}(\mathbf{B})
	&=\underline{\x}_A' \A\underline{\x}_A-\underline{x}_B'\mathbf{B}\underline{\x}_B \\
	&\leq \underline{\x}_B'(\A-\mathbf{B})\underline{\x}_B \\
	&\leq |\underline{\x}_B'(\A-\mathbf{B})\underline{\x}_B| \leq \| \A-\mathbf{B} \|. 
	\end{split}
	\]
	Also, we can prove the other direction. Notice that
	\[
	\begin{split}
	\lambda_{min}(\A)-\lambda_{min}(B)
	&= \underline{\x}_A' \A\underline{x}_A-\underline{\x}_B'\mathbf{B}\underline{\x}_B \\
	&\geq \underline{\x}_A'(\A-\mathbf{B})\underline{\x}_A \\
	&\geq -|\underline{\x}_A'(\A-\mathbf{B})\underline{\x}_A| \geq -\| \A-\mathbf{B} \|. 
	\end{split}
	\]
	
	Then, we have the required result.

\end{proof}

\begin{lemma}\label{lemma: eigenvalues of Q and Q_hat}
Under \ref{as:basic}, \ref{assumption: estimation of a_i}, \ref{assumption:sieve basis}, and \ref{assumption: Lipschitz condition},  W.p.a.1, there exists a positive constant $C > 0$ such that 
\begin{equation*}
\frac{1}{C} \leq \lambda_{min}\left(\frac{\Q_N^{\prime} \Q_N}{N}\right), \; \lambda_{min}\left(\frac{\widehat{\Q}_N^{\prime} \widehat{\Q}_N}{N}\right). \label{desired.eigenvalue.Phi'Phi}
\end{equation*}
\end{lemma}
	\begin{proof} First we show that there exists a positive constant $C$ such that $\frac{1}{C} \leq \lambda_{min}\left(\frac{\Q_N^{\prime} \Q_N}{N}\right),$ which follows by Assumption \ref{assumption:sieve basis}(i) if we show
		\[
		\left| \lambda_{min}\left(\frac{\Q_N^{\prime} \Q_N}{N}\right) - \mathbb{E}[\q^{K_N}(a_i)\q^{K_N}(a_i)']  \right| = o_p(1).
		\]
For this, by Lemma \ref{lemma: lambda_A-lambda_B}, we have
\begin{eqnarray*}
	&&\left| 
	\lambda_{min}\left(\frac{\Q_N^{\prime} \Q_N}{N}\right) - \mathbb{E}[\q^{K_N}(a_i)\q^{K_N}(a_i)']  
	\right|
	\leq \left\| 
	\frac{\Q_N^{\prime} \Q_N}{N} - \mathbb{E}[\q^{K_N}(a_i)\q^{K_N}(a_i)']  
	\right\| \\
	&=& \left\| 
	\frac{1}{N} \sum_{i=1}^N \left( \q^{K_N}(a_i)\q^{K_N}(a_i)' - \mathbb{E}[\q^{K_N}(a_i)\q^{K_N}(a_i)'] \right)  
	\right\|.
\end{eqnarray*}
Then, by Assumption \ref{assumption:sieve basis}(ii), we have
\begin{eqnarray*}
&& \mathbb{E} \left\| 
 \frac{1}{N} \sum_{i=1}^N \left( \q^{K_N}(a_i)\q^{K_N}(a_i)' - \mathbb{E}[\q^{K_N}(a_i)\q^{K_N}(a_i)'] \right)  
\right\|^2  \\
&=& \sum_{k=1}^{K_N} \sum_{l=1}^{K_N} \mathbb{E} 
\left(
\frac{1}{N} \sum_{i=1}^N \left(\q_k(a_i)\q_l(a_i) - \mathbb{E}[\q_k(a_i)\q_l(a_i)] \right)
\right)^2 \\
&\leq& \frac{1}{N}\sum_{k=1}^{K_N} \sum_{l=1}^{K_N} \mathbb{E} [\q_k(a_i)\q_l(a_i)]^2 \leq \frac{1}{N} \sup_a \left( \sum_{k=1}^{K_N} \q_k(a)^2 \right)^2 \\
&\leq& \frac{\zeta_0(K_N)^4}{N} =o(1), 
\end{eqnarray*}	 
where the last line holds by Assumptions \ref{assumption:sieve basis}(ii) and  \ref{assumption: Lipschitz condition}.

Next, given the first part of the lemma, the second claim of the lemma follows if we show
\[
\left| \lambda_{min} \left(\frac{\widehat{\Q}_N^{\prime} \widehat{\Q}_N}{N}\right) - \lambda_{min}\left(\frac{\Q_N^{\prime} \Q_N}{N}\right)  \right| = o_p(1).
\]
Notice by Lemma \ref{lemma: lambda_A-lambda_B}, for symmetric matrices $\A$ and $\mathbf{B}$, we have  
$$ | \lambda_{min}(\A)-\lambda_{min}(\mathbf{B}) | \leq \| \A-\mathbf{B}\|.$$ 
Then,
\begin{eqnarray*}
\left| \lambda_{min}\left(\frac{\widehat{\Q}_N^{\prime} \widehat{\Q}_N}{N}\right) 
-\lambda_{min} \left(\frac{\Q_N^{\prime} \Q_N}{N}\right)  \right|  
&\leq & \left\| \frac{\widehat{\Q}_N^{\prime} \widehat{\Q}_N}{N} - \frac{\Q_N^{\prime} \Q_N}{N} \right\|  \\
&\leq & \left\| \frac{(\widehat{\Q}_N-\Q_N)'}{\sqrt{N}}\frac{\Q_N}{\sqrt{N}}
\right\|
+ \left\| \frac{\Q_N'}{\sqrt{N}}\frac{(\widehat{\Q}_N-\Q_N)}{\sqrt{N}} \right\| \\ 
& +&
\left\|\frac{(\widehat{\Q}_N-\Q_N)'}{\sqrt{N}}\frac{(\widehat{\Q}_N-\Q_N)}{\sqrt{N}} \right\|.
\end{eqnarray*}
Then, by lemmas \ref{lemma: Q^2/N} and \ref{lemma:Q_hat-Q} and by Assumption \ref{assumption: Lipschitz condition}, we have 
\[
\left| \lambda_{min} \left(\frac{\widehat{\Q}_N^{\prime} \widehat{\Q}_N}{N}\right) - \lambda_{min}\left(\frac{\Q_N^{\prime} \Q_N}{N}\right)  \right| 
\leq M \left(\zeta_0(K_N)
			\sqrt{\frac{1}{\zeta_a(N)^2}\sum_{k=1}^{K_N}\zeta_1^2(k)}
	+ 
			\frac{1}{\zeta_a(N)^2}\sum_{k=1}^{K_N}\zeta_1^2(k)\right)=o_p(1),
\]
as desired.
\end{proof}

\subsection{Controlling the Series Approximation Error}\label{appendix: series approximation error}

\begin{lemma}[Series Approximation]\label{lemma: series approximation error} Assume the assumptions in Lemma \ref{lemma: error from A_hat}. Then, we have
\begin{itemize}
	\item[(a)] $\frac{1}{N} \sum_{i=1}^{N} 
	\left(\w_i - \widehat{\h}^\w(a_i) \right)
	\left(\z_i - \widehat{\h}^\z(a_i) \right)'
	= \frac{1}{N} \sum_{i=1}^{N} 
	\left(\w_i - \h^\w(a_i) \right)
	\left(\z_i - \h^\z(a_i) \right)'+ o_p(1)$,
	\item[(b)] $\frac{1}{N} \sum_{i=1}^{N} 
	\left(\z_i - \widehat{\h}^\z(a_i) \right)
	\left(\z_i - \widehat{\h}^\z(a_i) \right)'
	= \frac{1}{N} \sum_{i=1}^{N} 
	\left(\z_i - \h^\z(a_i) \right)
	\left(\z_i - \h^\z(a_i) \right)' +o_p(1)$,
	\item[(c)] $\frac{1}{\sqrt{N}} \sum_{i=1}^{N}
	\left(\z_i - \widehat{\h}^\z(a_i) \right)\eta^{\upsilon}_i 
	= \frac{1}{\sqrt{N}} \sum_{i=1}^{N}
	\left(\z_i - \h^\z(a_i) \right)\eta^{\upsilon}_i+ o_p(1) $.
\end{itemize}
\end{lemma}
\begin{proof}
Lemma \ref{lemma: series approximation error} follows if we show
\begin{itemize}
	\item[(i)] $ \frac{1}{N} \sum_{i=1}^{N} \left( \widehat{\h}^\w(a_i) - \h^\w(a_i) \right) \left( \widehat{\h}^\w(a_i) - \h^\w(a_i) \right)^\prime= o_p(1).$
	\item[(ii)] $ \frac{1}{N} \sum_{i=1}^{N} \left( \widehat{\h}^\z(a_i) - \h^\z(a_i) \right) \left( \widehat{\h}^\z(a_i) - \h^\z(a_i) \right)^\prime = o_p(1)$.
	\item[(iii)] $ \frac{1}{\sqrt{N}} \sum_{i=1}^{N} \left( \widehat{\h}^\z(a_i) - \h^\z(a_i) \right) \eta^{\upsilon}_i = o_p(1).$
\end{itemize}	

Lemma \ref{lemma: series approximation error} (i) and (ii) is true by Lemma \ref{lemma.series est error} and Lemma \ref{lemma: series approximation error} (iii) follows from (ii). See the remainder of this section.\end{proof}

Following \cite{Newey1997}, we assume $\B=\I$ in Assumption \ref{assumption:sieve basis}, hence, $\tilde{q}^K(a)=q^K(a)$. \\
Also, we assume $P = \E[\q^K(a_i)(\q^K(a_i))']=I$.\footnote{The Lemmas in this section follow Section 15.6 in \citet{Li2008a}.}\\

\begin{lemma}\label{lemma.QQ/N} Assume Assumption \ref{assumption:sieve basis}. Then, $\E[\| \tilde{\PP}-\I\|^2]=O(\zeta_0(K_N)^2K_N/N)$, where $\tilde{\PP}=(\Q_N'\Q_N)/N$.
\end{lemma}
\begin{proof}
	For proof see \cite{Li2008a} page 481.
\end{proof}

Note that Lemmas \ref{lemma: lambda_A-lambda_B} and \ref{lemma.QQ/N} imply that 
\[ 
 | \lambda_{\min} ( \tilde{\PP} ) - 1 | \leq \| \tilde{\PP}-\I\| =O_p(\zeta_0(K_N)\sqrt{K_N/N}=o_p(1).
\] 
That is, the smallest eigenvalue of $\tilde{\PP}$ converges to one in probability. Letting $\mathbf{1}_N$ be the indicator function for the smallest eigenvalue of $\tilde{\PP}$ being greater than $1/2$, we have $\Pr(\mathbf{1}_N=1)\to 1$.

\begin{lemma}\label{lemma.sieve.coefficient} Assume Assumption \ref{assumption:sieve basis}. Then, $\| \tilde{\alpha}^{f}-\alpha
^{f}\|=O_p(K_N^{-\kappa})$, where $\tilde{\alpha}^{f}=(\Q_N'\Q_N)^{-1}\Q_N'f$, where $\alpha^{(f)}$ satisfies Assumption \ref{assumption:sieve basis} and $f(a) \in \{ h^{y}(a),h^\z(a),h^\w(a) \}$.
\end{lemma}
\begin{proof}\begin{eqnarray*}
\mathbf{1}_N\| \tilde{\alpha}^{(f)}-\alpha
^{(f)}\|&=& \mathbf{1}_N\| (\Q_N'\Q_N)^{-1}\Q_N'(f-\Q_N \alpha^{f})\|\\
&=& \mathbf{1}_N\{(f-\Q_N\alpha^{f})'\Q_N(\Q_N'\Q_N)^{-1}(\Q_N'\Q_N/N)^{-1}\Q_N'(f-\Q_N\alpha^{f})/N\}^{1/2}\\
&=&\mathbf{1}_NO_P(1)\{(f-\Q_N\alpha^{(f)})'\Q_N(\Q_N'\Q_N)^{-1}\Q_N'(f-\Q_N\alpha^{f})/N\}^{1/2}\\
&\leq & O_p(1)\{(f-\Q_N\alpha^{f})'(f-\Q_N\alpha^{f})/N\}^{1/2}=O_p(K_N^{-\kappa})
\end{eqnarray*}
by Lemma \ref{lemma.QQ/N}, Assumption \ref{assumption:sieve basis}(iii), the fact that $\Q_N(\Q_N'\Q_N)^{-1}\Q_N'$ is idempotent and $\Pr(\mathbf{1}_N=1)\to 1$.
\end{proof}


\begin{lemma}\label{lemma.series est error} Assume Assumption \ref{assumption:sieve basis}. Let $f(a) \in ( h^{y}(a),\h^z{\z}(a),\h^{\w}(a))$ and $\tilde{f}=\Q_N\tilde{ \boldsymbol{\alpha}}_N^{f}$. Then, 
$\frac{1}{N}\| f-\tilde{f}\| ^2 = O_p(K_N^{-2\kappa})=o_p(N^{-1/2})$.
\end{lemma}
\begin{proof}
The required result for the lemma follows because 
\begin{eqnarray*}
\frac{1}{N}\| f-\tilde{f}\| ^2  
&\leq& \frac{1}{N}\{\| f-\Q_N\boldsymbol{\alpha}_N^{f} \| ^2 + \| \Q_N(\boldsymbol{\alpha}_N^{(f)}-\tilde{\boldsymbol{\alpha}}_N^{f})\|^2\}\\
&=&O(K_N^{-2\kappa})+(\boldsymbol{\alpha}_N^{f}-\tilde{\boldsymbol{\alpha}}_N^{f})'(\Q_N'\Q_N/N)(\boldsymbol{\alpha}_N^{f}-\tilde{\boldsymbol{\alpha}}_N^{f})\\
&=&O(K_N^{-2\kappa})+O_p(1)\| \boldsymbol{\alpha}_N^{f}-\tilde{\boldsymbol{\alpha}}_N^{f} \|^2= O_p(K_N^{-2\kappa})
\end{eqnarray*}
by Assumption \ref{assumption:sieve basis}(iii), Lemma \ref{lemma.QQ/N} and Lemma \ref{lemma.sieve.coefficient}.
\end{proof}

\subsection{Limiting Distribution of $\widehat{\beta}_{2SLS}$}\label{appendix: distribution of est}
In this section we derive the distribution of the infeasible estimator $\widehat{\beta}^{inf}_{2SLS}$.
All supporting lemmas can be found in Section \ref{appendix: further supporting lemmas for beta_hat}.

We introduce the following notation.
Let $s_0(\x_i,a_i)$ be a function of $(\x_i,a_i)$ such that $s_0(\cdot,\cdot)$ is bounded over the support of $(\x_i,a_i)$. We denote an $N$ vector-valued function that stacks $s_0(\x_i,a_i)$ over $i=1,...,N$ as
$\Ss_{0,N} = (s_0(\x_1,a_1),\ldots, s_0(\x_N,a_N))'.$  
Define
\begin{equation}
s_{0,N,i} := s_0(\x_i,a_i). \label{def.s 0 N i}
\end{equation}
Next, for $m=1,2,...,$ we define recursively 
\begin{align}
s_{m,N,i} &:= \sum_{j=1, \neq i}^N g_{ij} s_{m-1,N,i} = [\G_N \Ss_{m-1,N}]_i, \label{def.s m N i}
\end{align}
where 
\[ 
\Ss_{m-1,N}:=(s_{m-1,N,1},\ldots, s_{m-1,N,N})'.
\]
For $m=0,1,2,...$, we define $s^{\x_1}_{m,N,i}$ and $\Ss^{\x_1}_{m,N}$  with initial function $s_{0,N,i}=s_{0}(\x_i ,a_i) = \x_{1i}$, and define $s^{a}_{m,N,i}$ and $\Ss^{a}_{m,N}$ with initial function $s_{0,N,i}=s_0(\x_i,a_i) = h^{\upsilon}(a_i)$.

Next, we define recursively the probability limit of $s_{m,N,i}$ defined with the initial function $s_{0,N,i}=s_0(\x_i,a_i)$ for each $i$ as $N \rightarrow \infty$. For this, let
\[
\tilde{s}_0(\x_i,a_i) = s_0(\x_i,a_i) = s_{0,N,i}.
\]
Note that for fixed $i$, $s_{1,N,i}$ has the following limit as $N \rightarrow \infty$: 
\begin{align}
s_{1,N,i} &=
[\G_N \Ss_{0,N}]_i \nonumber \\ 
&= \left( \frac{1}{N} \sum_{j \neq i} d_{ij} \right)^{-1}
\frac{1}{N} \sum_{j \neq i} d_{ij}s_{0}(\x_j,a_j) \nonumber \\
&= \left( \frac{1}{N} \sum_{j \neq i} 
\mathbb{I} \left\{ g( t(\x_{2i},\x_{2j}), a_i, a_j )  \geq u_{ij} \right\}
\right)^{-1} \nonumber \\
& \quad \times \frac{1}{N} \sum_{j \neq i} 
\mathbb{I} \left\{ g( t(\x_{2i},\x_{2j}), a_i, a_j )  \geq u_{ij} \right\}
s_{0}(\x_j,a_j)  \nonumber \\ 
& \xrightarrow[]{p}
\frac{\int\int\int p(g( t(\x_{2i},\x_2), a_i, a)s_{0}(\x,a)\pi(\x,a)d\x da}
{ \int\int p(g( t(\x_{2i},\x_2), a_i, a)\pi(\x_2,a)d\x_2 da} \nonumber \\
&=\frac{\E[d_{ij} s_{0}(\x_j,a_j) | \x_i,a_i  ]}{\E[d_{ij} | \x_i,a_i  ]}
=:\tilde{s}_{1}(\x_i,a_i), \label{def.tilde.s1}
\end{align}
where $\pi(\x,a)$ with $\x=(\x_1,\x_2)$ is the joint density of $\x_i = (\x_{1i},\x_{2i})$ and $a_i$, and $\pi(\x_2,a)$ is the joint density of $(\x_{2i},a_i)$. 
Here note that the limit $\tilde{s}_{1}(\x_i,a_i)$ depends only on $(\x_i,a_i)$, not on $(\x_{-i},a_{-i})$, while $s_{1,N,i}$ depends on both $(\x_i,a_i )$ and $(\x_{-i},a_{-i})$.

We define the following recursively for $m=2,3,\cdots$ as follows:
\begin{align}
\tilde{s}_{m}(\x_i,a_i) 
&:= \frac{\E[d_{ij} \tilde{s}_{m-1}(\x_j,a_j) | \x_i,a_i  ]}{\E[d_{ij} | \x_i,a_i  ]} \label{def.tilde.sm} \\
&= \frac{\int\int p( g( t(\x_{2i},\x_2), a_i, a)\tilde{s}_{m-1}(\x,a)\pi(\x,a)d\x da}
{ \int\int p( t(\x_{2i},\x_2), a_i, a)\pi(\x_2,a)d\x_2da} \nonumber \\
&= \plim_{N \rightarrow \infty}
\left( \frac{1}{N} \sum_{j \neq i} d_{ij} \right)^{-1}
\frac{1}{N} \sum_{j \neq i} d_{ij}\tilde{s}_{m-1}(\x_j,a_j) \nonumber \\
& = \plim_{N \rightarrow \infty} [\G_N \tilde{\Ss}_{m-1}]_i, \nonumber
\end{align}
where $\tilde{\Ss}_{m} = (\tilde{s}_{m}(\x_1,a_1),...,\tilde{s}_{m}(\x_N,a_N)).$


Using this general definitions of (\ref{def.tilde.s1}) and (\ref{def.tilde.sm}), with $\tilde{s}_0^{\x_1}(\x_i,a_i) = s^{\x_1}_0(\x_i,a_i) = \x_{1i}$ and $\tilde{s}_0^{a}(\x_i,a_i) = s^{a}_0(\x_i,a_i) = h(a_i)$, we define $\tilde{s}^{\x_1}_m(\x_i,a_i)$ and $\tilde{s}_m^{a}(\x_i,a_i)$, respectively, for $m=1,2,...$.
Let $\tilde{S}^{\x_1}_m =
(\tilde{s}^{\x_1}_{m}(\x_1,a_1),\ldots\tilde{s}^{\x_1}_{m}(\x_N,a_N))'.
$ and $ \tilde{S}^{a}_m = (\tilde{s}^{a}_{m}(\x_1,a_1),\ldots\tilde{s}^{a}_{m}(\x_N,a_N))' $.\\

Next, with the initial function $s^{\upsilon}_{0,N,i} = \eta_i^{\upsilon}$ and $\Ss_{0,N}^{\upsilon} := (s^{\eta}_{0,N,1},\ldots,s^{\eta}_{0,N,N})'$,  we define recursively
\begin{equation}
s^{\upsilon}_{m,N,i}: = [\G_N \Ss^{\upsilon}_{m-1,N}]_i = \sum_{j=1, \neq i}^N g_{ij} s^{\upsilon}_{m-1,N,i}, 
\label{def.s_eta_m N i}
\end{equation}
and $ \Ss^{\upsilon}_{m,N}:=(s^{\upsilon}_{m,N,1},\ldots, s^{\upsilon}_{m,N,N})' $ for $m=1,2,...$.

\begin{lemma}\label{lemma: limit of S^ZZ and S^WZ}  Under Assumptions \ref{as:basic} and \ref{assumption:limit.dist}, as $N \rightarrow \infty$, we have  
\begin{align*}
(a) & \frac{1}{N}\sum_{i=1}^N(\w_i-\h^{\w}(a_i))
(\z_i-\h^{\z}(a_i))' \\
&=: \begin{pmatrix}
\frac{1}{N}\sum_{i=1}^N
\bm{\eta}_i^{GY} (\bm{\eta}_i^{\x_1})'
&\frac{1}{N}\sum_{i=1}^N\bm{\eta}_i^{GY}(\bm{\eta}_i^{G \x_1})'
&\frac{1}{N}\sum_{i=1}^N\bm{\eta}_i^{GY}(\bm{\eta}_i^{G^2 \x_1})' \\
\frac{1}{N}\sum_{i=1}^N\bm{\eta}_i^{ \x_1}(\bm{\eta}_i^{ \x_1})'
&\frac{1}{N}\sum_{i=1}^N\bm{\eta}_i^{ \x_1}(\bm{\eta}_i^{G \x_1})'
&\frac{1}{N}\sum_{i=1}^N\bm{\eta}_i^{ \x_1}(\bm{\eta}_i^{G^2 \x_1})' \\
\frac{1}{N}\sum_{i=1}^N\bm{\eta}_i^{G \x_1} (\bm{\eta}_i^{\x_1})'
&\frac{1}{N}\sum_{i=1}^N\bm{\eta}_i^{G \x_1}(\bm{\eta}_i^{G \x_1})'
&\frac{1}{N}\sum_{i=1}^N\bm{\eta}_i^{G \x_1}
(\bm{\eta}_i^{G^2 \x_1})' \end{pmatrix} \\ 
& \xrightarrow{p} 
\begin{pmatrix}
S^{GY,\x_1} & S^{GY,G\x_1} & S^{GY,G^2\x_1} \\
S^{\x_1,\x_1} & S^{\x_1,G \x_1} & S^{\x_1, G^2 \x_1} \\
S^{G\x_1,\x_1} & S^{G\x_1,G\x_1} & S^{G\x_1,G^2\x_1} 
\end{pmatrix}
=: S^{\w\z},
\end{align*}

\begin{align*}
(b) \quad &\frac{1}{N}\sum_{i=1}^N(\z_i-\h^{\z}(a_i))
(\z_i-\h^{\z}(a_i))' \\
&=: \begin{pmatrix}
\frac{1}{N}\sum_{i=1}^N
\bm{\eta}_i^{\x_1} (\bm{\eta}_i^{\x_1})'
&\frac{1}{N}\sum_{i=1}^N\bm{\eta}_i^{\x_1}(\bm{\eta}_i^{G \x_1})'
&\frac{1}{N}\sum_{i=1}^N\bm{\eta}_i^{\x_1}(\bm{\eta}_i^{G^2 \x_1})' \\
\frac{1}{N}\sum_{i=1}^N\bm{\eta}_i^{ G \x_1}(\bm{\eta}_i^{ \x_1})'
&\frac{1}{N}\sum_{i=1}^N\bm{\eta}_i^{ G \x_1}(\bm{\eta}_i^{\x_1})'
&\frac{1}{N}\sum_{i=1}^N\bm{\eta}_i^{ G \x_1}(\bm{\eta}_i^{G \x_1})' \\
\frac{1}{N}\sum_{i=1}^N\bm{\eta}_i^{G^2 \x_1} (\bm{\eta}_i^{ G^2 \x_1})'
&\frac{1}{N}\sum_{i=1}^N\bm{\eta}_i^{G^2  \x_1}(\bm{\eta}_i^{ \x_1})'
&\frac{1}{N}\sum_{i=1}^N\bm{\eta}_i^{G^2 \x_1} (\bm{\eta}_i^{G^2 \x_1})' \end{pmatrix} \\ 
& \xrightarrow{p} 
\begin{pmatrix}
S^{\x_1,\x_1} & S^{\x_1,G\x_1} & S^{\x_1,G^2\x_1} \\
S^{G\x_1,\x_1} & S^{G\x_1,G \x_1} & S^{G\x_1 G^2 \x_1} \\
S^{G^2\x_1,\x_1} & S^{G^2\x_1,G\x_1} & S^{G^2\x_1,G^2\x_1} 
\end{pmatrix}
=: S^{\z\z},
\end{align*}
where
\begin{align*}
S^{GY,G^r\x_1} &=\E\left[\left(
\sum_{m=0}^\infty \beta_{2}^{0'}\tilde{\tilde{s}}^{\x_1}_{m}(\x_i,a_i) +\beta_3^{0'} \tilde{\tilde{s}}^{\x_1}_{m+1}(\x_i,a_i)+\tilde{\tilde{s}}^a_{m}(\x_i,a_i)\right)\left(\tilde{\tilde{s}}^{\x_1}_r(\x_i,a_i)\right)'
\right],\ r=0,1,2 \\
S^{G^r \x_1,G^s\x_1} &= \E\left[\tilde{\tilde{s}}^{\x_1}_r(\x_i,a_i))\left(\tilde{\tilde{s}}^{\x_1}_s(\x_i,a_i)\right)'
\right],\ r,s=0,1,2 \\
\tilde{\tilde{s}}_{m}^{\x_1}(\x_i,a_i) &= \tilde{s}_{m}^{\x_1}(\x_i,a_i)-\E[\tilde{s}_{m}^{\x_1}(\x_i,a_i)|a_i]) \quad \text{with} \quad
\tilde{s}^{\x_1}_0(\x_i,a_i) = \x_{1i} \\
\tilde{\tilde{s}}_{m}^{a}(\x_i,a_i) &= \tilde{s}_{m}^{a}(\x_i,a_i)-\E[\tilde{s}_{m}^{a}(\x_i,a_i)|a_i]) \quad \text{with} \quad
\tilde{s}^{a}_0(\x_i,a_i) = h^{\upsilon}(a_i).
\end{align*}
and $\tilde{\tilde{s}}_{m}^{\x_1}(\x_i,a_i)$ and $\tilde{\tilde{s}}_{m}^{a}(\x_i,a_i)$ are defined recursively as in (\ref{def.tilde.sm}).
\end{lemma}

\noindent {\bf {Proof}}

We take the element $\frac{1}{N}\sum_{i=1}^N\bm{\eta}_i^{GY}(\bm{\eta}_i^{G^2 \x_1})'$ as an example. The proofs of the rest are similar and we omit them.  

When $|\beta_1^0| < 1$,
\[
\G_N \y_N = \sum_{m=0}^{\infty} (\beta_1^0)^m \G_N^{m} 
	(\X_{1N} \beta_2^0 + \G_N \X_{1N} \beta_3^0 + \h^{\upsilon}(\an) + \bm{\eta}^{\upsilon}_N),	
\]
and 
\begin{align*}
&[\G_N \y_N]_i \\
&= \beta_2^{0\prime} \left[\sum_{m=0}^{\infty} (\beta_1^0)^m \G_N^{m} \X_{1N} \right]_i
+ \beta_3^{0\prime} \left[\sum_{m=0}^{\infty} (\beta_1^0)^m \G_N^{m+1} 
 \X_{1N} \right]_i \\
& \quad + \left[\sum_{m=0}^{\infty} (\beta_1^0)^m \G_N^{m} \h(\an) \right]_i
+ \left[\sum_{m=0}^{\infty} (\beta_1^0)^m \G_N^{m} \bm{\eta}^{v}_N \right]_i.
\end{align*}

Set $s^{\x_1}_{0}(\x_i,a_i) = \tilde{s}^{\x_1}_{0}(\x_i,a_i) = \x_{1i}$. We have
\begin{equation*}
\begin{split}
&\frac{1}{N}\sum_{i=1}^N\bm{\eta}_i^{GY}(\bm{\eta}_i^{G^2 \x_1})' \\
&=\frac{1}{N}\sum_{i=1}^N
\left(
[\G_N \y_N]_i-\E\{[\G_N \y_N]_i|a_i\}
\right)
\left(
[\G_N^2\X_{1N}]_i-\E\{[\G_N^2\X_{1N}]_i|a_i\}
\right)'\\
&=\frac{1}{N}\sum_{i=1}^N
\left(
\beta_2^{0'} \sum_{m=0}^\infty (\beta_1^0)^m
\left\{
s^{\x_1}_{m,N,i}
-\E[s^{\x_1}_{m,N,i}|a_i]
\right\}
\right)
\left(
s^{\x_1}_{2,N,i} - \E[s^{\x_1}_{2,N,i}|a_i]
\right)'\\
&+
\frac{1}{N}\sum_{i=1}^N
\left(
\beta_3^{0'}\sum_{m=0}^\infty (\beta_1^0)^m
\left\{
s^{\x_1}_{m+1,N,i}
-\E[s^{\x_1}_{m+1,N,i}|a_i]
\right\}
\right)
\left(
s^{\x_1}_{2,N,i} - \E[s^{\x_1}_{2,N,i}|a_i]
\right)'\\
&+\frac{1}{N}\sum_{i=1}^N
\left(
\beta_2^{0'}\sum_{m=0}^\infty (\beta_1^0)^m
\left\{
s^a_{m,N,i}
-\E[s^a_{m,N,i}|a_i]
\right\}
\right)
\left(
s^{\x_1}_{2,N,i} - \E[s^{\x_1}_{2,N,i}|a_i]
\right)' \\
&+\frac{1}{N}\sum_{i=1}^N
\left(
\beta_2^{0'}\sum_{m=0}^\infty (\beta_1^0)^m
\left\{
s^{\upsilon}_{m,N,i}-\E[s^{\upsilon}_{m,N,i}|a_i]
\right\}
\right)
\left(
s^{\x_1}_{2,N,i} - \E[s^{\x_1}_{2,N,i}|a_i]
\right)' \\
& = I + II + III + IV, \quad say.
\end{split}
\end{equation*}
Consider term $I$,
\begin{equation*}
\begin{split}
&\frac{1}{N}\sum_{i=1}^N
\left(
\beta_2^{0'} \sum_{m=0}^\infty (\beta_1^0)^m
\left\{
s^{\x_1}_{m,N,i}
-\E[s^{\x_1}_{m,N,i}|a_i]
\right\}
\right)
\left(
s^{\x_1}_{2,N,i} - \E[s^{\x_1}_{2,N,i}|a_i]
\right)'.
\end{split}
\end{equation*}

Denote
\begin{align*}
A_{1i} &:= \beta_2^{0'} \sum_{m=0}^\infty (\beta_1^0)^m
\left\{
s^{\x_1}_{m,N,i}
-\E[s^{\x_1}_{m,N,i}|a_i]
\right\} \\
A_{2i} &:= s^{\x_1}_{2,N,i} - \E[s^{\x_1}_{2,N,i}|a_i] \\
A_{3i} &:= \sum_{m=0}^\infty (\beta_1^0)^m
\left\{
s^{\upsilon}_{m,N,i}-\E[s^{\upsilon}_{m,N,i}|a_i]
\right\} \\
B_{1i} &:=\beta_2^{0'} \sum_{m=0}^\infty (\beta_1^0)^m
\left\{
\tilde{s}^{\x_1}_{m}(\x_i,a_i)
-\E[\tilde{s}^{\x_1}_{m}(\x_i,a_i)|a_i]
\right\} \\ 
B_{2i} &:= \tilde{s}^{\x_1}_{2}(\x_i,a_i) - \E[\tilde{s}^{\x_1}_{2}(\x_i,a_i)|a_i] \\
B_{3i} &:= \eta^{\upsilon}_i = \upsilon_i - \E[\upsilon|a_i].
\end{align*}
First, notice that
\begin{align}
\left\| \frac{1}{N} \sum_{i=1}^N A_{1i} A_{2i}' - \frac{1}{N}\sum_{i=1}^N B_{1i} B_{2i}' \right\| 
&= \left\| \frac{1}{N}\sum_{i=1}^N (A_{1i} - B_{1i}) A_{2i}' + \frac{1}{N} \sum_{i=1}^N B_{1i} (A_{2i} - B_{2i})' \right\| \nonumber \\
&\leq \left\|  \frac{1}{N}\sum_{i=1}^N (A_{1i} - B_{1i}) A_{2i}' \right\| + \left\| \frac{1}{N} \sum_{i=1}^N B_{1i} (A_{2i} - B_{2i})' \right\| \nonumber \\
&\leq \sup_i \|A_{1i} - B_{1i} \| \sup_i \| A_{2i} \| + \sup_i \| B_{1i} \| \sup_i \| A_{2i} - B_{2i} \| \label{eq.ap.A1A2-B1B2}
\end{align}
According to Lemma \ref {lemma:sum_m beta s_miN - stilde_mi} and Lemma \ref{lemma plim s uniform in i}, we have
\[
\sup_i \| A_{1i} - B_{1i} \| = o_p(1), \quad \sup_i \| A_{2i} - B_{2i}\| = o_p(1).
\]
Also, under Assumption \ref{assumption:limit.dist}, $\sup_i \| A_{2i} \|$ and $\sup_i \| B_{1i} \|$ are bounded by a finite constant. Therefore, we deduce that
\[
I = \frac{1}{N}\sum_{i=1}^N B_{1i} B_{2i}' + o_p(1).
\]
Then, we apply the WLLN to $\frac{1}{N}\sum_{i=1}^N B_{1i} B_{2i}'$ and deduce
\begin{align*}
\frac{1}{N}\sum_{i=1}^N B_{1i} B_{2i}' 
&\xrightarrow{p} \E\left[ B_{1i} B_{2i}' \right] \\
& = \E \left[ 
\left( \beta_2^{0'} \sum_{m=0}^\infty (\beta_1^0)^m
\left\{
\tilde{s}^{\x_1}_{m}(\x_i,a_i)
-\E[\tilde{s}^{\x_1}_{m}(\x_i,a_i)|a_i]
\right\} \right)
\left( 
\tilde{s}^{\x_1}_{2}(\x_i,a_i) - \E[s^{\x_1}_{2}(\x_i,a_i)|a_i]
\right)
\right] \\
& = \E \left[ 
\left( \beta_2^{0'} \sum_{m=0}^\infty (\beta_1^0)^m
\tilde{\tilde{s}}^{\x_1}_{m}(\x_i,a_i)
\right)
\tilde{\tilde{s}}^{\x_1}_{2}(\x_i,a_i) 
\right]
\end{align*}

We can derive the probability limits of terms $II$ and $III$ by similar fashion.

For term $IV$, first notice that for each $m=0,1,2,...$,
\begin{align*}
\E[s^{\upsilon}_{m,N,i}|a_i] &= \E\left(  [\G^m_N \bm\eta^{\upsilon}_N]_i |a_i \right) \\
&= \E \left\{ \E\left(  [\G^m_N \bm\eta^{\upsilon}_N]_i |\X_N,\D_N,a_i \right) | a_i  \right\} \\
&= \E \left\{ [ \G^m_N \E(  \bm\eta^{\upsilon}_N |\X_N,\D_N, a_i ) ]_i | a_i  \right\} = 0,
\end{align*}
where the last equality holds by Lemma \ref{lemma:control.function}.
Then, $A_{3i} := \sum_{m=0}^\infty (\beta_1^0)^m s^{\upsilon}_{m,N,i}.$ 

Similar to the bound in (\ref{eq.ap.A1A2-B1B2}), notice that
\begin{align*}
\left\| \frac{1}{N} \sum_{i=1}^N A_{3i} A_{2i}' - \frac{1}{N}\sum_{i=1}^N B_{3i} B_{2i}' \right\| 
&\leq \sup_i \|A_{3i} - B_{3i} \| \sup_i \| A_{2i} \| + \sup_i \| B_{3i} \| \sup_i \| A_{2i} - B_{2i} \|.
\end{align*}

According to Lemma \ref{lemma:sum_m beta s_miN - stilde_mi} and Lemma \ref {lemma plim s uniform in i}, 
\[
\sup_i \| A_{3i} - B_{3i} \| = o_p(1), \quad \sup_i \| A_{2i} - B_{2i}\| = o_p(1).
\]
Also, under Assumption \ref{assumption:limit.dist}, $\sup_i \| A_{2i} \|$ and $\sup_i \| B_{3i} \|$ are bounded by a finite constant. Therefore, we deduce that
\[
IV = \frac{1}{N}\sum_{i=1}^N B_{3i} B_{2i}' + o_p(1).
\]
Then, we apply the WLLN to $\frac{1}{N}\sum_{i=1}^N B_{3i} B_{2i}'$ and deduce
\begin{align*}
\frac{1}{N}\sum_{i=1}^N B_{3i} B_{2i}' 
&\xrightarrow{p} \E\left[ B_{3i} B_{2i}' \right] \\
& = \E \left[ 
\eta^{a}_i
\left( 
\tilde{s}^{\x_1}_{2}(\x_i,a_i) - \E[s^{\x_1}_{2}(\x_i,a_i)|a_i]
\right)
\right] \\
& = \E \left[ 
( \upsilon_i - \E[\upsilon_i|a_i] )
\tilde{\tilde{s}}^{\x_1}_{2}(\x_i,a_i) 
\right] \\
& = \E \left\{ \E \left(  \upsilon_i - \E[\upsilon_i|a_i] | \x_i,a_i \right) \tilde{\tilde{s}}^{\x_1}_{2}(\x_i,a_i)  \right\} \\
& = 0.
\end{align*}
$\square$

Let $\sigma^2(\x_i,a_i) := \E[ (\eta^{\upsilon}_i)^2 | \x_i,a_i ] = \mathbb{E}[(\upsilon_i - \E[ \upsilon_i|a_i])^2| \x_i,a_i].$ 

\begin{lemma} \label{lemma:numerator.limit variance} Under Assumptions \ref{as:basic} and \ref{assumption:limit.dist}, as $N \rightarrow \infty$, we have 
	\[ 
	\frac{1}{N}\sum_{i=1}^N (\z_i-\h^{\z}(a_i))(\z_i-\h^{\z}(a_i))' \sigma^2(\x_i,a_i)
	\xrightarrow{p} \mathbf{S}^{\z \z \sigma},
	\]
	where the limit variance $\mathbf{S}^{\z \z \sigma}$ is defined in Lemma \ref{lemma: CLT}.	
\end{lemma}

\noindent {\bf Proof} 

The proof is similar to that of the results in Lemma \ref{lemma: limit of S^ZZ and S^WZ} and we omit it.
$\square$

\begin{lemma} \label{lemma: CLT} Under Assumptions \ref{as:basic} and \ref{assumption:limit.dist}, as $N \rightarrow \infty$, we have 
\[ 
	\frac{1}{\sqrt{N}}\sum_{i=1}^N(\z_i-\h^{\z}(a_i))
	\eta^{\upsilon}_i 
	\Rightarrow \mathcal{N}(0,\mathbf{S}^{\z \z \sigma}),
\]
where
\begin{align*}
&\mathbf{S}^{\z\z\sigma}= 
\begin{pmatrix}
S^{\x_1\x_1\sigma} & S^{\x_1 G\x_1 \sigma} & S^{\x_1 G^2\x_1 \sigma} \\
S^{G\x_1 \x_1 \sigma} & S^{G\x_1 G \x_1 \sigma} & S^{G\x_1 G^2 \x_1 \sigma} \\
S^{G^2\x_1 \x_1 \sigma} & S^{G^2\x_1 G\x_1 \sigma} & S^{G^2\x_1 G^2\x_1 \sigma} 
\end{pmatrix}
\end{align*}
and 
\begin{align*}
S^{G^r \x_1 G^s\x_1 \sigma} &= \E\left[\tilde{\tilde{s}}^{\x_1}_r(\x_i,a_i))\left(\tilde{\tilde{s}}^{\x_1}_s(\x_i,a_i)\right)' \sigma^2(\x_i,a_i)
\right],\ r,s=0,1,2 \\
\tilde{\tilde{s}}_{m}^{\x_1}(\x_i,a_i) &= \tilde{s}_{m}^{\x_1}(\x_i,a_i)-\E[\tilde{s}_{m}^{\x_1}(\x_i,a_i)|a_i]) \quad \text{with} \quad
\tilde{s}^{\x_1}_0(\x_i,a_i) = \x_{1i} \\
\sigma^2(\x_i,a_i) &:= \E[ (\eta^{\upsilon}_i)^2 | \x_i,a_i ] = \mathbb{E}[(\upsilon_i - \E[ \upsilon_i|a_i])^2| \x_i,a_i],
\end{align*}
where $\tilde{\tilde{s}}_{m}^{\x_1}(\x_i,a_i)$ is defined recursively as in (\ref{def.tilde.sm}).
\end{lemma}

\noindent {\bf Proof}

Let $\mathcal{F}_i=(\X_{1N},\D_N,a_i,\eta^{v}_1,\ldots,\eta^{v}_{i-1})$. Conditional on $(\X_{1N}, \D_N, a_i)$, 
\[ 
\E[(\z_i-\h^{\z}(a_i))\eta^{v}_i|\mathcal{F}_i]= (\z_i-\h^{\z}(a_i)) \E[\eta^{v}_i|\mathcal{F}_i]=0,
\]
and so  $\left\{ (\z_i-\h^{\z}(a_i)) \eta^{v}_i,\mathcal{F}_i \right\}$ is a martingale difference sequence.

Since  $\eta^{\upsilon}_i = \upsilon_i - \E[\upsilon_i|a_i]$ is bounded by a constant under Assumption \ref{assumption:limit.dist},  
\begin{equation}
	\E[(\eta^{\upsilon}_i)^4|F_{i-1}]< M \label{eq.ap.bounded.fourth.moment.eta upsilon} 
\end{equation}
for some finite constant $M$. 

Also notice under Assumptions \ref{as:basic}, we have  
\begin{eqnarray*}
	\mathbb{E} [(\eta^{\upsilon}_i)^2 | \mathcal{F}_i ] 
	&=& \mathbb{E}[(\upsilon_i - \E(\upsilon|a_i))^2| \x_i,a_i, \x_{-i},\mathbf{a}_{-i},\D_N(\x_{-i},\mathbf{a}_{-i},\{ u_{ij}\}_{i,j=1,...,N}, \x_i, a_i), \{ \eta^{\upsilon}_j \}_{j < i}] \\
	&=& \mathbb{E}[(\upsilon_i - \E(\upsilon|a_i))^2| \x_i,a_i] \\
	&=:& \sigma^2(\x_i,a_i).	
\end{eqnarray*}	
Let $\ell$ be a nonzero vector whose dimension is the same as the IVs $\z_i$. 
Then,
\begin{align*}
\mathbb{E}[\ell'\bm{\eta}_i^{Z} (\bm{\eta}_i^{Z})' \ell (\eta^{v}_i)^2 | \mathcal{F}_i] 
& = [ \ell' (\z_i-\h^{\z}(a_i)) (\z_i-\h^{\z}(a_i))' \ell ] \mathbb{E} [(\eta^{\upsilon}_i)^2 | \mathcal{F}_i ] \\ 
& = [ \ell' (\z_i-\h^{\z}(a_i)) (\z_i-\h^{\z}(a_i))' \ell ] \sigma^2(\x_i,a_i).
\end{align*}	

Let 
\[
s^2_N := \frac{1}{N}\sum_{i=1}^N \mathbb{E}[\ell' (\z_i-\h^{\z}(a_i)) (\z_i-\h^{\z}(a_i))' \ell (\eta^{\upsilon}_i)^2 | \mathcal{F}_i] 
 = \frac{1}{N} \sum_{i=1}^n [\ell'(\z_i-\h^{\z}(a_i)) (\z_i-\h^{\z}(a_i))'\ell] \sigma^2(\x_i,a_i). 
\]
According to Lemma \ref{lemma:numerator.limit variance}, 
\[
	s^2_N  \xrightarrow{p} \mathbf{S}^{\z\z\sigma}.
\]
Also, since $\ell'(\z_i-\h^{\z}(a_i)) \eta^{\upsilon}_i = \ell'(\z_i-\h^{\z}(a_i))(\upsilon_i - \E[\upsilon_i|a_i])$ is bounded by a constant, under Assumption \ref{assumption:limit.dist} the Lindeberg-Feller condition is satisfied, that is, for any $\epsilon > 0$, 
\begin{align*}
& \frac{1}{N} \sum_{i=1}^N \E\left[
[ \ell'(\z_i-\h^{\z}(a_i)) (\z_i-\h^{\z}(a_i))' \ell ] (\eta^{\upsilon}_i)^2 \mathbb{I}\left\{ | \ell'(\z_i-\h^{\z}(a_i)) \eta^{\upsilon}_i|>\epsilon \sqrt{N} \right\} | \mathcal{F}_i \right] \\
& \leq \sum_{i=1}^N \frac{1}{\epsilon^2 N^2} \E\left[
[ \ell'(\z_i-\h^{\z}(a_i)) (\z_i-\h^{\z}(a_i))' \ell ]^2 (\eta^{\upsilon}_i)^4 | \mathcal{F}_i \right] \\
&\leq \frac{M}{\epsilon N}  \to 0
\end{align*}
as $ N \to \infty.$ 

Then, by the Martingale Central Limit Theorem (e.g., see Corollary 3.1 \cite{HallHeyde2014}), we have the desired result for theorem:
\[
\frac{1}{\sqrt{N}}\sum_{i=1}^N(\z_i-\h^{\z}(a_i))
\eta^{\upsilon}_i 
\Rightarrow \mathcal{N}(0,\mathbf{S}^{\z \z \sigma}).
\] 
$\square$
\\

\noindent {\bf Proof of Theorem \ref{theorem: central limit theorem}.}\label{appendix: proof of clt}

Theorem \ref{theorem: central limit theorem} follows from Lemma \ref{lemma: error from A_hat}, Lemma \ref{lemma: series approximation error},  Lemma \ref{lemma: limit of S^ZZ and S^WZ}, and Lemma \ref{lemma: CLT}. 
$\square$
\\

\subsection{Further Supporting Lemmas}\label{appendix: further supporting lemmas for beta_hat}

\begin{lemma}[Uniform Convergence of $s_{m,N,i}$ in $i$]\label{lemma plim s uniform in i} 
	Assume Assumptions \ref{as:basic},   \ref{assumption: estimation of a_i}, \ref{assumption:sieve basis}, \ref{assumption: Lipschitz condition} and \ref{assumption:limit.dist}. Suppose that $s_0(\x_i,a_i)$ is a bounded function of $\x_i$ and $a_i$. Suppose that we define $s_{m,N,i}$ as in (\ref{def.s m N i}) and consider its probability limit $\tilde{s}_{m}(\x_i,a_i)$ in equation (\ref{def.tilde.sm}) for each $i$. Then, for each $m = 0,1,2,\cdots $
\begin{align*}
(a) & \sup_{1\leq i \leq N} \left| s_{m,N,i}-\tilde{s}_{m}(\x_i,a_i) \right|=o_p(1) \\
(b) & \sup_{1\leq i \leq N} \left| \E[s_{m,N,i}|a_i]-\E[ \tilde{s}_{m}(\x_i,a_i)|a_i] \right|=o_p(1).
\end{align*}
\end{lemma}

\noindent {\bf Proof} 

\noindent {\bf Part (a).}  

\noindent {\bf For $m=0.$} \\
The required result for the lemma holds trivially because of the definition that $s_{0,N,i} = \tilde{s}_{0}(\x_i,a_i)$. 

Next we show the required result for $m=1$ and then use mathematical induction for the rest $m=2,3,...$.

\noindent {\bf For $m=1.$} \\ 
The claim for the case $m=1$ is proved in three steps.

\noindent {\bf Step 1.} \\
Notice that  
\begin{align*}
		s_{1,N,i}  
		&= \left( \frac{1}{N} \sum_{j \neq i} d_{ij} \right)^{-1}
		\frac{1}{N} \sum_{j \neq i} d_{ij}s_{1,N,j} \\
		&= \left( \frac{1}{N} \sum_{j \neq i} 
		\mathbb{I} \left\{ g ( t(\x_{2i},\x_{2j}), a_i, a_j)  \geq u_{ij} \right\}
		 \right)^{-1} \\
		& \quad \times \frac{1}{N} \sum_{j \neq i} 
		\mathbb{I} \left\{ g( t(\x_{2i},\x_{2j}),a_i,a_j )  \geq u_{ij} \right\}
		s_{0}(\x_j,a_j).
\end{align*}
Then, by the WLLN, for each $i$, 
		\begin{eqnarray}
		\frac{1}{N} \sum_{j \neq i} 
		\mathbb{I} \left\{ g (t(\x_{2i},\x_{2j}), a_i, a_j)  \geq u_{ij} \right\}
		&\xrightarrow[]{p}& 
		\int \int \Phi( ( t(\x_{2i},\x_2), a_i, a)\pi(\x_2,a)d\x_2 da \nonumber \\
		& =& \E[d_{ij}|\x_i,a_i]  \label{eq.lemma.supp.point.limit.denom} \\
		\frac{1}{N} \sum_{j \neq i} 
		\mathbb{I} \left\{ g( t(\x_{2i},\x_{2j}), a_i, a_j)  \geq u_{ij} \right\}
		s_{0}(\x_j,a_j)
		&\xrightarrow[]{p}& 
		\int \int  \Phi(t(\x_{2i},\x_2)'\lambda^0+a_i+a) s_0(\x,a)\pi(\x,a)d\x da \nonumber \\
		&=& \E[d_{ij} s_0(\x_j,a_j) | \x_i,a_i ].
		\label{eq.lemma.supp.point.limit.num} 
		\end{eqnarray}
Since $\E[d_{ij}|\x_i,a_i] > 0$ uniformly in $i,j$ under Assumption \ref{assumption:limit.dist} (vi),(v), and (vi) for each $i$ as $N \rightarrow \infty$, we have  
\[ 
 s_{1,N,i} \rightarrow_{p} \tilde{s}_{1}(\x_i,a_i) =
\frac{\int\int \Phi(g( t(\x_{2i},\x_2), a_i, a)s_{0}(\x,a)\pi(\x,a)d\x da}
{ \int \int \Phi( g ( t(\x_{2i},\x_2), a_i, a)\pi(\x_2,a)d\x_2 da}. 
\]

\noindent {\bf Step 2.}\\
In this step, we show that the convergences in (\ref{eq.lemma.supp.point.limit.denom}) and (\ref{eq.lemma.supp.point.limit.num}) hold uniformly in $i$. 
For this, we introduce the following notation. Let 
\[
\zeta_{i,N,1}=\frac{1}{N}\sum_{j=1, \neq i}^N \left( d_{ij}-\E[d_{ij}|\x_i,a_i] \right)
\] 
and 
\[
\zeta_{i,N,2}=\frac{1}{N}\sum_{j=1,\neq i}^N \left(d_{ij}s_{0}(\x_j,a_j)-\E[d_{ij}s_{0}(\x_j,a_j)|\x_i,a_i]\right).
\]
Notice that conditional on $(\x_i,a_i)$, $d_{ij}$ and $d_{ij}s_0(\x_j,a_j)$ are iid  with conditional mean zero and bounded by a constant across $j=1,...,N, \neq i.$ Then, there exists a finite constant $M_1$ such that 
\[
\sup_i \E \left( \| \sqrt{N} \zeta_{i,N,k}  \|^4 | \x_i,a_i \right) \leq M_1, 
\]
and we can deduce the desired result 
\[ 
\sup_i \| \zeta_{i,N,k} \| = O_p(N^{-1/4}) = o_p(1)
\] 
because for any $\epsilon > 0$, we choose $M_2 = \frac{\epsilon}{M_1}$ and then
\begin{align*}
\PR\{ \sup_i \| \zeta_{i,N,k} \| \geq N^{-1/4} M_2^{1/4} | \x_i,a_i \} 
&= \PR\{ \sup_i  N^{-1/4} \| \sqrt{N} \zeta_{i,N,k} \| \geq  M_2^{1/4} | \x_i,a_i \} \\
&= \PR\{ \sup_i  N^{-1} \| \sqrt{N} \zeta_{i,N,k} \|^4 \geq  M_2 | \x_i,a_i \} \\
& \leq  \PR\left\{   \frac{1}{N} \sum_{i=1}^N \| \sqrt{N} \zeta_{i,N,k} \|^4 \geq  M_2 | \x_i,a_i \right\} \\
& \leq \frac{1}{M_2} \frac{1}{N} \sum_{i=1}^N \E \left( \| \sqrt{N} \zeta_{i,N,k}  \|^4 | \x_i,a_i \right) \\
& \leq \frac{M_1}{M_2} = \epsilon.
\end{align*}\\

\noindent {\bf Step 3.}\\
Now we prove the desired result for the case $m=1$. Define 
$\Psi_{i,N,1}=\frac{1}{N} \sum_{j \neq i} d_{ij}$ and $\Psi_{i,N,2}=\frac{1}{N} \sum_{j \neq i} d_{ij}s_{0}(\x_j,a_j)$. Then,
\[
s_{1,N,i}= \frac{\Psi_{i,N,1}}{\Psi_{i,N,2}}.
\] 
Let
$\phi_{i,1}=\frac{1}{N}\sum_{j=1,\neq i}^N\E[d_{ij}|\x_i,a_i]$ and $\psi_{i,2}=\frac{1}{N}\sum_{j=1, \neq i}^N\E[d_{ij}s_{0}(\x_j,a_j)|\x_i,a_i]$. 
Notice that 
\begin{align*}
\sup_i \| s_{1,N,i} \| 
&= \sup_i \left\| \frac{\Psi_{i,N,2}}{\Psi_{i,N,1}} - \frac{\Psi_{i,2}}{\Psi_{i,1}} \right\| \\
&\leq \sup_i \left\| \frac{\Psi_{i,N,2}-\Psi_{i,2}}{\Psi_{i,N,1}} \right\| 
+ \sup_i \left\|  \frac{\Psi_{i,2}(\Psi_{i,N,1}-\Psi_{i,1})}{\Psi_{i,N,1} \Psi_{i,1}} \right\|
= o_p(1),
\end{align*}
where the last line holds because $\| \Psi_{i,N,k}-\Psi_{i,k} \| = o_p(1)$ by Step 2, and $\Psi_{i,1} > 0$ and $\| \Psi_{i,2} \|$ is bounded by a constant. 
This shows the required result
\[
\sup_{i} \| s_{1,N,i} - \tilde{s}_{1}(\x_i,a_i) \| = o_p(1).
\]
\\
\noindent {\bf For $m \geq 2$.}

Given that we show the required result of the lemma with $m=1$, we show the rest by mathematical induction. For this, suppose that 
\[ 
	\sup_{1\leq i \leq N}\|s_{m,N,i}-\tilde{s}_{m}(\x_i,a_i)\|=o_p(1).
\] 
Then, we have
\begin{align*}
& \sup_{1 \leq i \leq N} \| s_{m+1,N,i}-\tilde{s}_{m+1}(\x_i,a_i) \| \\
& = \sup_{1 \leq i \leq N} \left\| 
\frac{\frac{1}{N} \sum_{j=1,\neq i}^N d_{ij} s_{m,N,i}}{\frac{1}{N} \sum_{j=1,\neq i}^N d_{ij}} 
- \frac{\E[ d_{ij} \tilde{s}_m(\x_j,a_j)|\x_i,a_i]}{\E[ d_{ij} |\x_i,a_i]}
\right\| \\
&\leq  \sup_{1 \leq i \leq N}  
\frac{\left\| \frac{1}{N} \sum_{j=1,\neq i}^N d_{ij} \left( s_{m,N,i} - \E[ d_{ij} \tilde{s}_m(\x_j,a_j)|\x_i,a_i] \right) \right\|}
{\frac{1}{N} \sum_{j=1,\neq i}^N d_{ij}} 
 \\
& \quad +  \sup_{1 \leq i \leq N} \left\| \E[ d_{ij} \tilde{s}_m(\x_j,a_j)|\x_i,a_i]  \right\|
\sup_{1 \leq i \leq N} \left| 
\frac{1}{\frac{1}{N} \sum_{j=1,\neq i}^N d_{ij}} 
- \frac{1}{\E[ d_{ij} |\x_i,a_i]} 
\right|.
\end{align*}
For the first term, we have by the definition of $g_{ij} = \frac{d_{ij}}{\sum_{j=1,\neq i}^N d_{ij}}$ and since $\sum_{j=1,\neq i} g_{ij} = 1$, we have
\begin{align*}
& \sup_{1 \leq i \leq N}  
\frac{\left\| \frac{1}{N} \sum_{j=1,\neq i}^N d_{ij} \left( s_{m,N,i} - \E[ d_{ij} \tilde{s}_m(\x_j,a_j)|\x_i,a_i] \right) \right\|}
{\frac{1}{N} \sum_{j=1,\neq i}^N d_{ij}}  \\
& = \sup_{1 \leq i \leq N}  \left\| \frac{1}{N} \sum_{j=1,\neq i}^N g_{ij} \left( s_{m,N,i} - \E[ d_{ij} \tilde{s}_m(\x_j,a_j)|\x_i,a_i] \right) \right\| \\
&\leq \sup_{1 \leq i \leq N}  \left\|  s_{m,N,i} - \E[ d_{ij} \tilde{s}_m(\x_j,a_j)|\x_i,a_i]  \right\| \\
& = o_p(1),
\end{align*}
where the last line holds by the assumption of mathematical induction.
We can show the second term 
\[
\sup_{1 \leq i \leq N} \left\| \E[ d_{ij} \tilde{s}_m(\x_j,a_j)|\x_i,a_i]  \right\|
\sup_{1 \leq i \leq N} \left| 
\frac{1}{\frac{1}{N} \sum_{j=1,\neq i}^N d_{ij}} 
- \frac{1}{\E[ d_{ij} |\x_i,a_i]} 
\right| = o_p(1)
\]
by using similar argument used in the proof of Step 3 of the case $m=1$. $\square$
\\

\noindent {\bf Part (b).}

Notice that under Assumption \ref{assumption:limit.dist}, $\E[ s_{m,N,i}|a_i]$ and $\E[ \tilde{s}_m(\x_i,a_i) | a_i]$ are bounded by a finite constant. The required argument follows by similar arguments used in the proof of Part (a). $\square$
\\

\begin{lemma}[Uniform Convergence of $s^{\upsilon}_{m,N,i}$ in $i$]\label{lemma plim s_eta uniform in i} 
	Assume Assumptions Assumptions \ref{as:basic}, \ref{assumption: estimation of a_i}, \ref{assumption:sieve basis}, \ref{assumption: Lipschitz condition} and \ref{assumption:limit.dist}. Suppose that we define $s^{\upsilon}_{m,N,i}$ as in (\ref{def.s_eta_m N i}).  Then, for each $m = 1,2,\cdots $
	\begin{equation*}
	\sup_{1\leq i \leq N}|s^{\upsilon}_{m,N,i} |=o_p(1). 
	\end{equation*}
\end{lemma}

\noindent {\bf {Proof}}

	The proof is similar to that of Lemma \ref{lemma plim s uniform in i}. 	
	First, we show that for each $i$ and $m=1,2,...$ the probability limit of $s^{\upsilon}_{m,N,i}$ defined with $s_{0,i}^{\upsilon} = \eta^{\upsilon}_i = \upsilon_i - \E[\upsilon_i|a_i]$ recursively as (\ref{def.s_eta_m N i}) is zero  as $N \rightarrow \infty$. To verify this, let
	\[
	\tilde{s}_{0,i}^{\upsilon} = \eta^{\upsilon}_i = \upsilon_i - \E[\upsilon_i|a_i].
	\]
	
	For $m=1$, 
	\[
		s_{1,N,i}^{\upsilon} =
		\left( \frac{1}{N} \sum_{j \neq i} d_{ij} \right)^{-1}
		\frac{1}{N} \sum_{j \neq i} d_{ij} s_{0,j}^{\upsilon}.
	\]
	
	Consider the numerator. Notice by definition that 
	\[ 
	d_{ij} s_{0,j}^{\upsilon} = \mathbb{I} \left\{ g( t(\x_{2i},\x_{2j}, a_i,a_j)  \geq u_{ij} \right\}
	(\upsilon_j - \E[\upsilon_j|a_j])
	\] 
	are i.i.d. across $j$ conditioning on $(\x_{2i},a_i)$ and bounded by a finite constant under Assumption \ref{assumption:limit.dist}. Then, by the WLLN conditioning on $(\x_{2i},a_i)$, we have
	\begin{align*}
	\frac{1}{N} \sum_{j \neq i} d_{ij} s_{0,j}^{\upsilon} 
	& \xrightarrow[]{p} \E \left[ 
	d_{ij} (\upsilon_j - \E[\upsilon_j|a_j] ) | \x_{2i},a_i
	\right] \\
	& = \E \left[ 
	d_{ij} \E \left( 
	\upsilon_j - \E[\upsilon_j|a_j] | \X_N,\D_N,a_i	
	\right)
	| \x_{2i},a_i \right] \\
	& = 0,
	\end{align*}
	where the last equality holds by Lemma \ref{lemma:control.function}.
	The denominator converges to 
	\[
	\frac{1}{N} \sum_{j \neq i} 
	\mathbb{I} \left\{ g( t(\x_{2i},\x_{2j}), a_i, a_j )  \geq u_{ij} \right\}
	\rightarrow_{p}
	\int\int \Phi( g( t(\x_{2i},\x_2), a_i,  a)\pi(\x_2,a)d\x_2 da > 0,
	\]
	where the last inequality holds under Assumption \ref{assumption:limit.dist}. 
	
	This shows that as $N \rightarrow \infty$ 
	\[
	\frac{1}{N} \sum_{j \neq i} g_{ij} s_{0,j}^{\upsilon} \xrightarrow[]{p} 0 =: \tilde{s}^{\upsilon}_{1,i}
	\]
	for each $i$. 
	
	Then, using similar argument in Step 2 of the proof of Lemma \ref{lemma plim s uniform in i}, we deduce
	\[
	\sup_{1 \leq i \leq N} \left| \frac{1}{N} \sum_{j \neq i} g_{ij} s_{0,j}^{\upsilon} \right| = o_p(1).
	\]
	 
	Also, for $m=2,...$, we follow the same mathematical induction argument in Steps 3 and 4 of the proof of Lemma \ref{lemma plim s uniform in i} and deduce that
	\[ 
	\sup_{1 \leq i \leq N} \left| \frac{1}{N} \sum_{j \neq i} g_{ij} s_{m,N,j}^{\upsilon} \right| = o_p(1).
	\]
	$\square$

\begin{lemma}\label{lemma:sum_m beta s_miN - stilde_mi} Assume Assumptions  \ref{as:basic},   \ref{assumption: estimation of a_i}, \ref{assumption:sieve basis}, \ref{assumption: Lipschitz condition} and \ref{assumption:limit.dist}. Suppose that $s_0(\x_i,a_i)$ is a bounded function of $\x_i$ and $a_i$. Suppose that we define $s_{m,N,i}$ as in equation (\ref{def.s m N i}) and consider its probability limit $\tilde{s}_{m}(\x_i,a_i)$ in equation (\ref{def.tilde.sm}) for each $i$. Then,
\begin{align*}
(a) & \sup_{1\leq i\leq N} \left|\sum_{m=0}^\infty (\beta_1^0)^m\left(  s_{m,N,i} -\tilde{s}_{m}(\x_i,a_i)\right) \right|= o_p(1) \\
(b) & \sup_{1\leq i\leq N} \left|\sum_{m=0}^\infty (\beta_1^0)^m\left( \E[ s_{m,N,i} | a_i] - \E[ \tilde{s}_{m}(\x_i,a_i) | a_i] \right) \right|= o_p(1).
\end{align*}
Also, suppose that we define $s^{\eta}_{m,N,i}$ as in equation (\ref{def.s_eta_m N i}). Let $\tilde{s}^{\eta}_{0,i} = \eta^{a}_i$ and $\tilde{s}^{\eta}_{m,i} = 0$ for $m=1,2,...$. Then,
\begin{equation*}
(c) \:\: \sup_{1\leq i\leq N} \left|\sum_{m=0}^\infty (\beta_1^0)^m\left( s^{\eta}_{m,N,i} -\tilde{s}^{\eta}_{m} \right)  \right|= o_p(1).
\end{equation*}
\end{lemma}

\noindent {\bf Proof}

\noindent {\bf Part (a).} 

Notice from Assumption \ref{assumption:limit.dist} that $| \beta_1^0| < 1$ and $  s_{m,N,i}, \tilde{s}_m(\x_i,a_i), \E[ s_{m,N,i} | a_i], \E[ \tilde{s}_m(\x_i,a_i)|a_i ] $ are bounded by a finite constant, say, $M$.
For given $\epsilon > 0$, we choose $m^{*}$ such that $2M \sum_{m = m^{*}+1}^{\infty} (\beta_1^0)^m \leq \epsilon$.
Then, by definition, we have 
\[
\sup_{1\leq i\leq N} \left|\sum_{m=m^{*}+1}^\infty (\beta_1^0)^m\left(  s_{m,N,i} -\tilde{s}_{m}(\x_i,a_i)\right) \right| 
 \leq 2M \sum_{m=m^{*}+1}^\infty (\beta_1^0)^m \leq \epsilon.
\]
Notice that
\begin{align*}
 \sup_{1\leq i\leq N} \left|\sum_{m=0}^\infty (\beta_1^0)^m\left(  s_{m,N,i} -\tilde{s}_{m}(\x_i,a_i)\right) \right| 
& \leq \sup_{1\leq i\leq N} \left|\sum_{m=0}^{m^{*}} (\beta_1^0)^m\left(  s_{m,N,i} -\tilde{s}_{m}(\x_i,a_i)\right) \right| + \epsilon \\
& \leq m^{*} \sup_{1\leq i\leq N} |s_{m,N,i} -\tilde{s}_{m}(\x_i,a_i)|   + \epsilon  \\
& = o_p(1) + \epsilon,
\end{align*}
where the last inequality holds since $m^{*}$ is finite and by Lemma \ref{lemma:sum_m beta s_miN - stilde_mi}. 
Since $\epsilon$ is arbitrary, we have the desired result for Part (a). $\square$ 
\\
 
\noindent {\bf Parts (b) and (c).}

Under Assumption \ref{assumption:limit.dist}, $\E[ s_{m,N,i} | a_i], \E[ \tilde{s}_{m}(\x_i,a_i) | a_i]$, and $\eta^{\upsilon}_i = \upsilon_i - h^{\upsilon}(a_i)$ are bounded by a constant. Apply the same argument used in the proof of Part (a), then we deduce the required result of Parts (b) and (c).
$\square$

\section{For $\bar{\beta}_{2SLS}$}

\subsection{Limiting distribution of $\bar{\beta}_{2SLS}$}

Recall the definition that for any variable $b_i^l$ being an element of $(y_i,\w_i,\w_i)$ and $\upsilon_i$, 
\[
\eta^{l}_{*i}:= b^l_i - h_{*}^l(\x_{2i},a_i) = b^l_i - h_{**}^l(\x_{2i},\text{deg}_i), \quad \eta^{\upsilon}_{*i}: =\upsilon_i - h^{\upsilon}_{*}(\x_{2i},a_i) \upsilon_i - h^{\upsilon}_{**}(\x_{2i},\text{deg}_i).
\]
Let $\bm{\eta}_{*N}^{\upsilon} = (\eta^{\upsilon}_{*1},...,\eta^{\upsilon}_{*N})'$.

Outline:\begin{itemize}
	\item[Step 1] Show that \begin{eqnarray}\label{eq: step 1 for beta_bar proof}
	&&\sqrt{N}(\bar{\beta}_{2SLS} -\beta^0)\nonumber\\
	&=& \left( \W_N^{\prime}\M_{{\R}_N}\Z_N 
	\left( \Z_N^{\prime} \M_{{\R}_N} \Z_N \right)^{-1} 
	\Z_N^{\prime} \M_{{\R}_N} \W_N \right)^{-1} \nonumber \\
	&&
	\times \W_N^{\prime}\M_{{\R}_N}\Z_N 
	\left( \Z^{\prime} \M_{{\R}_N} \Z_N \right)^{-1} \Z_N^{\prime} \M_{{\R}_N} \bm{\eta}_{*N}^{\upsilon}+o_p(1).
	\end{eqnarray}
	\item[Step 2] Show
	\begin{align*}
	\frac{1}{N} \sum_{i=1}^{N} 
	\left(b^l_i - \widehat{h}^l_{**}(\x_{2i},\text{deg}_i) \right)
	\left(b^l_i - \widehat{h}^l_{**}(\x_{2i},\text{deg}_i) \right)'
	& \frac{1}{N} \sum_{i=1}^{N} 
	\left(b^l_i - h^l_{**}(\x_{2i},\text{deg}_i) \right)
	\left(b^l_i - h^l_{**}(\x_{2i},\text{deg}_i) \right)'+o_p(1)
	\end{align*}
	and \[
	\frac{1}{\sqrt{N}} \sum_{i=1}^{N}
	\left(b_i^l- \widehat{h}^l_{**}(\x_{2i},\text{deg}_i) \right)\eta^{\upsilon}_{*i}
	= \frac{1}{\sqrt{N}} \sum_{i=1}^{N}
	\left(b_i^l- h^l_{**}(\x_{2i}, \text{deg}_i) \right)\eta^{\upsilon}_{*i}+ o_p(1).
	\]
	\item[Step 3] Derive the limits of 
		\begin{equation*}
	\frac{1}{N} \sum_{i=1}^{N} 
	\left(b^l_i - h^l_{**}(\x_{2i},\text{deg}_i) \right)
	\left(b^l_i - h^l_{**}(\x_{2i},\text{deg}_i) \right)'
	= \frac{1}{N} \sum_{i=1}^{N} 
	\left(b^l_i - h^l_{**}(\x_{2i},a_i) \right)
	\left(b^l_i - h^l_{**}(\x_{2i},a_i \right)' 
	\end{equation*}
	and 
	\[
\frac{1}{\sqrt{N}} \sum_{i=1}^{N}
	\left(b_i^l- h^l_{**}(\x_{2i},\text{deg}_i) \right)\eta^{\upsilon}_{*i}
	= \frac{1}{\sqrt{N}} \sum_{i=1}^{N}
	\left(b_i^l- h^l_{*}(\x_{2i},a_i) \right)\eta^{\upsilon}_{*i}
	\]
\end{itemize}

\subsection{Controlling the Sampling Error $\widehat{deg}_i-deg_i$ in Sieve Estimation}\label{appenxid: error from deg_hat-deg}

Equation \eqref{eq: step 1 for beta_bar proof} holds if the following Lemma is true.
\begin{lemma}\label{lemma: error from deg_hat} 
	Assume Assumptions Assumptions  \ref{as:basic}, \ref{as:basic.alternative}, \ref{assumption:rank.alternative},  \ref{assumption:sieve basis.alternative}, \ref{assumtion:sieve with x2, Lipschitz} and \ref{assumption:limit.dist}. Then the following holds.
	\begin{enumerate}
		\item[(a)] $\frac{1}{N}(\Z_N^{\prime}\PP_{\widehat{R}_N}\W_N - \Z_N^{\prime}\PP_{R_N}\W_N) = o_p(1)$.
		\item[(b)] $\frac{1}{N}(\Z_N^{\prime}\PP_{\widehat{R}_N}\Z_N - \Z_N^{\prime}\PP_{R_N}\Z_N) = o_p(1)$.
		\item[(c)] $\frac{1}{\sqrt{N}}(\Z_N^{\prime}\PP_{\widehat{R}_N}\bm{\eta^{v}}_{*N}- \Z_N^{\prime}\PP_{R_N}\bm{\eta^{v}}_{*N}) = o_p(1)$.
		\item[(d)] $\frac{1}{\sqrt{N}}(\Z^{\prime}\M_{\widehat{R}_N} (H(\an)- \widehat{\textbf{R}}_N \gamma)) = o_p(1)$.
	\end{enumerate}
\end{lemma}

\begin{proof}
	We can apply a similar argument as in Lemma \ref{lemma: error from A_hat} and derive the desired result.
\end{proof}

\subsection{Controlling the Series Approximation Error for $\mathbf{r}^K(\x_{2i},\text{deg}_i)$}
 
	\begin{lemma}[Series Approximation]\label{lemma: series approximation error.alternative} Assume the assumptions in Lemma \ref{lemma: error from deg_hat}.  Then, we have
		\begin{itemize}
			\item[(a)] $\frac{1}{N} \sum_{i=1}^{N} 
			(\w_i-\widehat{\h}_{**}^{\w}(\x_{2i},\text{deg}_i))
			(\z_i-\widehat{\h}_{**}^\z(\x_{2i},\text{deg}_i))'
			= \frac{1}{N} \sum_{i=1}^{N} 
			(\w_i-\h_{**}^{\w}(\x_{2i},\text{deg}_i))
			(\z_i-\h_{**}^\z(\x_{2i},\text{deg}_i))'+ o_p(1)$,
			\item[(b)] $\frac{1}{N} \sum_{i=1}^{N} 
			(\z_i-\widehat{\h}_{**}^{\z}(\x_{2i},\text{deg}_i))
			(\z_i-\widehat{\h}_{**}^\z(\x_{2i},\text{deg}_i))'
			= \frac{1}{N} \sum_{i=1}^{N} 
			(\z_i-\h_{**}^{\z}(\x_{2i},\text{deg}_i))
			(\z_i-\h_{**}^\z(\x_{2i},\text{deg}_i))'+ o_p(1)$,
			\item[(c)] $\frac{1}{\sqrt{N}} \sum_{i=1}^{N}
			(\z_i-\widehat{\h}_{**}^{\z}(\x_{2i},\text{deg}_i))
			\eta^{\upsilon}_{*i}
			= \frac{1}{\sqrt{N}} \sum_{i=1}^{N}
			(\z_i-\h_{**}^{\z}(\x_{2i},\text{deg}_i))
			\eta^{\upsilon}_{*i}+ o_p(1) $.
		\end{itemize}
	\end{lemma}
	Then the proofs are analogous to the proofs presented in Section \ref{appendix: series approximation error} and we omit them.

\subsection{Limiting distribution of $\bar{\beta}_{2SLS}$}
Note that $h^{l}_{**}(\x_{2i},deg_i)=h^l_*(\x_{2i},a_i)$. 
Using this relationship we can state the following Lemmas.
\begin{lemma}\label{lemma: limit of S^ZZ and S^WZ.alternative}  Under Assumption \ref{as:basic}, \ref{as:basic.alternative}, and \ref{assumption:limit.dist}, we have 
	\begin{align*}
\frac{1}{N}\sum_{i=1}^N(\w_i-\h_*^{\w}(\x_{2i},a_i))
	(\z_i-\h_*^{\z}(\x_{2i},a_i))'  \xrightarrow{p} 
	\begin{pmatrix}
	\bar{S}^{GY,\x_1} & \bar{S}^{GY,G\x_1} & \bar{S}^{GY,G^2\x_1} \\
	\bar{S}^{\x_1,\x_1} & \bar{S}^{\x_1,G \x_1} & \bar{S}^{\x_1, G^2 \x_1} \\
	\bar{S}^{G\x_1,\x_1} & \bar{S}^{G\x_1,G\x_1} & \bar{S}^{G\x_1,G^2\x_1} 
	\end{pmatrix}
	=: \bar{S}^{\w\z},
	\end{align*}
	and 
	\begin{align*}
\frac{1}{N}\sum_{i=1}^N(\z_i-\h_*^{\z}(\x_{2i},a_i))
	(\z_i-\h_*^{\z}(\x_{2i},a_i))' 
	 \xrightarrow{p} 
	\begin{pmatrix}
	\bar{S}^{\x_1,\x_1} & \bar{S}^{\x_1,G\x_1} & \bar{S}^{\x_1,G^2\x_1} \\
	\bar{S}^{G\x_1,\x_1} & \bar{S}^{G\x_1,G \x_1} & \bar{S}^{G\x_1 G^2 \x_1} \\
	\bar{S}^{G^2\x_1,\x_1} & \bar{S}^{G^2\x_1,G\x_1} & \bar{S}^{G^2\x_1,G^2\x_1} 
	\end{pmatrix}
	=: \bar{S}^{\z\z},
	\end{align*}
	where
	\begin{align*}
	\bar{S}^{GY,G^r\x_1} &=\E\left[\left(
	\sum_{m=0}^\infty \beta_{2}^{0'}\tilde{\tilde{s}}^{\x_1}_{*m}(\x_i,a_i) +\beta_3^{0'} \tilde{\tilde{s}}^{\x_1}_{*,m+1}(\x_i,a_i)+\tilde{\tilde{s}}^a_{*m}(\x_i,a_i)\right)\left(\tilde{\tilde{s}}^{\x_1}_{*r}(\x_i,a_i)\right)'
	\right],\ r=0,1,2 \\
	\bar{S}^{G^r \x_1,G^s\x_1} &= \E\left[\tilde{\tilde{s}}^{\x_1}_{*r}(\x_i,a_i))\left(\tilde{\tilde{s}}^{\x_1}_{*s}(\x_i,a_i)\right)'
	\right],\ r,s=0,1,2 \\
	\tilde{\tilde{s}}_{*m}^{\x_1}(\x_i,a_i) &= \tilde{s}_{*m}^{\x_1}(\x_i,a_i)-\E[\tilde{s}_{*m}^{\x_1}(\x_i,a_i)|\x_{2i},a_i]) \quad \text{with} \quad
	\tilde{s}^{*\x_1}_0(\x_i,a_i) = \x_{1i} \\
	\tilde{\tilde{s}}_{*m}^{a}(\x_i,a_i) &= \tilde{s}_{*m}^{a}(\x_i,a_i)-\E[\tilde{s}_{*m}^{a}(\x_i,a_i)|\x_{2i},a_i]) \quad \text{with} \quad
	\tilde{s}^{a}_{*0}(\x_i,a_i) = h^{\upsilon}_{*}(\x_{2i},a_i),
	\end{align*}
	where $\tilde{\tilde{s}}_{*m}^{\x_1}(\x_i,a_i)$ and $\tilde{\tilde{s}}_{*m}^{a}(\x_i,a_i)$ are defined recursively as in (\ref{def.tilde.sm}).
\end{lemma}

\begin{lemma} \label{lemma:numerator.limit variance.alternative} Under Assumption \ref{as:basic}, \ref{as:basic.alternative}, and \ref{assumption:limit.dist}, 
	\[ 
	\frac{1}{N}\sum_{i=1}^N (\z_i-\h_{*}^{\z}(\x_{2i},a_i))(\z_i-\h_{*}^{\z}(\x_{2i},a_i))' \sigma^2_{*}(\x_i,a_i)
	\rightarrow_{p} \mathbf{\bar{S}}^{\z \z \sigma},
	\]
	where the limit variance $\mathbf{\bar{S}}^{\z \z \sigma}$ is defined in Lemma \ref{lemma: CLT.alternative}.	
\end{lemma}

$\square$

\begin{lemma} \label{lemma: CLT.alternative}  Under Assumption \ref{as:basic}, \ref{as:basic.alternative}, and \ref{assumption:limit.dist},
	\[ 
	\frac{1}{\sqrt{N}}\sum_{i=1}^N(\z_i-\h_{*}^{\z}(\x_{2i},a_i))
	\eta^{\upsilon}_{*i}
	\Rightarrow \mathcal{N}(0,\mathbf{\bar{S}}^{\z \z \sigma}),
	\]
	where
	\begin{align*}
	&\mathbf{\bar{S}}^{\z\z\sigma}= 
	\begin{pmatrix}
	\bar{S}^{\x_1\x_1\sigma} & \bar{S}^{\x_1 G\x_1 \sigma} & \bar{S}^{\x_1 G^2\x_1 \sigma} \\
	\bar{S}^{G\x_1 \x_1 \sigma} & \bar{S}^{G\x_1 G \x_1 \sigma} & \bar{S}^{G\x_1 G^2 \x_1 \sigma} \\
	\bar{S}^{G^2\x_1 \x_1 \sigma} & \bar{S}^{G^2\x_1 G\x_1 \sigma} & \bar{S}^{G^2\x_1 G^2\x_1 \sigma} 
	\end{pmatrix}
	\end{align*}
	and 
	\begin{align*}
	\bar{S}^{G^r \x_1 G^s\x_1 \sigma} &= \E\left[\tilde{\tilde{s}}^{\x_1}_{*r}(\x_i,a_i))\left(\tilde{\tilde{s}}^{\x_1}_{*s}(\x_i,a_i)\right)' \sigma^2_{*}(\x_i,a_i)
	\right],\ r,s=0,1,2 \\
	\tilde{\tilde{s}}_{*m}^{\x_1}(\x_i,a_i) &= \tilde{s}_{*m}^{\x_1}(\x_i,a_i)-\E[\tilde{s}_{*m}^{\x_1}(\x_i,a_i)|\x_{2i},a_i]) \quad \text{with} \quad
	\tilde{s}^{\x_1}_{*0}(\x_i,a_i) = \x_{1i} \\
	\sigma^2_{*}(\x_i,a_i) &:= \E[ (\eta^{\upsilon}_{*i})^2 | \x_i,a_i ] = \mathbb{E}[(\upsilon_i - \E[ \upsilon_{i}|\x_{2i},a_i])^2| \x_i,a_i], 
	\end{align*}
	where $\tilde{\tilde{s}}_{*m}^{\x_1}(\x_i,a_i)$ is defined recursively as in (\ref{def.tilde.sm}).
\end{lemma}\clearpage 
\section{Supplementary Monte Carlo results}\label{appendix: supplementary monte carlo}
In this section we present Monte Carlo results for the dense and sparse network formation designs presented in Tables \ref{table: dense network design supplement} and \ref{table: sparse network design supplement}.
Design 5-8 for both the dense and sparse networks involve degree heterogeneity distributions that are correlated with $x_{2i}$ and right skewed, which mimics distributions observed in real world networks. \\

 In Section \ref{section: dense kn4 hermite} and \ref{section: dense kn8 hermite} we present results for the dense network formation design, Hermite polynomial sieve with $K_N=4$ and $K_N=8$, respectively. The corresponding results for sparse network formation designs are included in Sections \ref{section: sparse kn4 hermite} and \ref{section: sparse kn8 hermite}. Sections \ref{section: dense kn4 polynomial} and \ref{section: dense kn8 polynomial} include results for dense network formation designs and polynomial sieve. Sections
\ref{section: sparse kn4 polynomial} and \ref{section: sparse kn8 polynomial} show results for sparse network designs and polynomial sieve. Overall the results are similar to main text but one noticeable finding is that when $\bf{x}_{2i}$ and $a_i$ are strongly correlated, the network is sparse, and the $h(a_i)$ function is exponential, the control for degree approach suffers from size distortion even tough the estimate has a very small bias. The sparse network case violates the regularity conditions, and we leave it as future research why we have this finite sample issue in the sparse case. In the dense case there is no size distortion. 

\begin{table}[!h]\caption{\footnotesize {\bf Statistics for dense network designs}} \label{table: dense network design supplement}
	\begin{threeparttable} 
		\centering \footnotesize
		\scalebox{1}{\begin{tabular}{|l|c|c|c|c|c|c|c|c|}\toprule 
				Design &\bf 1& \bf 2 & \bf 3 & \bf 4& \bf 5 & \bf 6 & \bf 7 & \bf 8 \\ \midrule 
				$\mu_0$ & 1.00& 1.00& 1.00& 0.25& 0.25& 0.25& 0.25& 0.25 \\ 
				$\mu_1$ & 1.00& 1.00& 1.00& 0.75& 0.75& 0.75& 0.75& 0.75 \\ 
				$\alpha_L$ & -0.50& 0.00& -0.25& -0.75& -0.50& -0.67& -0.50& -0.75 \\ 
				$\alpha_H$ & -0.50& 0.00& -0.25& -0.75& 0.00& -0.17& 0.00& -0.50 \\ 
				$corr(a_i,\bm{x}_{2i})$ & -0.00& -0.00& -0.00& 0.01& 0.64& 0.64& 0.64& 0.38 \\ 
				Avg. Degree & 31.01& 49.52& 40.03& 22.97& 39.70& 33.81& 39.70& 26.88 \\ 
				Avg. Skewness & 0.13& -0.02& 0.05& 0.66& 0.17& 0.21& 0.17& 0.50 \\ \bottomrule 
		\end{tabular}} 
	\end{threeparttable} 
\begin{tablenotes}\footnotesize
\item The statistics are calculated for $N=100$.
\end{tablenotes}
\end{table}

\begin{table}[!h]\caption{\footnotesize {\bf Statistics for sparse network designs}}\label{table: sparse network design supplement} 
	\begin{threeparttable} 
		\centering \footnotesize
		\scalebox{1}{\begin{tabular}{|l|c|c|c|c|c|c|c|c|}\toprule 
				Design &\bf 1& \bf 2 & \bf 3 & \bf 4& \bf 5 & \bf 6 & \bf 7 & \bf 8 \\ \midrule 
				$\mu_0$ & 1.00& 0.25& 1.00& 1.00& 0.25& 0.25& 0.25& 1.00 \\ 
				$\mu_1$ & 1.00& 0.75& 1.00& 1.00& 0.75& 0.75& 0.75& 1.00 \\ 
				$\alpha_L$ & -0.50& -0.50& 0.00& -0.25& -0.50& -0.67& -0.75& -0.50 \\ 
				$\alpha_H$ & -0.50& -0.50& 0.00& -0.25& 0.00& 0.25& 0.00& 0.50 \\ 
				$corr(a_i,\bm{x}_{2i})$ & -0.00& 0.01& -0.00& -0.00& 0.64& 0.83& 0.78& 0.87 \\ 
				Avg. Degree & 1.10& 1.11& 2.88& 1.78& 1.99& 2.62& 1.75& 3.94 \\ 
				Avg. Skewness & 0.98& 1.06& 0.67& 0.81& 1.07& 1.10& 1.19& 0.80 \\ \bottomrule 
		\end{tabular}} 
	\end{threeparttable} 
\begin{tablenotes}\footnotesize
	\item The statistics are calculated for $N=100$.
\end{tablenotes}
\end{table} 

\subsection{Dense Network, $K_N=4$, Hermite polynomial sieve}\label{section: dense kn4 hermite}
\begin{table}[!h]\caption{\footnotesize {\bf Design 1 dense network: Parameter values across 1000 Monte Carlo replications with $K_N=4$ and Hermite polynomial sieve}} 
\begin{threeparttable} 
\centering \footnotesize
\scalebox{.8}{\begin{tabular}{|c|c|c|c|c|c|c|c|c|c|c|c|c|c|c|c|}\toprule 
\multicolumn{14}{c}{$h(a_i) = \exp(a_i)$}\\ 
\cellcolor{yellow}$N$&\multicolumn{6}{|c|}{\cellcolor{yellow}$100$}&\multicolumn{6}{|c|}{\cellcolor{yellow}$250$}&\\\hline 
CF&$(0)$&$(1)$&$(2)$&$(3)$&$(4)$&$(5)$& $(0)$ &$(1)$&$(2)$&$(3)$&$(4)$&$(5)$&\\\hline 
\multirow{4}{*}{$\beta_1=0.8$}& 0.005 & 0.007 &-0.000 &0.000 &0.000 &0.000& 0.009 &0.015 &0.000& 0.001 & 0.000& 0.000 & \textit{mean bias} \\ 
&(0.016 )&(0.018 )&(0.024 )&(0.023 )&(0.035 )&(0.014 )&(0.017 )&(0.020 )&(0.022 )&(0.022 )&(0.033 )&(0.012 )&\textit{std}\\ 
& 0.403 & 0.333 &0.073 &0.066 &0.070 &0.050& 0.773 &0.656 &0.054& 0.061 &0.042& 0.062 &\textit{size} \\ \midrule
\multirow{4}{*}{$\beta_2=5$}& -0.005 & -0.006 &-0.001 &-0.001 &-0.002 &-0.002& -0.003 &-0.005 &0.000& -0.000 &0.000& 0.000 &\textit{mean bias} \\ 
&(0.031 )&(0.032 )&(0.034 )&(0.033 )&(0.034 )&(0.031 )&(0.020 )&(0.020 )&(0.020 )&(0.020 )&(0.021 )&(0.019 )&\textit{std}\\
\end{tabular}} 
\begin{tablenotes}\tiny 
\item CF - control function. $(0)$ - none, $(1)$ - $\lambda_a\hat{a}_i$,  $(2)$ - $\hat{h}(\hat{a}_i)$, $(3)$ - $\hat{h}(a_i)$, $(4)$ - $\hat{h}(\widehat{deg}_i,x_{2i})$, $(5)$ - $h(a_i)$. 
\item The network design parameters are $\mu_0=1.00$, $\mu_1=1.00$, $\alpha_L=-0.50$, $\alpha_H=-0.50$ 
\item Average number of links for $N=100$ is $31.0$, for $N=250$ it is $77.9$. 
\item Average skewness for $N=100$ is $0.12$, for $N=250$ it is $0.11$. 
\item Size is the empirical size of t-test against the truth. 
\item N$=100$, $corr(a_i,\bm{x}_{2i})=-0.001$,N$=250$, $corr(a_i,\bm{x}_{2i})=-0.002$ 
\item The bias of $\hat{a}_i$ is calculated as $a_i-\hat{a}_i$. 
\item For $N=100$, $\hat{a}_i$ mean bias$=0.009$, median bias$=0.004$, std$=0.249$. 
\item For $N=250$, $\hat{a}_i$ mean bias$=0.004$, median bias$=0.002$, std$=0.154$. 
  \end{tablenotes} 
\end{threeparttable} 
\end{table} 
\begin{table}[!h]\caption{\footnotesize {\bf Design 2 dense network: Parameter values across 1000 Monte Carlo replications with $K_N=4$ and Hermite polynomial sieve}} 
\begin{threeparttable} 
\centering \footnotesize
\scalebox{.8}{\begin{tabular}{|c|c|c|c|c|c|c|c|c|c|c|c|c|c|c|c|}\toprule 
\multicolumn{14}{c}{$h(a_i) = \exp(a_i)$}\\ 
\cellcolor{yellow}$N$&\multicolumn{6}{|c|}{\cellcolor{yellow}$100$}&\multicolumn{6}{|c|}{\cellcolor{yellow}$250$}&\\\hline 
CF&$(0)$&$(1)$&$(2)$&$(3)$&$(4)$&$(5)$& $(0)$ &$(1)$&$(2)$&$(3)$&$(4)$&$(5)$&\\\hline 
\multirow{4}{*}{$\beta_1=0.8$}& -0.038 & -0.049 &-0.001 &-0.000 &-0.001 &-0.000& 5.864 &0.165 &0.002& 0.002 & 0.003& 0.011 & \textit{mean bias} \\ 
&(0.693 )&(0.783 )&(0.071 )&(0.058 )&(0.102 )&(0.036 )&(186.308 )&(8.967 )&(0.062 )&(0.057 )&(0.092 )&(0.243 )&\textit{std}\\ 
& 0.896 & 0.904 &0.054 &0.060 &0.065 &0.051& 0.946 &0.949 &0.054& 0.072 &0.049& 0.055 &\textit{size} \\ \midrule
\multirow{4}{*}{$\beta_2=5$}& 0.009 & 0.011 &-0.001 &-0.002 &-0.001 &-0.002& -0.650 &-0.022 &0.000& -0.000 &0.001& -0.002 &\textit{mean bias} \\ 
&(0.174 )&(0.195 )&(0.042 )&(0.033 )&(0.043 )&(0.031 )&(20.492 )&(1.049 )&(0.023 )&(0.020 )&(0.024 )&(0.043 )&\textit{std}\\
\end{tabular}} 
\begin{tablenotes}\tiny 
\item CF - control function. $(0)$ - none, $(1)$ - $\lambda_a\hat{a}_i$,  $(2)$ - $\hat{h}(\hat{a}_i)$, $(3)$ - $\hat{h}(a_i)$, $(4)$ - $\hat{h}(\widehat{deg}_i,x_{2i})$, $(5)$ - $h(a_i)$. 
\item The network design parameters are $\mu_0=1.00$, $\mu_1=1.00$, $\alpha_L=0.00$, $\alpha_H=0.00$ 
\item Average number of links for $N=100$ is $49.5$, for $N=250$ it is $124.6$. 
\item Average skewness for $N=100$ is $-0.01$, for $N=250$ it is $-0.01$. 
\item Size is the empirical size of t-test against the truth. 
\item N$=100$, $corr(a_i,\bm{x}_{2i})=-0.001$,N$=250$, $corr(a_i,\bm{x}_{2i})=-0.002$ 
\item The bias of $\hat{a}_i$ is calculated as $a_i-\hat{a}_i$. 
\item For $N=100$, $\hat{a}_i$ mean bias$=-0.001$, median bias$=-0.001$, std$=0.232$. 
\item For $N=250$, $\hat{a}_i$ mean bias$=-0.000$, median bias$=-0.000$, std$=0.145$. 
  \end{tablenotes} 
\end{threeparttable} 
\end{table} 
\begin{table}[!h]\caption{\footnotesize {\bf Design 3 dense network: Parameter values across 1000 Monte Carlo replications with $K_N=4$ and Hermite polynomial sieve}} 
\begin{threeparttable} 
\centering \footnotesize
\scalebox{.8}{\begin{tabular}{|c|c|c|c|c|c|c|c|c|c|c|c|c|c|c|c|}\toprule 
\multicolumn{14}{c}{$h(a_i) = \exp(a_i)$}\\ 
\cellcolor{yellow}$N$&\multicolumn{6}{|c|}{\cellcolor{yellow}$100$}&\multicolumn{6}{|c|}{\cellcolor{yellow}$250$}&\\\hline 
CF&$(0)$&$(1)$&$(2)$&$(3)$&$(4)$&$(5)$& $(0)$ &$(1)$&$(2)$&$(3)$&$(4)$&$(5)$&\\\hline 
\multirow{4}{*}{$\beta_1=0.8$}& 0.008 & 0.012 &-0.001 &-0.000 &-0.000 &0.000& 0.014 &0.024 &0.001& 0.001 & 0.002& -0.000 & \textit{mean bias} \\ 
&(0.037 )&(0.043 )&(0.037 )&(0.035 )&(0.056 )&(0.017 )&(0.050 )&(0.058 )&(0.035 )&(0.034 )&(0.053 )&(0.014 )&\textit{std}\\ 
& 0.823 & 0.830 &0.054 &0.062 &0.074 &0.047& 0.926 &0.927 &0.049& 0.056 &0.064& 0.052 &\textit{size} \\ \midrule
\multirow{4}{*}{$\beta_2=5$}& -0.005 & -0.007 &-0.001 &-0.002 &-0.001 &-0.002& -0.002 &-0.005 &0.000& -0.000 &0.000& 0.000 &\textit{mean bias} \\ 
&(0.036 )&(0.036 )&(0.035 )&(0.033 )&(0.036 )&(0.030 )&(0.024 )&(0.023 )&(0.021 )&(0.020 )&(0.021 )&(0.019 )&\textit{std}\\
\end{tabular}} 
\begin{tablenotes}\tiny 
\item CF - control function. $(0)$ - none, $(1)$ - $\lambda_a\hat{a}_i$,  $(2)$ - $\hat{h}(\hat{a}_i)$, $(3)$ - $\hat{h}(a_i)$, $(4)$ - $\hat{h}(\widehat{deg}_i,x_{2i})$, $(5)$ - $h(a_i)$. 
\item The network design parameters are $\mu_0=1.00$, $\mu_1=1.00$, $\alpha_L=-0.25$, $\alpha_H=-0.25$ 
\item Average number of links for $N=100$ is $40.0$, for $N=250$ it is $100.7$. 
\item Average skewness for $N=100$ is $0.05$, for $N=250$ it is $0.05$. 
\item Size is the empirical size of t-test against the truth. 
\item N$=100$, $corr(a_i,\bm{x}_{2i})=-0.001$,N$=250$, $corr(a_i,\bm{x}_{2i})=-0.002$ 
\item The bias of $\hat{a}_i$ is calculated as $a_i-\hat{a}_i$. 
\item For $N=100$, $\hat{a}_i$ mean bias$=0.003$, median bias$=0.000$, std$=0.236$. 
\item For $N=250$, $\hat{a}_i$ mean bias$=0.001$, median bias$=0.001$, std$=0.147$. 
  \end{tablenotes} 
\end{threeparttable} 
\end{table} 
\begin{table}[!h]\caption{\footnotesize {\bf Design 4 dense network: Parameter values across 1000 Monte Carlo replications with $K_N=4$ and Hermite polynomial sieve}} \label{table: MC dense}
\begin{threeparttable} 
\centering \footnotesize
\scalebox{.8}{

	\begin{tabular}{|c|c|c|c|c|c|c|c|c|c|c|c|c|c|}
		\toprule 
\multicolumn{14}{c}{$h(a_i) = \exp(a_i)$}\\ 
\cellcolor{yellow}$N$&\multicolumn{6}{|c|}{\cellcolor{yellow}$100$}&\multicolumn{6}{|c|}{\cellcolor{yellow}$250$}&\\\hline 
CF&$(0)$&$(1)$&$(2)$&$(3)$&$(4)$&$(5)$& $(0)$ &$(1)$&$(2)$&$(3)$&$(4)$&$(5)$&\\\hline 
\multirow{4}{*}{$\beta_1=0.8$}& 0.002 & 0.004 &-0.000 &-0.000 &0.000 &-0.000& 0.004 &0.007 &-0.001& -0.001 & -0.001& -0.000 & \textit{mean bias} \\ 
&(0.010 )&(0.013 )&(0.015 )&(0.015 )&(0.024 )&(0.010 )&(0.009 )&(0.013 )&(0.015 )&(0.015 )&(0.025 )&(0.009 )&\textit{std}\\ 
& 0.133 & 0.115 &0.056 &0.061 &0.058 &0.058& 0.306 &0.225 &0.057& 0.057 &0.064& 0.050 &\textit{size} \\ \midrule
\multirow{4}{*}{$\beta_2=5$}& -0.003 & -0.004 &-0.000 &-0.000 &0.000 &-0.000& -0.002 &-0.004 &0.000& 0.000 &-0.000& 0.000 &\textit{mean bias} \\ 
&(0.031 )&(0.032 )&(0.034 )&(0.033 )&(0.035 )&(0.031 )&(0.020 )&(0.021 )&(0.020 )&(0.020 )&(0.021 )&(0.020 )&\textit{std}\\
\end{tabular}} 
\begin{tablenotes}\tiny 
\item CF - control function. $(0)$ - none, $(1)$ - $\lambda_a\hat{a}_i$,  $(2)$ - $\hat{h}(\hat{a}_i)$, $(3)$ - $\hat{h}(a_i)$, $(4)$ - $\hat{h}(\widehat{deg}_i,x_{2i})$, $(5)$ - $h(a_i)$. 
\item The network design parameters are $\mu_0=0.25$, $\mu_1=0.75$, $\alpha_L=-0.75$, $\alpha_H=-0.75$ 
\item Average number of links for $N=100$ is $23.0$, for $N=250$ it is $57.8$. 
\item Average skewness for $N=100$ is $0.66$, for $N=250$ it is $0.89$. 
\item Size is the empirical size of t-test against the truth. 
\item N$=100$, $corr(a_i,\bm{x}_{2i})=0.004$,N$=250$, $corr(a_i,\bm{x}_{2i})=0.001$ 
\item The bias of $\hat{a}_i$ is calculated as $a_i-\hat{a}_i$. 
\item For $N=100$, $\hat{a}_i$ mean bias$=0.018$, median bias$=0.008$, std$=0.271$. 
\item For $N=250$, $\hat{a}_i$ mean bias$=0.007$, median bias$=0.004$, std$=0.167$. 
  \end{tablenotes} 
\end{threeparttable} 
\end{table} 
\begin{table}[!h]\caption{\footnotesize {\bf Design 5 dense network: Parameter values across 1000 Monte Carlo replications with $K_N=4$ and Hermite polynomial sieve}} 
\begin{threeparttable} 
\centering \footnotesize
\scalebox{.8}{\begin{tabular}{|c|c|c|c|c|c|c|c|c|c|c|c|c|c|}\toprule 
\multicolumn{14}{c}{$h(a_i) = \exp(a_i)$}\\ 
\cellcolor{yellow}$N$&\multicolumn{6}{|c|}{\cellcolor{yellow}$100$}&\multicolumn{6}{|c|}{\cellcolor{yellow}$250$}&\\\hline 
CF&$(0)$&$(1)$&$(2)$&$(3)$&$(4)$&$(5)$& $(0)$ &$(1)$&$(2)$&$(3)$&$(4)$&$(5)$&\\\hline 
\multirow{4}{*}{$\beta_1=0.8$}& 0.020 & 0.033 &0.011 &-0.002 &0.004 &0.000& 0.022 &0.049 &0.008& -0.002 & -0.002& -0.000 & \textit{mean bias} \\ 
&(0.036 )&(0.049 )&(0.045 )&(0.036 )&(0.060 )&(0.011 )&(0.051 )&(0.063 )&(0.038 )&(0.034 )&(0.056 )&(0.007 )&\textit{std}\\ 
& 0.898 & 0.945 &0.049 &0.066 &0.041 &0.065& 0.987 &0.991 &0.042& 0.061 &0.046& 0.044 &\textit{size} \\ \midrule
\multirow{4}{*}{$\beta_2=5$}& 0.039 & 0.004 &0.003 &-0.000 &0.000 &-0.000& 0.031 &0.006 &0.001& -0.000 &-0.000& -0.000 &\textit{mean bias} \\ 
&(0.057 )&(0.053 )&(0.045 )&(0.033 )&(0.047 )&(0.030 )&(0.038 )&(0.031 )&(0.023 )&(0.020 )&(0.023 )&(0.020 )&\textit{std}\\
\end{tabular}} 
\begin{tablenotes}\tiny 
\item CF - control function. $(0)$ - none, $(1)$ - $\lambda_a\hat{a}_i$,  $(2)$ - $\hat{h}(\hat{a}_i)$, $(3)$ - $\hat{h}(a_i)$, $(4)$ - $\hat{h}(\widehat{deg}_i,x_{2i})$, $(5)$ - $h(a_i)$. 
\item The network design parameters are $\mu_0=0.25$, $\mu_1=0.75$, $\alpha_L=-0.50$, $\alpha_H=0.00$ 
\item Average number of links for $N=100$ is $39.8$, for $N=250$ it is $99.9$. 
\item Average skewness for $N=100$ is $0.16$, for $N=250$ it is $0.21$. 
\item Size is the empirical size of t-test against the truth. 
\item N$=100$, $corr(a_i,\bm{x}_{2i})=0.634$,N$=250$, $corr(a_i,\bm{x}_{2i})=0.632$ 
\item The bias of $\hat{a}_i$ is calculated as $a_i-\hat{a}_i$. 
\item For $N=100$, $\hat{a}_i$ mean bias$=0.002$, median bias$=0.000$, std$=0.241$. 
\item For $N=250$, $\hat{a}_i$ mean bias$=0.001$, median bias$=0.001$, std$=0.150$. 
  \end{tablenotes} 
\end{threeparttable} 
\end{table} 
\begin{table}[!h]\caption{\footnotesize {\bf Design 6 dense network: Parameter values across 1000 Monte Carlo replications with $K_N=4$ and Hermite polynomial sieve}} 
\begin{threeparttable} 
\centering \footnotesize
\scalebox{.8}{\begin{tabular}{|c|c|c|c|c|c|c|c|c|c|c|c|c|c|c|c|}\toprule 
\multicolumn{14}{c}{$h(a_i) = \exp(a_i)$}\\ 
\cellcolor{yellow}$N$&\multicolumn{6}{|c|}{\cellcolor{yellow}$100$}&\multicolumn{6}{|c|}{\cellcolor{yellow}$250$}&\\\hline 
CF&$(0)$&$(1)$&$(2)$&$(3)$&$(4)$&$(5)$& $(0)$ &$(1)$&$(2)$&$(3)$&$(4)$&$(5)$&\\\hline 
\multirow{4}{*}{$\beta_1=0.8$}& 0.014 & 0.026 &0.005 &-0.002 &0.001 &0.000& 0.016 &0.038 &0.003& -0.001 & -0.002& -0.000 & \textit{mean bias} \\ 
&(0.016 )&(0.020 )&(0.030 )&(0.028 )&(0.041 )&(0.008 )&(0.013 )&(0.016 )&(0.027 )&(0.026 )&(0.041 )&(0.005 )&\textit{std}\\ 
& 0.798 & 0.899 &0.051 &0.077 &0.051 &0.056& 0.984 &0.986 &0.046& 0.054 &0.059& 0.038 &\textit{size} \\ \midrule
\multirow{4}{*}{$\beta_2=5$}& 0.019 & -0.004 &0.002 &-0.000 &0.000 &-0.000& 0.015 &-0.002 &0.001& -0.000 &0.000& -0.000 &\textit{mean bias} \\ 
&(0.042 )&(0.040 )&(0.038 )&(0.033 )&(0.039 )&(0.030 )&(0.027 )&(0.024 )&(0.022 )&(0.020 )&(0.022 )&(0.020 )&\textit{std}\\
\end{tabular}} 
\begin{tablenotes}\tiny 
\item CF - control function. $(0)$ - none, $(1)$ - $\lambda_a\hat{a}_i$,  $(2)$ - $\hat{h}(\hat{a}_i)$, $(3)$ - $\hat{h}(a_i)$, $(4)$ - $\hat{h}(\widehat{deg}_i,x_{2i})$, $(5)$ - $h(a_i)$. 
\item The network design parameters are $\mu_0=0.25$, $\mu_1=0.75$, $\alpha_L=-0.67$, $\alpha_H=-0.17$ 
\item Average number of links for $N=100$ is $33.9$, for $N=250$ it is $85.1$. 
\item Average skewness for $N=100$ is $0.20$, for $N=250$ it is $0.25$. 
\item Size is the empirical size of t-test against the truth. 
\item N$=100$, $corr(a_i,\bm{x}_{2i})=0.634$,N$=250$, $corr(a_i,\bm{x}_{2i})=0.632$ 
\item The bias of $\hat{a}_i$ is calculated as $a_i-\hat{a}_i$. 
\item For $N=100$, $\hat{a}_i$ mean bias$=0.006$, median bias$=0.002$, std$=0.248$. 
\item For $N=250$, $\hat{a}_i$ mean bias$=0.003$, median bias$=0.001$, std$=0.154$. 
  \end{tablenotes} 
\end{threeparttable} 
\end{table} 
\begin{table}[!h]\caption{\footnotesize {\bf Design 7 dense network: Parameter values across 1000 Monte Carlo replications with $K_N=4$ and Hermite polynomial sieve}} 
\begin{threeparttable} 
\centering \footnotesize
\scalebox{.8}{\begin{tabular}{|c|c|c|c|c|c|c|c|c|c|c|c|c|c|c|c|}\toprule 
\multicolumn{14}{c}{$h(a_i) = \exp(a_i)$}\\ 
\cellcolor{yellow}$N$&\multicolumn{6}{|c|}{\cellcolor{yellow}$100$}&\multicolumn{6}{|c|}{\cellcolor{yellow}$250$}&\\\hline 
CF&$(0)$&$(1)$&$(2)$&$(3)$&$(4)$&$(5)$& $(0)$ &$(1)$&$(2)$&$(3)$&$(4)$&$(5)$&\\\hline 
\multirow{4}{*}{$\beta_1=0.8$}& 0.020 & 0.033 &0.011 &-0.002 &0.004 &0.000& 0.022 &0.049 &0.008& -0.002 & -0.002& -0.000 & \textit{mean bias} \\ 
&(0.036 )&(0.049 )&(0.045 )&(0.036 )&(0.060 )&(0.011 )&(0.051 )&(0.063 )&(0.038 )&(0.034 )&(0.056 )&(0.007 )&\textit{std}\\ 
& 0.898 & 0.945 &0.049 &0.066 &0.041 &0.065& 0.987 &0.991 &0.042& 0.061 &0.046& 0.044 &\textit{size} \\ \midrule
\multirow{4}{*}{$\beta_2=5$}& 0.039 & 0.004 &0.003 &-0.000 &0.000 &-0.000& 0.031 &0.006 &0.001& -0.000 &-0.000& -0.000 &\textit{mean bias} \\ 
&(0.057 )&(0.053 )&(0.045 )&(0.033 )&(0.047 )&(0.030 )&(0.038 )&(0.031 )&(0.023 )&(0.020 )&(0.023 )&(0.020 )&\textit{std}\\
\end{tabular}} 
\begin{tablenotes}\tiny 
\item CF - control function. $(0)$ - none, $(1)$ - $\lambda_a\hat{a}_i$,  $(2)$ - $\hat{h}(\hat{a}_i)$, $(3)$ - $\hat{h}(a_i)$, $(4)$ - $\hat{h}(\widehat{deg}_i,x_{2i})$, $(5)$ - $h(a_i)$. 
\item The network design parameters are $\mu_0=0.25$, $\mu_1=0.75$, $\alpha_L=-0.50$, $\alpha_H=0.00$ 
\item Average number of links for $N=100$ is $39.8$, for $N=250$ it is $99.9$. 
\item Average skewness for $N=100$ is $0.16$, for $N=250$ it is $0.21$. 
\item Size is the empirical size of t-test against the truth. 
\item N$=100$, $corr(a_i,\bm{x}_{2i})=0.634$,N$=250$, $corr(a_i,\bm{x}_{2i})=0.632$ 
\item The bias of $\hat{a}_i$ is calculated as $a_i-\hat{a}_i$. 
\item For $N=100$, $\hat{a}_i$ mean bias$=0.002$, median bias$=0.000$, std$=0.241$. 
\item For $N=250$, $\hat{a}_i$ mean bias$=0.001$, median bias$=0.001$, std$=0.150$. 
  \end{tablenotes} 
\end{threeparttable} 
\end{table} 
\begin{table}[!h]\caption{\footnotesize {\bf Design 8 dense network: Parameter values across 1000 Monte Carlo replications with $K_N=4$ and Hermite polynomial sieve}} 
\begin{threeparttable} 
\centering \footnotesize
\scalebox{.8}{\begin{tabular}{|c|c|c|c|c|c|c|c|c|c|c|c|c|c|c|c|}\toprule 
\multicolumn{14}{c}{$h(a_i) = \exp(a_i)$}\\ 
\cellcolor{yellow}$N$&\multicolumn{6}{|c|}{\cellcolor{yellow}$100$}&\multicolumn{6}{|c|}{\cellcolor{yellow}$250$}&\\\hline 
CF&$(0)$&$(1)$&$(2)$&$(3)$&$(4)$&$(5)$& $(0)$ &$(1)$&$(2)$&$(3)$&$(4)$&$(5)$&\\\hline 
\multirow{4}{*}{$\beta_1=0.8$}& 0.006 & 0.011 &0.001 &-0.000 &0.001 &-0.000& 0.009 &0.019 &0.000& -0.001 & -0.001& -0.000 & \textit{mean bias} \\ 
&(0.009 )&(0.013 )&(0.019 )&(0.019 )&(0.028 )&(0.009 )&(0.008 )&(0.011 )&(0.018 )&(0.018 )&(0.029 )&(0.006 )&\textit{std}\\ 
& 0.325 & 0.397 &0.055 &0.063 &0.051 &0.052& 0.759 &0.765 &0.052& 0.050 &0.060& 0.042 &\textit{size} \\ \midrule
\multirow{4}{*}{$\beta_2=5$}& 0.001 & -0.006 &0.000 &-0.000 &0.000 &-0.000& 0.000 &-0.005 &0.000& 0.000 &-0.000& -0.000 &\textit{mean bias} \\ 
&(0.032 )&(0.033 )&(0.034 )&(0.033 )&(0.035 )&(0.030 )&(0.021 )&(0.021 )&(0.020 )&(0.020 )&(0.021 )&(0.020 )&\textit{std}\\
\end{tabular}} 
\begin{tablenotes}\tiny 
\item CF - control function. $(0)$ - none, $(1)$ - $\lambda_a\hat{a}_i$,  $(2)$ - $\hat{h}(\hat{a}_i)$, $(3)$ - $\hat{h}(a_i)$, $(4)$ - $\hat{h}(\widehat{deg}_i,x_{2i})$, $(5)$ - $h(a_i)$. 
\item The network design parameters are $\mu_0=0.25$, $\mu_1=0.75$, $\alpha_L=-0.75$, $\alpha_H=-0.50$ 
\item Average number of links for $N=100$ is $26.9$, for $N=250$ it is $67.6$. 
\item Average skewness for $N=100$ is $0.50$, for $N=250$ it is $0.64$. 
\item Size is the empirical size of t-test against the truth. 
\item N$=100$, $corr(a_i,\bm{x}_{2i})=0.380$,N$=250$, $corr(a_i,\bm{x}_{2i})=0.377$ 
\item The bias of $\hat{a}_i$ is calculated as $a_i-\hat{a}_i$. 
\item For $N=100$, $\hat{a}_i$ mean bias$=0.012$, median bias$=0.004$, std$=0.260$. 
\item For $N=250$, $\hat{a}_i$ mean bias$=0.005$, median bias$=0.003$, std$=0.161$. 
  \end{tablenotes} 
\end{threeparttable} 
\end{table} 
\clearpage
\subsection{Dense Network, $K_N=8$, Hermite polynomial sieve}\label{section: dense kn8 hermite}
\begin{table}[!h]\caption{\footnotesize {\bf Design 1 dense network: Parameter values across 1000 Monte Carlo replications with $K_N=8$ and Hermite polynomial sieve}} 
\begin{threeparttable} 
\centering \footnotesize
\scalebox{.8}{\begin{tabular}{|c|c|c|c|c|c|c|c|c|c|c|c|c|c|c|c|}\toprule 
\multicolumn{14}{c}{$h(a_i) = \exp(a_i)$}\\ 
\cellcolor{yellow}$N$&\multicolumn{6}{|c|}{\cellcolor{yellow}$100$}&\multicolumn{6}{|c|}{\cellcolor{yellow}$250$}&\\\hline 
CF&$(0)$&$(1)$&$(2)$&$(3)$&$(4)$&$(5)$& $(0)$ &$(1)$&$(2)$&$(3)$&$(4)$&$(5)$&\\\hline 
\multirow{4}{*}{$\beta_1=0.8$}& 0.005 & 0.007 &-0.000 &0.000 &-0.000 &0.000& 0.009 &0.015 &0.000& 0.001 & 0.000& 0.000 & \textit{mean bias} \\ 
&(0.016 )&(0.018 )&(0.024 )&(0.023 )&(0.035 )&(0.014 )&(0.017 )&(0.020 )&(0.022 )&(0.022 )&(0.033 )&(0.012 )&\textit{std}\\ 
& 0.403 & 0.333 &0.070 &0.068 &0.069 &0.050& 0.773 &0.656 &0.054& 0.065 &0.042& 0.062 &\textit{size} \\ \midrule
\multirow{4}{*}{$\beta_2=5$}& -0.005 & -0.006 &-0.001 &-0.002 &-0.002 &-0.002& -0.003 &-0.005 &0.000& -0.000 &0.000& 0.000 &\textit{mean bias} \\ 
&(0.031 )&(0.032 )&(0.034 )&(0.034 )&(0.034 )&(0.031 )&(0.020 )&(0.020 )&(0.020 )&(0.020 )&(0.021 )&(0.019 )&\textit{std}\\
\end{tabular}} 
\begin{tablenotes}\tiny 
\item CF - control function. $(0)$ - none, $(1)$ - $\lambda_a\hat{a}_i$,  $(2)$ - $\hat{h}(\hat{a}_i)$, $(3)$ - $\hat{h}(a_i)$, $(4)$ - $\hat{h}(\widehat{deg}_i,x_{2i})$, $(5)$ - $h(a_i)$. 
\item The network design parameters are $\mu_0=1.00$, $\mu_1=1.00$, $\alpha_L=-0.50$, $\alpha_H=-0.50$ 
\item Average number of links for $N=100$ is $31.0$, for $N=250$ it is $77.9$. 
\item Average skewness for $N=100$ is $0.12$, for $N=250$ it is $0.11$. 
\item Size is the empirical size of t-test against the truth. 
\item N$=100$, $corr(a_i,\bm{x}_{2i})=-0.001$,N$=250$, $corr(a_i,\bm{x}_{2i})=-0.002$ 
\item The bias of $\hat{a}_i$ is calculated as $a_i-\hat{a}_i$. 
\item For $N=100$, $\hat{a}_i$ mean bias$=0.009$, median bias$=0.004$, std$=0.249$. 
\item For $N=250$, $\hat{a}_i$ mean bias$=0.004$, median bias$=0.002$, std$=0.154$. 
  \end{tablenotes} 
\end{threeparttable} 
\end{table} 
\begin{table}[!h]\caption{\footnotesize {\bf Design 2 dense network: Parameter values across 1000 Monte Carlo replications with $K_N=8$ and Hermite polynomial sieve}} 
\begin{threeparttable} 
\centering \footnotesize
\scalebox{.8}{\begin{tabular}{|c|c|c|c|c|c|c|c|c|c|c|c|c|c|c|c|}\toprule 
\multicolumn{14}{c}{$h(a_i) = \exp(a_i)$}\\ 
\cellcolor{yellow}$N$&\multicolumn{6}{|c|}{\cellcolor{yellow}$100$}&\multicolumn{6}{|c|}{\cellcolor{yellow}$250$}&\\\hline 
CF&$(0)$&$(1)$&$(2)$&$(3)$&$(4)$&$(5)$& $(0)$ &$(1)$&$(2)$&$(3)$&$(4)$&$(5)$&\\\hline 
\multirow{4}{*}{$\beta_1=0.8$}& -0.038 & -0.049 &-0.001 &-0.000 &-0.001 &-0.000& 5.864 &0.165 &0.001& 0.002 & 0.002& 0.011 & \textit{mean bias} \\ 
&(0.693 )&(0.783 )&(0.072 )&(0.059 )&(0.104 )&(0.036 )&(186.308 )&(8.967 )&(0.062 )&(0.057 )&(0.092 )&(0.243 )&\textit{std}\\ 
& 0.896 & 0.904 &0.061 &0.069 &0.069 &0.051& 0.946 &0.949 &0.053& 0.074 &0.051& 0.055 &\textit{size} \\ \midrule
\multirow{4}{*}{$\beta_2=5$}& 0.009 & 0.011 &-0.001 &-0.002 &-0.001 &-0.002& -0.650 &-0.022 &0.000& -0.000 &0.001& -0.002 &\textit{mean bias} \\ 
&(0.174 )&(0.195 )&(0.043 )&(0.034 )&(0.043 )&(0.031 )&(20.492 )&(1.049 )&(0.024 )&(0.020 )&(0.024 )&(0.043 )&\textit{std}\\
\end{tabular}} 
\begin{tablenotes}\tiny 
\item CF - control function. $(0)$ - none, $(1)$ - $\lambda_a\hat{a}_i$,  $(2)$ - $\hat{h}(\hat{a}_i)$, $(3)$ - $\hat{h}(a_i)$, $(4)$ - $\hat{h}(\widehat{deg}_i,x_{2i})$, $(5)$ - $h(a_i)$. 
\item The network design parameters are $\mu_0=1.00$, $\mu_1=1.00$, $\alpha_L=0.00$, $\alpha_H=0.00$ 
\item Average number of links for $N=100$ is $49.5$, for $N=250$ it is $124.6$. 
\item Average skewness for $N=100$ is $-0.01$, for $N=250$ it is $-0.01$. 
\item Size is the empirical size of t-test against the truth. 
\item N$=100$, $corr(a_i,\bm{x}_{2i})=-0.001$,N$=250$, $corr(a_i,\bm{x}_{2i})=-0.002$ 
\item The bias of $\hat{a}_i$ is calculated as $a_i-\hat{a}_i$. 
\item For $N=100$, $\hat{a}_i$ mean bias$=-0.001$, median bias$=-0.001$, std$=0.232$. 
\item For $N=250$, $\hat{a}_i$ mean bias$=-0.000$, median bias$=-0.000$, std$=0.145$. 
  \end{tablenotes} 
\end{threeparttable} 
\end{table} 
\begin{table}[!h]\caption{\footnotesize {\bf Design 3 dense network: Parameter values across 1000 Monte Carlo replications with $K_N=8$ and Hermite polynomial sieve}} 
\begin{threeparttable} 
\centering \footnotesize
\scalebox{.8}{\begin{tabular}{|c|c|c|c|c|c|c|c|c|c|c|c|c|c|c|c|}\toprule 
\multicolumn{14}{c}{$h(a_i) = \exp(a_i)$}\\ 
\cellcolor{yellow}$N$&\multicolumn{6}{|c|}{\cellcolor{yellow}$100$}&\multicolumn{6}{|c|}{\cellcolor{yellow}$250$}&\\\hline 
CF&$(0)$&$(1)$&$(2)$&$(3)$&$(4)$&$(5)$& $(0)$ &$(1)$&$(2)$&$(3)$&$(4)$&$(5)$&\\\hline 
\multirow{4}{*}{$\beta_1=0.8$}& 0.008 & 0.012 &-0.001 &-0.000 &-0.000 &0.000& 0.014 &0.024 &0.001& 0.001 & 0.002& -0.000 & \textit{mean bias} \\ 
&(0.037 )&(0.043 )&(0.038 )&(0.035 )&(0.057 )&(0.017 )&(0.050 )&(0.058 )&(0.035 )&(0.034 )&(0.053 )&(0.014 )&\textit{std}\\ 
& 0.823 & 0.830 &0.059 &0.065 &0.079 &0.047& 0.926 &0.927 &0.048& 0.060 &0.063& 0.052 &\textit{size} \\ \midrule
\multirow{4}{*}{$\beta_2=5$}& -0.005 & -0.007 &-0.001 &-0.002 &-0.001 &-0.002& -0.002 &-0.005 &-0.000& -0.000 &0.000& 0.000 &\textit{mean bias} \\ 
&(0.036 )&(0.036 )&(0.035 )&(0.033 )&(0.036 )&(0.030 )&(0.024 )&(0.023 )&(0.021 )&(0.020 )&(0.021 )&(0.019 )&\textit{std}\\
\end{tabular}} 
\begin{tablenotes}\tiny 
\item CF - control function. $(0)$ - none, $(1)$ - $\lambda_a\hat{a}_i$,  $(2)$ - $\hat{h}(\hat{a}_i)$, $(3)$ - $\hat{h}(a_i)$, $(4)$ - $\hat{h}(\widehat{deg}_i,x_{2i})$, $(5)$ - $h(a_i)$. 
\item The network design parameters are $\mu_0=1.00$, $\mu_1=1.00$, $\alpha_L=-0.25$, $\alpha_H=-0.25$ 
\item Average number of links for $N=100$ is $40.0$, for $N=250$ it is $100.7$. 
\item Average skewness for $N=100$ is $0.05$, for $N=250$ it is $0.05$. 
\item Size is the empirical size of t-test against the truth. 
\item N$=100$, $corr(a_i,\bm{x}_{2i})=-0.001$,N$=250$, $corr(a_i,\bm{x}_{2i})=-0.002$ 
\item The bias of $\hat{a}_i$ is calculated as $a_i-\hat{a}_i$. 
\item For $N=100$, $\hat{a}_i$ mean bias$=0.003$, median bias$=0.000$, std$=0.236$. 
\item For $N=250$, $\hat{a}_i$ mean bias$=0.001$, median bias$=0.001$, std$=0.147$. 
  \end{tablenotes} 
\end{threeparttable} 
\end{table} 
\begin{table}[!h]\caption{\footnotesize {\bf Design 4 dense network: Parameter values across 1000 Monte Carlo replications with $K_N=8$ and Hermite polynomial sieve}} 
\begin{threeparttable} 
\centering \footnotesize
\scalebox{.8}{\begin{tabular}{|c|c|c|c|c|c|c|c|c|c|c|c|c|c|c|c|}\toprule 
\multicolumn{14}{c}{$h(a_i) = \exp(a_i)$}\\ 
\cellcolor{yellow}$N$&\multicolumn{6}{|c|}{\cellcolor{yellow}$100$}&\multicolumn{6}{|c|}{\cellcolor{yellow}$250$}&\\\hline 
CF&$(0)$&$(1)$&$(2)$&$(3)$&$(4)$&$(5)$& $(0)$ &$(1)$&$(2)$&$(3)$&$(4)$&$(5)$&\\\hline 
\multirow{4}{*}{$\beta_1=0.8$}& 0.002 & 0.004 &-0.000 &-0.000 &0.000 &-0.000& 0.004 &0.007 &-0.000& -0.001 & -0.001& -0.000 & \textit{mean bias} \\ 
&(0.010 )&(0.013 )&(0.016 )&(0.016 )&(0.024 )&(0.010 )&(0.009 )&(0.013 )&(0.015 )&(0.015 )&(0.025 )&(0.009 )&\textit{std}\\ 
& 0.133 & 0.115 &0.059 &0.062 &0.056 &0.058& 0.306 &0.225 &0.058& 0.054 &0.068& 0.050 &\textit{size} \\ \midrule
\multirow{4}{*}{$\beta_2=5$}& -0.003 & -0.004 &-0.000 &-0.000 &0.000 &-0.000& -0.002 &-0.004 &0.000& 0.000 &0.000& 0.000 &\textit{mean bias} \\ 
&(0.031 )&(0.032 )&(0.034 )&(0.034 )&(0.035 )&(0.031 )&(0.020 )&(0.021 )&(0.021 )&(0.021 )&(0.021 )&(0.020 )&\textit{std}\\
\end{tabular}} 
\begin{tablenotes}\tiny 
\item CF - control function. $(0)$ - none, $(1)$ - $\lambda_a\hat{a}_i$,  $(2)$ - $\hat{h}(\hat{a}_i)$, $(3)$ - $\hat{h}(a_i)$, $(4)$ - $\hat{h}(\widehat{deg}_i,x_{2i})$, $(5)$ - $h(a_i)$. 
\item The network design parameters are $\mu_0=0.25$, $\mu_1=0.75$, $\alpha_L=-0.75$, $\alpha_H=-0.75$ 
\item Average number of links for $N=100$ is $23.0$, for $N=250$ it is $57.8$. 
\item Average skewness for $N=100$ is $0.66$, for $N=250$ it is $0.89$. 
\item Size is the empirical size of t-test against the truth. 
\item N$=100$, $corr(a_i,\bm{x}_{2i})=0.004$,N$=250$, $corr(a_i,\bm{x}_{2i})=0.001$ 
\item The bias of $\hat{a}_i$ is calculated as $a_i-\hat{a}_i$. 
\item For $N=100$, $\hat{a}_i$ mean bias$=0.018$, median bias$=0.008$, std$=0.271$. 
\item For $N=250$, $\hat{a}_i$ mean bias$=0.007$, median bias$=0.004$, std$=0.167$. 
  \end{tablenotes} 
\end{threeparttable} 
\end{table} 
\begin{table}[!h]\caption{\footnotesize {\bf Design 5 dense network: Parameter values across 1000 Monte Carlo replications with $K_N=8$ and Hermite polynomial sieve}} 
\begin{threeparttable} 
\centering \footnotesize
\scalebox{.8}{\begin{tabular}{|c|c|c|c|c|c|c|c|c|c|c|c|c|c|c|c|}\toprule 
\multicolumn{14}{c}{$h(a_i) = \exp(a_i)$}\\ 
\cellcolor{yellow}$N$&\multicolumn{6}{|c|}{\cellcolor{yellow}$100$}&\multicolumn{6}{|c|}{\cellcolor{yellow}$250$}&\\\hline 
CF&$(0)$&$(1)$&$(2)$&$(3)$&$(4)$&$(5)$& $(0)$ &$(1)$&$(2)$&$(3)$&$(4)$&$(5)$&\\\hline 
\multirow{4}{*}{$\beta_1=0.8$}& 0.020 & 0.033 &0.011 &-0.000 &0.005 &0.000& 0.022 &0.049 &0.005& -0.000 & -0.002& -0.000 & \textit{mean bias} \\ 
&(0.036 )&(0.049 )&(0.045 )&(0.036 )&(0.061 )&(0.011 )&(0.051 )&(0.063 )&(0.039 )&(0.034 )&(0.056 )&(0.007 )&\textit{std}\\ 
& 0.898 & 0.945 &0.061 &0.055 &0.044 &0.065& 0.987 &0.991 &0.047& 0.060 &0.048& 0.044 &\textit{size} \\ \midrule
\multirow{4}{*}{$\beta_2=5$}& 0.039 & 0.004 &0.003 &-0.000 &-0.000 &-0.000& 0.031 &0.006 &0.001& -0.000 &-0.000& -0.000 &\textit{mean bias} \\ 
&(0.057 )&(0.053 )&(0.045 )&(0.034 )&(0.046 )&(0.030 )&(0.038 )&(0.031 )&(0.023 )&(0.021 )&(0.023 )&(0.020 )&\textit{std}\\
\end{tabular}} 
\begin{tablenotes}\tiny 
\item CF - control function. $(0)$ - none, $(1)$ - $\lambda_a\hat{a}_i$,  $(2)$ - $\hat{h}(\hat{a}_i)$, $(3)$ - $\hat{h}(a_i)$, $(4)$ - $\hat{h}(\widehat{deg}_i,x_{2i})$, $(5)$ - $h(a_i)$. 
\item The network design parameters are $\mu_0=0.25$, $\mu_1=0.75$, $\alpha_L=-0.50$, $\alpha_H=0.00$ 
\item Average number of links for $N=100$ is $39.8$, for $N=250$ it is $99.9$. 
\item Average skewness for $N=100$ is $0.16$, for $N=250$ it is $0.21$. 
\item Size is the empirical size of t-test against the truth. 
\item N$=100$, $corr(a_i,\bm{x}_{2i})=0.634$,N$=250$, $corr(a_i,\bm{x}_{2i})=0.632$ 
\item The bias of $\hat{a}_i$ is calculated as $a_i-\hat{a}_i$. 
\item For $N=100$, $\hat{a}_i$ mean bias$=0.002$, median bias$=0.000$, std$=0.241$. 
\item For $N=250$, $\hat{a}_i$ mean bias$=0.001$, median bias$=0.001$, std$=0.150$. 
  \end{tablenotes} 
\end{threeparttable} 
\end{table} 
\begin{table}[!h]\caption{\footnotesize {\bf Design 6 dense network: Parameter values across 1000 Monte Carlo replications with $K_N=8$ and Hermite polynomial sieve}} 
\begin{threeparttable} 
\centering \footnotesize
\scalebox{.8}{\begin{tabular}{|c|c|c|c|c|c|c|c|c|c|c|c|c|c|c|c|}\toprule 
\multicolumn{14}{c}{$h(a_i) = \exp(a_i)$}\\ 
\cellcolor{yellow}$N$&\multicolumn{6}{|c|}{\cellcolor{yellow}$100$}&\multicolumn{6}{|c|}{\cellcolor{yellow}$250$}&\\\hline 
CF&$(0)$&$(1)$&$(2)$&$(3)$&$(4)$&$(5)$& $(0)$ &$(1)$&$(2)$&$(3)$&$(4)$&$(5)$&\\\hline 
\multirow{4}{*}{$\beta_1=0.8$}& 0.014 & 0.026 &0.004 &-0.001 &0.001 &0.000& 0.016 &0.038 &0.002& -0.001 & -0.002& -0.000 & \textit{mean bias} \\ 
&(0.016 )&(0.020 )&(0.030 )&(0.028 )&(0.041 )&(0.008 )&(0.013 )&(0.016 )&(0.027 )&(0.026 )&(0.042 )&(0.005 )&\textit{std}\\ 
& 0.798 & 0.899 &0.058 &0.074 &0.054 &0.056& 0.984 &0.986 &0.057& 0.060 &0.063& 0.038 &\textit{size} \\ \midrule
\multirow{4}{*}{$\beta_2=5$}& 0.019 & -0.004 &0.002 &-0.000 &0.000 &-0.000& 0.015 &-0.002 &0.000& -0.000 &0.000& -0.000 &\textit{mean bias} \\ 
&(0.042 )&(0.040 )&(0.039 )&(0.033 )&(0.039 )&(0.030 )&(0.027 )&(0.024 )&(0.022 )&(0.021 )&(0.022 )&(0.020 )&\textit{std}\\
\end{tabular}} 
\begin{tablenotes}\tiny 
\item CF - control function. $(0)$ - none, $(1)$ - $\lambda_a\hat{a}_i$,  $(2)$ - $\hat{h}(\hat{a}_i)$, $(3)$ - $\hat{h}(a_i)$, $(4)$ - $\hat{h}(\widehat{deg}_i,x_{2i})$, $(5)$ - $h(a_i)$. 
\item The network design parameters are $\mu_0=0.25$, $\mu_1=0.75$, $\alpha_L=-0.67$, $\alpha_H=-0.17$ 
\item Average number of links for $N=100$ is $33.9$, for $N=250$ it is $85.1$. 
\item Average skewness for $N=100$ is $0.20$, for $N=250$ it is $0.25$. 
\item Size is the empirical size of t-test against the truth. 
\item N$=100$, $corr(a_i,\bm{x}_{2i})=0.634$,N$=250$, $corr(a_i,\bm{x}_{2i})=0.632$ 
\item The bias of $\hat{a}_i$ is calculated as $a_i-\hat{a}_i$. 
\item For $N=100$, $\hat{a}_i$ mean bias$=0.006$, median bias$=0.002$, std$=0.248$. 
\item For $N=250$, $\hat{a}_i$ mean bias$=0.003$, median bias$=0.001$, std$=0.154$. 
  \end{tablenotes} 
\end{threeparttable} 
\end{table} 
\begin{table}[!h]\caption{\footnotesize {\bf Design 7 dense network: Parameter values across 1000 Monte Carlo replications with $K_N=8$ and Hermite polynomial sieve}} 
\begin{threeparttable} 
\centering \footnotesize
\scalebox{.8}{\begin{tabular}{|c|c|c|c|c|c|c|c|c|c|c|c|c|c|c|c|}\toprule 
\multicolumn{14}{c}{$h(a_i) = \exp(a_i)$}\\ 
\cellcolor{yellow}$N$&\multicolumn{6}{|c|}{\cellcolor{yellow}$100$}&\multicolumn{6}{|c|}{\cellcolor{yellow}$250$}&\\\hline 
CF&$(0)$&$(1)$&$(2)$&$(3)$&$(4)$&$(5)$& $(0)$ &$(1)$&$(2)$&$(3)$&$(4)$&$(5)$&\\\hline 
\multirow{4}{*}{$\beta_1=0.8$}& 0.020 & 0.033 &0.011 &-0.000 &0.005 &0.000& 0.022 &0.049 &0.005& -0.000 & -0.002& -0.000 & \textit{mean bias} \\ 
&(0.036 )&(0.049 )&(0.045 )&(0.036 )&(0.061 )&(0.011 )&(0.051 )&(0.063 )&(0.039 )&(0.034 )&(0.056 )&(0.007 )&\textit{std}\\ 
& 0.898 & 0.945 &0.061 &0.055 &0.044 &0.065& 0.987 &0.991 &0.047& 0.060 &0.048& 0.044 &\textit{size} \\ \midrule
\multirow{4}{*}{$\beta_2=5$}& 0.039 & 0.004 &0.003 &-0.000 &-0.000 &-0.000& 0.031 &0.006 &0.001& -0.000 &-0.000& -0.000 &\textit{mean bias} \\ 
&(0.057 )&(0.053 )&(0.045 )&(0.034 )&(0.046 )&(0.030 )&(0.038 )&(0.031 )&(0.023 )&(0.021 )&(0.023 )&(0.020 )&\textit{std}\\
\end{tabular}} 
\begin{tablenotes}\tiny 
\item CF - control function. $(0)$ - none, $(1)$ - $\lambda_a\hat{a}_i$,  $(2)$ - $\hat{h}(\hat{a}_i)$, $(3)$ - $\hat{h}(a_i)$, $(4)$ - $\hat{h}(\widehat{deg}_i,x_{2i})$, $(5)$ - $h(a_i)$. 
\item The network design parameters are $\mu_0=0.25$, $\mu_1=0.75$, $\alpha_L=-0.50$, $\alpha_H=0.00$ 
\item Average number of links for $N=100$ is $39.8$, for $N=250$ it is $99.9$. 
\item Average skewness for $N=100$ is $0.16$, for $N=250$ it is $0.21$. 
\item Size is the empirical size of t-test against the truth. 
\item N$=100$, $corr(a_i,\bm{x}_{2i})=0.634$,N$=250$, $corr(a_i,\bm{x}_{2i})=0.632$ 
\item The bias of $\hat{a}_i$ is calculated as $a_i-\hat{a}_i$. 
\item For $N=100$, $\hat{a}_i$ mean bias$=0.002$, median bias$=0.000$, std$=0.241$. 
\item For $N=250$, $\hat{a}_i$ mean bias$=0.001$, median bias$=0.001$, std$=0.150$. 
  \end{tablenotes} 
\end{threeparttable} 
\end{table} 
\begin{table}[!h]\caption{\footnotesize {\bf Design 8 dense network: Parameter values across 1000 Monte Carlo replications with $K_N=8$ and Hermite polynomial sieve}} 
\begin{threeparttable} 
\centering \footnotesize
\scalebox{.8}{\begin{tabular}{|c|c|c|c|c|c|c|c|c|c|c|c|c|c|c|c|}\toprule 
\multicolumn{14}{c}{$h(a_i) = \exp(a_i)$}\\ 
\cellcolor{yellow}$N$&\multicolumn{6}{|c|}{\cellcolor{yellow}$100$}&\multicolumn{6}{|c|}{\cellcolor{yellow}$250$}&\\\hline 
CF&$(0)$&$(1)$&$(2)$&$(3)$&$(4)$&$(5)$& $(0)$ &$(1)$&$(2)$&$(3)$&$(4)$&$(5)$&\\\hline 
\multirow{4}{*}{$\beta_1=0.8$}& 0.006 & 0.011 &0.001 &-0.000 &0.001 &-0.000& 0.009 &0.019 &0.000& -0.001 & -0.001& -0.000 & \textit{mean bias} \\ 
&(0.009 )&(0.013 )&(0.020 )&(0.020 )&(0.028 )&(0.009 )&(0.008 )&(0.011 )&(0.018 )&(0.018 )&(0.029 )&(0.006 )&\textit{std}\\ 
& 0.325 & 0.397 &0.056 &0.070 &0.052 &0.052& 0.759 &0.765 &0.055& 0.055 &0.064& 0.042 &\textit{size} \\ \midrule
\multirow{4}{*}{$\beta_2=5$}& 0.001 & -0.006 &0.000 &-0.000 &0.000 &-0.000& 0.000 &-0.005 &0.000& 0.000 &-0.000& -0.000 &\textit{mean bias} \\ 
&(0.032 )&(0.033 )&(0.034 )&(0.033 )&(0.035 )&(0.030 )&(0.021 )&(0.021 )&(0.021 )&(0.021 )&(0.021 )&(0.020 )&\textit{std}\\
\end{tabular}} 
\begin{tablenotes}\tiny 
\item CF - control function. $(0)$ - none, $(1)$ - $\lambda_a\hat{a}_i$,  $(2)$ - $\hat{h}(\hat{a}_i)$, $(3)$ - $\hat{h}(a_i)$, $(4)$ - $\hat{h}(\widehat{deg}_i,x_{2i})$, $(5)$ - $h(a_i)$. 
\item The network design parameters are $\mu_0=0.25$, $\mu_1=0.75$, $\alpha_L=-0.75$, $\alpha_H=-0.50$ 
\item Average number of links for $N=100$ is $26.9$, for $N=250$ it is $67.6$. 
\item Average skewness for $N=100$ is $0.50$, for $N=250$ it is $0.64$. 
\item Size is the empirical size of t-test against the truth. 
\item N$=100$, $corr(a_i,\bm{x}_{2i})=0.380$,N$=250$, $corr(a_i,\bm{x}_{2i})=0.377$ 
\item The bias of $\hat{a}_i$ is calculated as $a_i-\hat{a}_i$. 
\item For $N=100$, $\hat{a}_i$ mean bias$=0.012$, median bias$=0.004$, std$=0.260$. 
\item For $N=250$, $\hat{a}_i$ mean bias$=0.005$, median bias$=0.003$, std$=0.161$. 
  \end{tablenotes} 
\end{threeparttable} 
\end{table} 
\clearpage
\subsection{Sparse Network, $K_N=4$, Hermite polynomial sieve}\label{section: sparse kn4 hermite}
\begin{table}[!h]\caption{\footnotesize {\bf Design 1 sparse network: Parameter values across 1000 Monte Carlo replications with $K_N=4$ and Hermite polynomial sieve}} 
\begin{threeparttable} 
\centering \footnotesize
\scalebox{.8}{\begin{tabular}{|c|c|c|c|c|c|c|c|c|c|c|c|c|c|}\toprule 
\multicolumn{12}{c}{$h(a_i) = \exp(a_i)$}\\ 
\cellcolor{yellow}$N$&\multicolumn{5}{|c|}{\cellcolor{yellow}$100$}&\multicolumn{5}{|c|}{\cellcolor{yellow}$250$}&\\\hline 
CF&$(0)$&$(1)$&$(2)$&$(3)$&$(4)$& $(0)$ &$(1)$&$(2)$&$(3)$&$(4)$&\\\hline 
\multirow{4}{*}{$\beta_1=0.8$}& 0.000 & 0.000 &0.000 &0.000 &0.000 &0.001& 0.000 &0.000 &-0.000& 0.000 &\textit{mean bias} \\ 
&(0.002 )&(0.002 )&(0.002 )&(0.003 )&(0.002 )&(0.002 )&(0.002 )&(0.002 )&(0.002 )&(0.002 )&\textit{std}\\ 
& 0.056 & 0.054 &0.049 &0.050 &0.042 &0.093& 0.076 &0.058 &0.054& 0.049 &\textit{size} \\ \midrule
\multirow{4}{*}{$\beta_2=5$}& 0.000 & -0.002 &-0.002 &-0.003 &-0.002 &-0.001& -0.001 &0.000 &0.001& 0.000 &\textit{mean bias} \\ 
&(0.032 )&(0.031 )&(0.031 )&(0.039 )&(0.030 )&(0.022 )&(0.022 )&(0.020 )&(0.023 )&(0.020 )&\textit{std}\\
\end{tabular}} 
\begin{tablenotes}\tiny 
\item CF - control function. $(0)$ - none, $(1)$ - $\lambda_a a_i$, $(2)$ - $\hat{h}(a_i)$, $(3)$ - $\hat{h}(\widehat{deg}_i,x_{2i})$, $(4)$ - $h(a_i)$. 
\item The network design parameters are $\mu_0=1.00$, $\mu_1=1.00$, $\alpha_L=-0.50$, $\alpha_H=-0.50$ 
\item Average number of links for $N=100$ is $1.1$, for $N=250$ it is $2.8$. 
\item Average skewness for $N=100$ is $0.98$, for $N=250$ it is $0.74$. 
\item Size is the empirical size of t-test against the truth. 
\item N$=100$, $corr(a_i,\bf{x}_{2i})=-0.001$,N$=250$, $corr(a_i,\bf{x}_{2i})=-0.002$ 
  \end{tablenotes} 
\end{threeparttable} 
\end{table} 
\begin{table}[!h]\caption{\footnotesize {\bf Design 2 sparse network: Parameter values across 1000 Monte Carlo replications with $K_N=4$ and Hermite polynomial sieve}} 
\begin{threeparttable} 
\centering \footnotesize
\scalebox{.8}{\begin{tabular}{|c|c|c|c|c|c|c|c|c|c|c|c|c|c|}\toprule 
\multicolumn{12}{c}{$h(a_i) = \exp(a_i)$}\\ 
\cellcolor{yellow}$N$&\multicolumn{5}{|c|}{\cellcolor{yellow}$100$}&\multicolumn{5}{|c|}{\cellcolor{yellow}$250$}&\\\hline 
CF&$(0)$&$(1)$&$(2)$&$(3)$&$(4)$& $(0)$ &$(1)$&$(2)$&$(3)$&$(4)$&\\\hline 
\multirow{4}{*}{$\beta_1=0.8$}& 0.000 & 0.000 &-0.000 &-0.000 &-0.000 &0.001& 0.001 &-0.000 &-0.000& -0.000 &\textit{mean bias} \\ 
&(0.002 )&(0.002 )&(0.002 )&(0.003 )&(0.002 )&(0.002 )&(0.002 )&(0.002 )&(0.002 )&(0.002 )&\textit{std}\\ 
& 0.060 & 0.059 &0.072 &0.061 &0.072 &0.085& 0.081 &0.051 &0.049& 0.053 &\textit{size} \\ \midrule
\multirow{4}{*}{$\beta_2=5$}& 0.003 & 0.002 &-0.000 &0.001 &-0.000 &-0.002& -0.002 &0.000 &0.000& 0.000 &\textit{mean bias} \\ 
&(0.034 )&(0.034 )&(0.032 )&(0.042 )&(0.031 )&(0.022 )&(0.022 )&(0.021 )&(0.024 )&(0.021 )&\textit{std}\\
\end{tabular}} 
\begin{tablenotes}\tiny 
\item CF - control function. $(0)$ - none, $(1)$ - $\lambda_a a_i$, $(2)$ - $\hat{h}(a_i)$, $(3)$ - $\hat{h}(\widehat{deg}_i,x_{2i})$, $(4)$ - $h(a_i)$. 
\item The network design parameters are $\mu_0=0.25$, $\mu_1=0.75$, $\alpha_L=-0.50$, $\alpha_H=-0.50$ 
\item Average number of links for $N=100$ is $1.1$, for $N=250$ it is $2.8$. 
\item Average skewness for $N=100$ is $1.07$, for $N=250$ it is $0.92$. 
\item Size is the empirical size of t-test against the truth. 
\item N$=100$, $corr(a_i,\bf{x}_{2i})=0.004$,N$=250$, $corr(a_i,\bf{x}_{2i})=0.001$ 
  \end{tablenotes} 
\end{threeparttable} 
\end{table} 
\begin{table}[!h]\caption{\footnotesize {\bf Design 3 sparse network: Parameter values across 1000 Monte Carlo replications with $K_N=4$ and Hermite polynomial sieve}} 
\begin{threeparttable} 
\centering \footnotesize
\scalebox{.8}{\begin{tabular}{|c|c|c|c|c|c|c|c|c|c|c|c|c|c|}\toprule 
\multicolumn{12}{c}{$h(a_i) = \exp(a_i)$}\\ 
\cellcolor{yellow}$N$&\multicolumn{5}{|c|}{\cellcolor{yellow}$100$}&\multicolumn{5}{|c|}{\cellcolor{yellow}$250$}&\\\hline 
CF&$(0)$&$(1)$&$(2)$&$(3)$&$(4)$& $(0)$ &$(1)$&$(2)$&$(3)$&$(4)$&\\\hline 
\multirow{4}{*}{$\beta_1=0.8$}& 0.003 & 0.002 &-0.000 &-0.000 &-0.000 &0.008& 0.008 &0.000 &-0.000& 0.000 &\textit{mean bias} \\ 
&(0.007 )&(0.006 )&(0.003 )&(0.005 )&(0.003 )&(0.009 )&(0.008 )&(0.003 )&(0.006 )&(0.003 )&\textit{std}\\ 
& 0.198 & 0.235 &0.059 &0.053 &0.060 &0.491& 0.576 &0.055 &0.051& 0.049 &\textit{size} \\ \midrule
\multirow{4}{*}{$\beta_2=5$}& -0.008 & -0.003 &-0.001 &-0.000 &-0.001 &-0.022& -0.023 &-0.000 &0.001& 0.000 &\textit{mean bias} \\ 
&(0.055 )&(0.053 )&(0.033 )&(0.052 )&(0.032 )&(0.039 )&(0.037 )&(0.021 )&(0.029 )&(0.020 )&\textit{std}\\
\end{tabular}} 
\begin{tablenotes}\tiny 
\item CF - control function. $(0)$ - none, $(1)$ - $\lambda_a a_i$, $(2)$ - $\hat{h}(a_i)$, $(3)$ - $\hat{h}(\widehat{deg}_i,x_{2i})$, $(4)$ - $h(a_i)$. 
\item The network design parameters are $\mu_0=1.00$, $\mu_1=1.00$, $\alpha_L=0.00$, $\alpha_H=0.00$ 
\item Average number of links for $N=100$ is $2.9$, for $N=250$ it is $7.2$. 
\item Average skewness for $N=100$ is $0.66$, for $N=250$ it is $0.53$. 
\item Size is the empirical size of t-test against the truth. 
\item N$=100$, $corr(a_i,\bf{x}_{2i})=-0.001$,N$=250$, $corr(a_i,\bf{x}_{2i})=-0.002$ 
  \end{tablenotes} 
\end{threeparttable} 
\end{table} 
\begin{table}[!h]\caption{\footnotesize {\bf Design 4 sparse network: Parameter values across 1000 Monte Carlo replications with $K_N=4$ and Hermite polynomial sieve}} 
\begin{threeparttable} 
\centering \footnotesize
\scalebox{.8}{\begin{tabular}{|c|c|c|c|c|c|c|c|c|c|c|c|c|c|}\toprule 
\multicolumn{12}{c}{$h(a_i) = \exp(a_i)$}\\ 
\cellcolor{yellow}$N$&\multicolumn{5}{|c|}{\cellcolor{yellow}$100$}&\multicolumn{5}{|c|}{\cellcolor{yellow}$250$}&\\\hline 
CF&$(0)$&$(1)$&$(2)$&$(3)$&$(4)$& $(0)$ &$(1)$&$(2)$&$(3)$&$(4)$&\\\hline 
\multirow{4}{*}{$\beta_1=0.8$}& 0.001 & 0.001 &0.000 &0.000 &0.000 &0.002& 0.003 &0.000 &-0.000& 0.000 &\textit{mean bias} \\ 
&(0.003 )&(0.003 )&(0.002 )&(0.003 )&(0.002 )&(0.004 )&(0.004 )&(0.002 )&(0.003 )&(0.002 )&\textit{std}\\ 
& 0.089 & 0.090 &0.052 &0.056 &0.049 &0.269& 0.257 &0.072 &0.055& 0.064 &\textit{size} \\ \midrule
\multirow{4}{*}{$\beta_2=5$}& -0.001 & -0.002 &-0.003 &-0.002 &-0.003 &-0.007& -0.008 &0.000 &0.001& 0.001 &\textit{mean bias} \\ 
&(0.039 )&(0.039 )&(0.033 )&(0.041 )&(0.032 )&(0.027 )&(0.027 )&(0.021 )&(0.025 )&(0.021 )&\textit{std}\\
\end{tabular}} 
\begin{tablenotes}\tiny 
\item CF - control function. $(0)$ - none, $(1)$ - $\lambda_a a_i$, $(2)$ - $\hat{h}(a_i)$, $(3)$ - $\hat{h}(\widehat{deg}_i,x_{2i})$, $(4)$ - $h(a_i)$. 
\item The network design parameters are $\mu_0=1.00$, $\mu_1=1.00$, $\alpha_L=-0.25$, $\alpha_H=-0.25$ 
\item Average number of links for $N=100$ is $1.8$, for $N=250$ it is $4.5$. 
\item Average skewness for $N=100$ is $0.81$, for $N=250$ it is $0.62$. 
\item Size is the empirical size of t-test against the truth. 
\item N$=100$, $corr(a_i,\bf{x}_{2i})=-0.001$,N$=250$, $corr(a_i,\bf{x}_{2i})=-0.002$ 
  \end{tablenotes} 
\end{threeparttable} 
\end{table} 
\begin{table}[!h]\caption{\footnotesize {\bf Design 5 sparse network: Parameter values across 1000 Monte Carlo replications with $K_N=4$ and Hermite polynomial sieve}} \label{table: MC sparse}
\begin{threeparttable} 
\centering \footnotesize
\scalebox{.8}{\begin{tabular}{|c|c|c|c|c|c|c|c|c|c|c|c|c|c|}\toprule 
\multicolumn{12}{c}{$h(a_i) = \exp(a_i)$}\\ 
\cellcolor{yellow}$N$&\multicolumn{5}{|c|}{\cellcolor{yellow}$100$}&\multicolumn{5}{|c|}{\cellcolor{yellow}$250$}&\\\hline 
CF&$(0)$&$(1)$&$(2)$&$(3)$&$(4)$& $(0)$ &$(1)$&$(2)$&$(3)$&$(4)$&\\\hline 
\multirow{4}{*}{$\beta_1=0.8$}& 0.007 & 0.008 &-0.000 &-0.000 &-0.000 &0.011& 0.012 &-0.000 &0.000& 0.000 &\textit{mean bias} \\ 
&(0.004 )&(0.004 )&(0.003 )&(0.004 )&(0.002 )&(0.003 )&(0.003 )&(0.002 )&(0.003 )&(0.002 )&\textit{std}\\ 
& 0.398 & 0.530 &0.055 &0.023 &0.051 &0.965& 0.976 &0.062 &0.009& 0.055 &\textit{size} \\ \midrule
\multirow{4}{*}{$\beta_2=5$}& 0.019 & -0.029 &-0.002 &0.003 &0.000 &0.006& -0.045 &-0.001 &-0.001& -0.000 &\textit{mean bias} \\ 
&(0.050 )&(0.056 )&(0.033 )&(0.060 )&(0.030 )&(0.033 )&(0.033 )&(0.021 )&(0.033 )&(0.018 )&\textit{std}\\
\end{tabular}} 
\begin{tablenotes}\tiny 
\item CF - control function. $(0)$ - none, $(1)$ - $\lambda_a a_i$, $(2)$ - $\hat{h}(a_i)$, $(3)$ - $\hat{h}(\widehat{deg}_i,x_{2i})$, $(4)$ - $h(a_i)$. 
\item The network design parameters are $\mu_0=0.25$, $\mu_1=0.75$, $\alpha_L=-0.50$, $\alpha_H=0.00$ 
\item Average number of links for $N=100$ is $2.0$, for $N=250$ it is $5.0$. 
\item Average skewness for $N=100$ is $1.08$, for $N=250$ it is $1.07$. 
\item Size is the empirical size of t-test against the truth. 
\item N$=100$, $corr(a_i,\bf{x}_{2i})=0.634$,N$=250$, $corr(a_i,\bf{x}_{2i})=0.632$ 
  \end{tablenotes} 
\end{threeparttable} 
\end{table} 
\begin{table}[!h]\caption{\footnotesize {\bf Design 6 sparse network: Parameter values across 1000 Monte Carlo replications with $K_N=4$ and Hermite polynomial sieve}} 
\begin{threeparttable} 
\centering \footnotesize
\scalebox{.8}{\begin{tabular}{|c|c|c|c|c|c|c|c|c|c|c|c|c|c|}\toprule 
\multicolumn{12}{c}{$h(a_i) = \exp(a_i)$}\\ 
\cellcolor{yellow}$N$&\multicolumn{5}{|c|}{\cellcolor{yellow}$100$}&\multicolumn{5}{|c|}{\cellcolor{yellow}$250$}&\\\hline 
CF&$(0)$&$(1)$&$(2)$&$(3)$&$(4)$& $(0)$ &$(1)$&$(2)$&$(3)$&$(4)$&\\\hline 
\multirow{4}{*}{$\beta_1=0.8$}& 0.015 & 0.017 &0.000 &0.000 &-0.000 &0.017& 0.020 &0.000 &0.000& 0.000 &\textit{mean bias} \\ 
&(0.005 )&(0.005 )&(0.003 )&(0.005 )&(0.002 )&(0.003 )&(0.003 )&(0.002 )&(0.003 )&(0.001 )&\textit{std}\\ 
& 0.735 & 0.867 &0.064 &0.003 &0.067 &1.000& 1.000 &0.045 &0.001& 0.043 &\textit{size} \\ \midrule
\multirow{4}{*}{$\beta_2=5$}& 0.065 & -0.110 &0.002 &0.008 &-0.000 &0.087& -0.104 &0.003 &-0.000& 0.000 &\textit{mean bias} \\ 
&(0.086 )&(0.100 )&(0.034 )&(0.098 )&(0.028 )&(0.058 )&(0.059 )&(0.021 )&(0.050 )&(0.017 )&\textit{std}\\
\end{tabular}} 
\begin{tablenotes}\tiny 
\item CF - control function. $(0)$ - none, $(1)$ - $\lambda_a a_i$, $(2)$ - $\hat{h}(a_i)$, $(3)$ - $\hat{h}(\widehat{deg}_i,x_{2i})$, $(4)$ - $h(a_i)$. 
\item The network design parameters are $\mu_0=0.25$, $\mu_1=0.75$, $\alpha_L=-0.67$, $\alpha_H=0.25$ 
\item Average number of links for $N=100$ is $2.6$, for $N=250$ it is $6.6$. 
\item Average skewness for $N=100$ is $1.08$, for $N=250$ it is $0.98$. 
\item Size is the empirical size of t-test against the truth. 
\item N$=100$, $corr(a_i,\bf{x}_{2i})=0.832$,N$=250$, $corr(a_i,\bf{x}_{2i})=0.831$ 
  \end{tablenotes} 
\end{threeparttable} 
\end{table} 
\begin{table}[!h]\caption{\footnotesize {\bf Design 7 sparse network: Parameter values across 1000 Monte Carlo replications with $K_N=4$ and Hermite polynomial sieve}} 
\begin{threeparttable} 
\centering \footnotesize
\scalebox{.8}{\begin{tabular}{|c|c|c|c|c|c|c|c|c|c|c|c|c|c|}\toprule 
\multicolumn{12}{c}{$h(a_i) = \exp(a_i)$}\\ 
\cellcolor{yellow}$N$&\multicolumn{5}{|c|}{\cellcolor{yellow}$100$}&\multicolumn{5}{|c|}{\cellcolor{yellow}$250$}&\\\hline 
CF&$(0)$&$(1)$&$(2)$&$(3)$&$(4)$& $(0)$ &$(1)$&$(2)$&$(3)$&$(4)$&\\\hline 
\multirow{4}{*}{$\beta_1=0.8$}& 0.007 & 0.009 &-0.000 &-0.000 &-0.000 &0.008& 0.010 &-0.000 &0.000& 0.000 &\textit{mean bias} \\ 
&(0.003 )&(0.004 )&(0.002 )&(0.004 )&(0.002 )&(0.002 )&(0.003 )&(0.002 )&(0.002 )&(0.001 )&\textit{std}\\ 
& 0.500 & 0.633 &0.052 &0.025 &0.047 &0.976& 0.988 &0.039 &0.008& 0.036 &\textit{size} \\ \midrule
\multirow{4}{*}{$\beta_2=5$}& 0.024 & -0.020 &-0.001 &0.003 &0.000 &0.026& -0.026 &-0.001 &-0.001& 0.000 &\textit{mean bias} \\ 
&(0.048 )&(0.057 )&(0.034 )&(0.059 )&(0.029 )&(0.031 )&(0.035 )&(0.020 )&(0.032 )&(0.017 )&\textit{std}\\
\end{tabular}} 
\begin{tablenotes}\tiny 
\item CF - control function. $(0)$ - none, $(1)$ - $\lambda_a a_i$, $(2)$ - $\hat{h}(a_i)$, $(3)$ - $\hat{h}(\widehat{deg}_i,x_{2i})$, $(4)$ - $h(a_i)$. 
\item The network design parameters are $\mu_0=0.25$, $\mu_1=0.75$, $\alpha_L=-0.75$, $\alpha_H=0.00$ 
\item Average number of links for $N=100$ is $1.8$, for $N=250$ it is $4.4$. 
\item Average skewness for $N=100$ is $1.19$, for $N=250$ it is $1.11$. 
\item Size is the empirical size of t-test against the truth. 
\item N$=100$, $corr(a_i,\bf{x}_{2i})=0.776$,N$=250$, $corr(a_i,\bf{x}_{2i})=0.774$ 
  \end{tablenotes} 
\end{threeparttable} 
\end{table} 
\begin{table}[!h]\caption{\footnotesize {\bf Design 8 sparse network: Parameter values across 1000 Monte Carlo replications with $K_N=4$ and Hermite polynomial sieve}} 
\begin{threeparttable} 
\centering \footnotesize
\scalebox{.8}{\begin{tabular}{|c|c|c|c|c|c|c|c|c|c|c|c|c|c|}\toprule 
\multicolumn{12}{c}{$h(a_i) = \exp(a_i)$}\\ 
\cellcolor{yellow}$N$&\multicolumn{5}{|c|}{\cellcolor{yellow}$100$}&\multicolumn{5}{|c|}{\cellcolor{yellow}$250$}&\\\hline 
CF&$(0)$&$(1)$&$(2)$&$(3)$&$(4)$& $(0)$ &$(1)$&$(2)$&$(3)$&$(4)$&\\\hline 
\multirow{4}{*}{$\beta_1=0.8$}& 0.024 & 0.021 &0.001 &-0.000 &0.000 &0.028& 0.025 &0.001 &-0.000& -0.000 &\textit{mean bias} \\ 
&(0.006 )&(0.005 )&(0.003 )&(0.005 )&(0.002 )&(0.005 )&(0.003 )&(0.003 )&(0.004 )&(0.001 )&\textit{std}\\ 
& 0.973 & 0.954 &0.052 &0.002 &0.044 &1.000& 1.000 &0.064 &0.000& 0.059 &\textit{size} \\ \midrule
\multirow{4}{*}{$\beta_2=5$}& 0.146 & -0.164 &0.007 &0.003 &-0.002 &0.185& -0.127 &0.009 &0.001& 0.000 &\textit{mean bias} \\ 
&(0.110 )&(0.104 )&(0.034 )&(0.117 )&(0.028 )&(0.070 )&(0.056 )&(0.021 )&(0.058 )&(0.016 )&\textit{std}\\
\end{tabular}} 
\begin{tablenotes}\tiny 
\item CF - control function. $(0)$ - none, $(1)$ - $\lambda_a a_i$, $(2)$ - $\hat{h}(a_i)$, $(3)$ - $\hat{h}(\widehat{deg}_i,x_{2i})$, $(4)$ - $h(a_i)$. 
\item The network design parameters are $\mu_0=1.00$, $\mu_1=1.00$, $\alpha_L=-0.50$, $\alpha_H=0.50$ 
\item Average number of links for $N=100$ is $4.0$, for $N=250$ it is $10.0$. 
\item Average skewness for $N=100$ is $0.79$, for $N=250$ it is $0.63$. 
\item Size is the empirical size of t-test against the truth. 
\item N$=100$, $corr(a_i,\bf{x}_{2i})=0.866$,N$=250$, $corr(a_i,\bf{x}_{2i})=0.866$ 
  \end{tablenotes} 
\end{threeparttable} 
\end{table} 
\clearpage
\subsection{Sparse Network, $K_N=8$, Hermite polynomial sieve}\label{section: sparse kn8 hermite}
\input{MC_K_8_d_sparse1_hermite_polynomial.tex}
\input{MC_K_8_d_sparse2_hermite_polynomial.tex}
\input{MC_K_8_d_sparse3_hermite_polynomial.tex}
\input{MC_K_8_d_sparse4_hermite_polynomial.tex}
\input{MC_K_8_d_sparse5_hermite_polynomial.tex}
\input{MC_K_8_d_sparse6_hermite_polynomial.tex}
\input{MC_K_8_d_sparse7_hermite_polynomial.tex}
\input{MC_K_8_d_sparse8_hermite_polynomial.tex}
\clearpage
\subsection{Dense Network, $K_N=4$, polynomial sieve}\label{section: dense kn4 polynomial}
\input{MC_K_4_d1_polynomial.tex}
\input{MC_K_4_d2_polynomial.tex}
\input{MC_K_4_d3_polynomial.tex}
\input{MC_K_4_d4_polynomial.tex}
\input{MC_K_4_d5_polynomial.tex}
\input{MC_K_4_d6_polynomial.tex}
\input{MC_K_4_d7_polynomial.tex}
\input{MC_K_4_d8_polynomial.tex}
\clearpage
\subsection{Dense Network, $K_N=8$, polynomial sieve}\label{section: dense kn8 polynomial}
\input{MC_K_8_d1_polynomial.tex}
\input{MC_K_8_d2_polynomial.tex}
\input{MC_K_8_d3_polynomial.tex}
\input{MC_K_8_d4_polynomial.tex}
\input{MC_K_8_d5_polynomial.tex}
\input{MC_K_8_d6_polynomial.tex}
\input{MC_K_8_d7_polynomial.tex}
\input{MC_K_8_d8_polynomial.tex}
\clearpage

\subsection{Sparse Network, $K_N=4$, polynomial sieve}\label{section: sparse kn4 polynomial}
\input{MC_K_4_d_sparse1_polynomial.tex}
\input{MC_K_4_d_sparse2_polynomial.tex}
\input{MC_K_4_d_sparse3_polynomial.tex}
\input{MC_K_4_d_sparse4_polynomial.tex}
\input{MC_K_4_d_sparse5_polynomial.tex}
\input{MC_K_4_d_sparse6_polynomial.tex}
\input{MC_K_4_d_sparse7_polynomial.tex}
\input{MC_K_4_d_sparse8_polynomial.tex}
\clearpage
\subsection{Sparse Network, $K_N=8$, polynomial sieve}\label{section: sparse kn8 polynomial}
\input{MC_K_8_d_sparse1_polynomial.tex}
\input{MC_K_8_d_sparse2_polynomial.tex}
\input{MC_K_8_d_sparse3_polynomial.tex}
\input{MC_K_8_d_sparse4_polynomial.tex}
\input{MC_K_8_d_sparse5_polynomial.tex}
\input{MC_K_8_d_sparse6_polynomial.tex}
\input{MC_K_8_d_sparse7_polynomial.tex}
\input{MC_K_8_d_sparse8_polynomial.tex}

\end{document}